\documentclass[onecolumn]{IEEEtran}
\usepackage[mathscr]{eucal}
\usepackage[cmex10]{amsmath}
\usepackage{epsfig,epsf,psfrag}
\usepackage{amssymb,amsmath,amsthm,amsfonts,latexsym}
\usepackage{amsmath,graphicx,bm,xcolor,url,overpic}
\usepackage[caption=false]{subfig} 
\usepackage{fixltx2e}
\usepackage{array}
\usepackage{verbatim}
\usepackage{bm}
\usepackage{algorithmic}
\usepackage{algorithm}
\usepackage{verbatim}
\usepackage{textcomp}
\usepackage{mathrsfs}
\usepackage{epstopdf}

\newcommand{\openone}{\leavevmode\hbox{\small1\normalsize\kern-.33em1}}

\catcode`~=11 \def\UrlSpecials{\do\~{\kern -.15em\lower .7ex\hbox{~}\kern .04em}} \catcode`~=13 

\allowdisplaybreaks[4]

\newcommand{\nn}{\nonumber}

\newcommand{\calA}{\mathcal{A}}
\newcommand{\calB}{\mathcal{B}}
\newcommand{\calC}{\mathcal{C}}
\newcommand{\calD}{\mathcal{D}}
\newcommand{\calE}{\mathcal{E}}
\newcommand{\calF}{\mathcal{F}}

\newcommand{\calL}{\mathcal{L}}

\newcommand{\calP}{\mathcal{P}}

\newcommand{\calR}{\mathcal{R}}
\newcommand{\calS}{\mathcal{S}}
\newcommand{\calT}{\mathcal{T}}
\newcommand{\calU}{\mathcal{U}}
\newcommand{\calV}{\mathcal{V}}
\newcommand{\calW}{\mathcal{W}}
\newcommand{\calX}{\mathcal{X}}
\newcommand{\calY}{\mathcal{Y}}
\newcommand{\calZ}{\mathcal{Z}}

\newcommand{\rmd}{\mathrm{d}}

\newcommand{\rmQ}{\mathrm{Q}}

\newcommand{\bbR}{\mathbb{R}}

\DeclareMathAlphabet{\mathbsf}{OT1}{cmss}{bx}{n}
\DeclareMathAlphabet{\mathssf}{OT1}{cmss}{m}{sl}

\newcommand{\rvR}{\mathsf{R}}

\DeclareSymbolFont{bsfletters}{OT1}{cmss}{bx}{n}  
\DeclareSymbolFont{ssfletters}{OT1}{cmss}{m}{n}
\DeclareMathSymbol{\bsfGamma}{0}{bsfletters}{'000}
\DeclareMathSymbol{\ssfGamma}{0}{ssfletters}{'000}
\DeclareMathSymbol{\bsfDelta}{0}{bsfletters}{'001}
\DeclareMathSymbol{\ssfDelta}{0}{ssfletters}{'001}
\DeclareMathSymbol{\bsfTheta}{0}{bsfletters}{'002}
\DeclareMathSymbol{\ssfTheta}{0}{ssfletters}{'002}
\DeclareMathSymbol{\bsfLambda}{0}{bsfletters}{'003}
\DeclareMathSymbol{\ssfLambda}{0}{ssfletters}{'003}
\DeclareMathSymbol{\bsfXi}{0}{bsfletters}{'004}
\DeclareMathSymbol{\ssfXi}{0}{ssfletters}{'004}
\DeclareMathSymbol{\bsfPi}{0}{bsfletters}{'005}
\DeclareMathSymbol{\ssfPi}{0}{ssfletters}{'005}
\DeclareMathSymbol{\bsfSigma}{0}{bsfletters}{'006}
\DeclareMathSymbol{\ssfSigma}{0}{ssfletters}{'006}
\DeclareMathSymbol{\bsfUpsilon}{0}{bsfletters}{'007}
\DeclareMathSymbol{\ssfUpsilon}{0}{ssfletters}{'007}
\DeclareMathSymbol{\bsfPhi}{0}{bsfletters}{'010}
\DeclareMathSymbol{\ssfPhi}{0}{ssfletters}{'010}
\DeclareMathSymbol{\bsfPsi}{0}{bsfletters}{'011}
\DeclareMathSymbol{\ssfPsi}{0}{ssfletters}{'011}
\DeclareMathSymbol{\bsfOmega}{0}{bsfletters}{'012}
\DeclareMathSymbol{\ssfOmega}{0}{ssfletters}{'012}

\newcommand{\hatT}{\hat{T}}

\newcommand{\hatx}{\hat{x}}
\newcommand{\hatX}{\hat{X}}

\newcommand{\haty}{\hat{y}}
\newcommand{\hatY}{\hat{Y}}

\DeclareMathOperator*{\argmin}{arg\,min}

\newtheorem{theorem}{Theorem} 
\newtheorem{lemma}[theorem]{Lemma}

\newtheorem{proposition}[theorem]{Proposition}

\newtheorem{definition}{Definition}

\newtheorem{remark}{Remark}

\usepackage{epstopdf}
\graphicspath{{./figures/}}

\usepackage[ colorlinks = true,
             linkcolor = blue,
             urlcolor  = blue,
             citecolor = red,
             anchorcolor = green,]{hyperref}

\begin{document}

\title{Discrete Lossy Gray-Wyner Revisited: Second-Order Asymptotics, Large and Moderate Deviations}

\author{Lin Zhou  $\qquad$Vincent Y.~F.~Tan$\qquad$
        Mehul Motani
\thanks{The authors are with the  Department of Electrical and Computer Engineering, National University of Singapore (NUS). V.~Y.~F.~Tan is also with the Department of Mathematics, NUS. Emails: lzhou@u.nus.edu; vtan@nus.edu.sg; motani@nus.edu.sg.}
\thanks{Part of this paper has been presented at ISIT 2016, Barcelona, Spain~\cite{zhou2016}.}
}
\maketitle

\begin{abstract}
In this paper, we revisit the discrete lossy Gray-Wyner problem. In particular, we derive its optimal second-order coding rate region, its error exponent (reliability function) and its moderate deviations constant under mild conditions on the source. To obtain the second-order asymptotics, we extend some ideas from Watanabe's work (2015). In particular, we leverage the properties of an appropriate generalization of the conditional distortion-tilted information density, which was first introduced by Kostina and Verd\'u (2012). The converse part uses a perturbation argument by Gu and Effros (2009) in their strong converse proof of the discrete Gray-Wyner problem. The achievability part uses two novel elements: (i) a generalization of various type covering lemmas;  and (ii)  the uniform continuity of the conditional rate-distortion function in both the source (joint) distribution and the  distortion level. To obtain the error exponent,  for the achievability part, we use the same generalized type covering lemma  and for the converse, we use the strong converse together with a change-of-measure technique. Finally, to obtain the moderate deviations constant, we apply the moderate deviations theorem to probabilities defined in terms of information spectrum quantities. 
\end{abstract}

\section{Introduction}
The lossy Gray-Wyner source coding problem \cite{gray1974source} is shown in Figure \ref{systemmodel}. There are three encoders and two decoders. Encoder $f_i$ has access to a source sequence pair $(X^n,Y^n)$ and compresses it into a message $S_i$. Decoder $\phi_1$ aims to recover source sequence $X^n$ under fidelity criterion $d_X$ and distortion level $D_1$ with the encoded message $S_0$ from encoder $f_0$ and $S_1$ from encoder $f_1$. Similarly, the decoder $\phi_2$ aims to recover $Y^n$ with messages $S_0$ and $S_2$. The optimal rate region for lossless and lossy Gray-Wyner source coding problem was characterized in \cite{gray1974source}. However, because an auxiliary random variable is involved in the description of the rate region, it is non-trivial to characterize the second-order coding rate region, the error exponent as well as the moderate deviations constant.

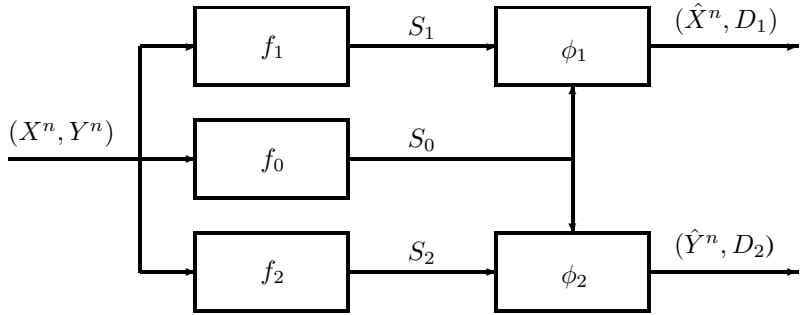
\begin{figure}[t]
\centering
\setlength{\unitlength}{0.5cm}
\scalebox{1}{
\begin{picture}(26,9)
\linethickness{1pt}
\put(1,5.5){\makebox{$(X^n,Y^n)$}}
\put(6,1){\framebox(4,2)}
\put(6,4){\framebox(4,2)}
\put(6,7){\framebox(4,2)}
\put(7.7,1.8){\makebox{$f_2$}}
\put(7.7,4.8){\makebox{$f_0$}}
\put(7.7,7.8){\makebox{$f_1$}}
\put(1,5){\vector(1,0){5}}
\put(4.5,5){\line(0,1){3}}
\put(4.5,8){\vector(1,0){1.5}}
\put(4.5,5){\line(0,-1){3}}
\put(4.5,2){\vector(1,0){1.5}}
\put(14,1){\framebox(4,2)}
\put(14,7){\framebox(4,2)}
\put(15.7,7.7){\makebox{$\phi_1$}}
\put(15.7,1.7){\makebox{$\phi_2$}}
\put(10,2){\vector(1,0){4}}
\put(12,2.5){\makebox(0,0){$S_2$}}
\put(10,5){\line(1,0){6}}
\put(12,5.5){\makebox(0,0){$S_0$}}
\put(16,5){\vector(0,1){2}}
\put(16,5){\vector(0,-1){2}}
\put(10,8){\vector(1,0){4}}
\put(12,8.5){\makebox(0,0){$S_1$}}
\put(18,2){\vector(1,0){4}}
\put(18.7,2.5){\makebox{$(\hat{Y}^n,D_2$)}}
\put(18,8){\vector(1,0){4}}
\put(18.7,8.5){\makebox{$(\hat{X}^n,D_1)$}}
\end{picture}}
\caption{The lossy Gray-Wyner source coding problem~\cite{gray1974source}. We study the second-order asymptotics, the error exponent and the moderate deviations constant for this problem.}
\label{systemmodel}
\end{figure}

\subsection{Related Works}
The most relevant work to ours is \cite{watanabe2015second}, in which Watanabe derived the optimal second-order coding region for the lossless Gray-Wyner source coding problem. Several of the techniques contained herein mirror those in \cite{watanabe2015second}. However, we also combine techniques from other works, develop some new results,   and make several new observations for this lossy problem.  We briefly summarize some other works that are related to Gray-Wyner's seminal work. Gu and Effros derived a strong converse for discrete memoryless sources in \cite{wei2009strong}. Viswanatha, Akyol and Rose \cite{viswanatha2014} derived a single-letter formula for the lossy version of Wyner's common information and also properties of the optimal test channel (which we exploit in our proofs). Xu, Liu and Chen \cite{xu2015} presented an alternative expression for lossy version of Wyner's common information.

There are several works that consider second-order asymptotics for lossy source coding. These include the study of point-to-point lossy source coding by Ingber and Kochman~\cite{ingber2011dispersion} and Kostina and Verd\'u   \cite{kostina2012fixed}, the Wyner-Ziv problem by Watanabe, Kuzuoka and Tan in \cite{watanabe2015} and by Yassaee, Aref and Gohari in \cite{yassaee2013technique}; the successive refinement source coding problem (which is closely related to the Gray-Wyner problem)  by No, Ingber and Weissman in \cite{no2016}.

In terms of error exponent analyses for lossy source coding, there are several related works. For point-to-point lossy source coding, Marton \cite{Marton74} derived the error exponent for discrete memoryless sources while Ihara and Kubo~\cite{ihara2000error} considered Gaussian memoryless sources. For successive refinement source coding, Kanlis and Narayan derived the error exponent in \cite{kanlis1996error} under joint excess-distortion criterion while Tuncel and Rose~\cite{tuncel2003} derived the error exponent under separate excess-distortion ceiteria. 

We also recall the related works on moderate deviations analysis. Chen  {\em et al.}~\cite{chen2007redundancy} and He {\em et al.}~\cite{he2009redundancy} initiated the study of moderate deviations for  fixed-to-variable length source coding with decoder side information. For fixed-to-fixed length analysis, Altu\u{g} and Wagner~\cite{altugwagner2014} initiated the study of moderate deviations in the context of discrete memoryless channels. Polyanksiy and Verd\'u~\cite{polyanskiy2010channel} relaxed some assumptions in the conference version of Altu\u{g} and Wagner's work~\cite{altug2010moderate} and they also considered moderate deviations for AWGN channels. Altu\u{g}, Wagner  and Kontoyiannis~\cite{altug2013lossless} considered moderate deviations for lossless source coding. For lossy source coding, the moderate deviations analysis was done by Tan in \cite{tan2012moderate} using ideas from Euclidean information theory~\cite{borade2008}.

\subsection{Main Contributions}
\label{sec:maincontribution}
In this paper, we derive the optimal second-order coding region, the error exponent and moderate deviations constant for discrete lossy Gray-Wyner source coding problem under some mild conditions.  To the best of our knowledge, even the error exponent for the lossy Gray-Wyner problem has not been established in the literature. We highlight some of the salient features of our analyses. 

\begin{enumerate}
\item As shown in Figure \ref{relationresult}, the achievability proofs for all three asymptotic regimes can be done in a unified manner and all of them hinge on a single covering lemma (Lemma \ref{achievable}) designed specifically for the discrete lossy Gray-Wyner source coding problem.  While the proof of this type covering lemma itself hinges on various other works~\cite{no2016, Marton74,watanabe2015second}, piecing the ingredients together and ensuring that the resultant asymptotic results are tight is non-trivial. 
\item One of the main challenges here in proving the  type covering lemma is the requirement to  establish the uniform continuity of the conditional rate-distortion function in {\em both} the source  distribution and distortion level, which we do in Lemmas \ref{continuityd2}, \ref{continuityp} and \ref{continuityD}. Palaiyanur and Sahai~\cite{palaiyanur2008uniform} only established this uniformity in the source  distribution for the rate-distortion function.

\item Several observations need to be made to establish the optimal second-order coding region.  We define a generalized distortion-tilted information density, leverage on its properties and make proper use of Taylor expansions and the Berry-Esseen Theorem. We encountered a slight obstacle on whether to define the  distortion-tilted information density according to the Gray-Wyner region defined in terms of conditional rate-distortion functions as in \cite{gray1974source} or (conditional) mutual information quantities as in~\cite[Exercise~14.9]{el2011network}. These are equivalent as stated in Theorem \ref{gwregion} and equation \eqref{gwregion2}. However, it turns out that the latter is more amenable since it does not explicitly involve an optimization (which is present in the characterization of the conditional rate-distortion function).  In the converse part, as shown in Figure \ref{relationresult}, we prove a type-based strong converse by using perturbation approach in \cite{wei2009strong} and similar analysis in \cite{watanabe2015second}. 

\item To evaluate the optimal second-order coding region for rate triplets on the Pangloss plane \cite{gray1974source}, we leverage a result in Viswanatha, Akyol, and  Rose \cite{viswanatha2014} which establishes several Markov chains for the Gray-Wyner problem. This helps to simplify the relevant tilted information densities.

\item For the error exponent analysis, we combine our type covering lemma with Marton's technique for establishing the reliability function of lossy data compression~\cite{Marton74}. The converse part follows from strong converse \cite{wei2009strong} and  the change-of-measure technique by Haroutunian~\cite{csiszar2011information, haroutunian68}.

\item Finally, for our moderate deviations analysis, we use an information spectrum calculation~\cite{polyanskiy2010channel} similar to that used for the second-order asymptotic analysis. We also invoke the moderate deviations principle/theorem in~\cite[Theorem~3.7.1]{dembo2009large}. Further, in the analysis of moderate deviations, compared with previous result for lossy source coding~\cite{tan2012moderate}, we removed the additional requirement that $\lim_{n\to\infty} \frac{n\rho_n^2}{\log n}\to \infty$ where $\rho_n$ controls the speed of convergence of rate to a boundary rate (pair) in the moderate deviations regime. Instead, all we need is the usual condition that $\lim_{n\to \infty} n\rho_n^2=\infty$. 
\end{enumerate}

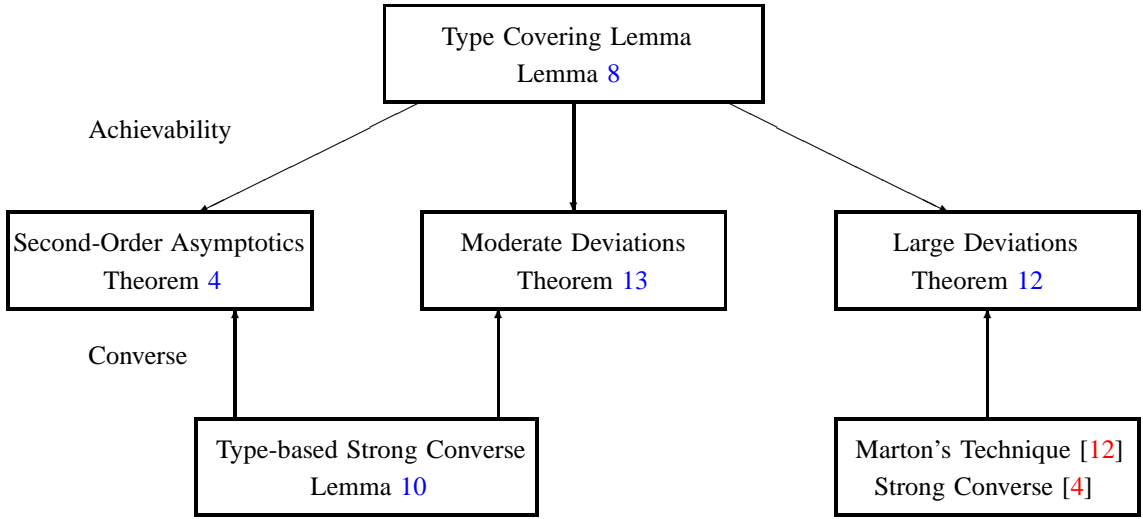
\begin{figure}[t]
\centering
\setlength{\unitlength}{0.5cm}
\scalebox{1}{
\begin{picture}(32,14)
\linethickness{1pt}
\put(11,11.5){\framebox(10,2.5)}
\put(12.5,13){\makebox{Type Covering Lemma}}
\put(14.5,12){\makebox{Lemma \ref{achievable}}}
\put(3.1,10.5){\makebox{Achievability}}
\put(3.1,4.5){\makebox{Converse}}
\put(1,6){\framebox(8,2.5)}
\put(1.1,7.5){\makebox{Second-Order Asymptotics}}
\put(3.5,6.5){\makebox{Theorem \ref{mainresult}}}
\put(12,6){\framebox(8,2.5)}
\put(13,7.5){\makebox{Moderate Deviations}}
\put(14.5,6.5){\makebox{Theorem \ref{theoremmdc}}}
\put(23,6){\framebox(8,2.5)}
\put(24.5,7.5){\makebox{Large Deviations}}
\put(25,6.5){\makebox{Theorem \ref{eegray}}}
\put(6,0.5){\framebox(9,2.5)}
\put(6.5,2){\makebox{Type-based Strong Converse}}
\put(9,1){\makebox{Lemma \ref{typestrongconverse}}}
\put(23,0.5){\framebox(8,2.5)}
\put(23.5,2){\makebox{Marton's Technique \cite{Marton74}}}
\put(24,1){\makebox{Strong Converse~\cite{wei2009strong}}}
\put(7,3){\vector(0,1){3}}
\put(14,3){\vector(0,1){3}}
\put(27,3){\vector(0,1){3}}
\put(16,11.5){\vector(0,-1){3}}
\put(12,11.5){\vector(-2,-1){6}}
\put(20,11.5){\vector(2,-1){6}}
\end{picture}}
\caption{Main results and proof techniques}
\label{relationresult}
\end{figure}

\subsection{Organization of the Paper}
The rest of the paper is organized as follows. In Section \ref{existingresult}, we set up the notation, formulate the discrete lossy Gray-Wyner problem and recapitulate the optimal rate region (first-order result). In Section \ref{secondorder}, we define the second-order coding region formally and present the main theorem which expresses the optimal second-order coding region in terms of a rate-dispersion function~\cite{kostina2012fixed}. In addition, we simplify the calculation of the region for rate triplets on the Pangloss region and provide an numerical example for a doubly symmetric binary source with hamming distortion measures. In Section \ref{secondorderproof}, we present the proof for second-order asymptotics.  For the achievability part, we present a type covering lemma for discrete lossy Gray-Wyner problem which is used extensively in various achievability proofs throughout the paper. In Section \ref{largedeviations}, we define the error exponent formally, present the result and provide a detailed proof. In Section \ref{moderatedeviations}, we provide a formal definition of moderate deviations constant, present the main result on moderate deviations as well as its detailed proof. Finally, we conclude the paper in Section \ref{conclusion}. To ensure that the main ideas  of the paper are presented seamlessly, we defer the proof of all supporting technical lemmas to the appendices.

\section{Problem Formulation and Existing Results}
\label{existingresult}
\subsection{Notation}
\label{sec:notation}
Random variables and their realizations are in capital (e.g.,\ $X$) and lower case (e.g.,\ $x$) respectively. All sets (e.g., alphabets of random variables) are denoted in calligraphic font (e.g.,\ $\mathcal{X}$). Let $X^n:=(X_1,\ldots,X_n)$ be a random vector of length $n$. The set of all probability distribution on $\calX$ is denoted as $\calP(\calX)$ and the set of all conditional probability distribution from $\calX$ to $\calY$ is denoted as $\calP(\calY|\calX)$. Given $P\in\calP(\calX)$ and $V\in\calP(\calY|\calX)$, we use $P\times V$ to denote the joint distribution induced by $P$ and $V$. In terms of the method of types, we use the notations as \cite{TanBook}. Given sequence $x^n$, the empirical distribution is denoted as $\hat{T}_{x^n}$. The set of types formed from length $n$ sequences in $\calX$ is denoted as $\calP_{n}(\calX)$. Given $P\in\calP_{n}(\calX)$, the set of all sequences of length $n$ with type $P$ is denoted as $\calT_{P}$. Given $x^n\in\calT_{P}$, the set of all sequences $y^n\in\calY^n$ such that the joint type of $(x^n,y^n)$ is $P\times V$ is denoted as $\calT_{V}(x^n)$. The set of all $V\in\calP(\calY|\calX)$ for which $\calT_{V}(x^n)$ is not empty for $x^n\in\calT_{P}$ is denoted as $\calV_{n}(\calY;P)$.

In terms of information theoretic quantities, we use $H(X)$ and $H(P_X)$ interchangeably to denote the entropy of a random variable $X$ with distribution $P_{X}$. Similarly, we use $H(X|Y)$ and $H(P_{X|Y}|P_Y)$ interchangeably. For mutual information, we use $I(X;Y)$ and $I(P_{X},P_{Y|X})$ interchangeably. For conditional mutual information, we use $I(X;Y|W)$ and $I(P_{X|W},P_{Y|XW}|P_{W})$ interchangeably.

We use $\exp(x)$ to denote $2^x$. We let $\rmQ(t) := \int_{t}^\infty\frac{1}{\sqrt{2\pi}}e^{-u^2/2}\, \rmd u$ be the  complementary cumulative distribution function   of the standard Gaussian. We let $\mathrm{Q}^{-1}$ be the inverse of $\rmQ$.  Given two integers $a$ and $b$, we use $[a:b]$ to denote all the integers between $a$ and $b$. We use standard asymptotic notation such as $O(\cdot)$  and $o(\cdot)$.

Given a joint probability mass function (pmf) $P_{XY}$, let $\calS=\mathrm{supp}(P_{XY})$ and $|\calS|=m$. Let us sort $P_{XY}(x,y)$ in an decreasing order for all $(x,y)\in\calX\times\calY$, and for all $i\in[1:m]$, let $(x_i,y_i)$ be the pair such that $P_{XY}(x_i,y_i)$ is the $i$-th largest. Let $\Gamma(P_{XY})$ be a joint distribution defined on $\calS$ such that $\Gamma_i(P_{XY})=P_{XY}(x_i,y_i)$ for all $i\in[1:m]$.

\subsection{Problem Formulation}
We consider a correlated source $(X,Y)$ with joint distribution $P_{XY}$ and a finite alphabet $\calX\times\calY$. The correlated source is assumed to be stationary and memoryless, hence $(X^n,Y^n)$ is an i.i.d.\ sequence where each $(X_i,Y_i)$ is generated according to $P_{XY}$. The basic definitions are as follows.

\begin{definition}
An $(n,M_0,M_1,M_2)$-code for lossy Gray-Wyner source coding consists of three encoders:
\begin{align}
f_0:\calX^n\times\calY^n\to\{1,2,\ldots,M_0\},\\
f_1:\calX^n\times\calY^n\to\{1,2,\ldots,M_1\},\\
f_2:\calX^n\times\calY^n\to\{1,2,\ldots,M_2\},
\end{align}
and two decoders:
\begin{align}
\phi_1:\{1,2,\ldots,M_0\}\times\{1,2,\ldots,M_1\}\to\hat{\calX}^n,\\
\phi_2:\{1,2,\ldots,M_0\}\times\{1,2,\ldots,M_2\}\to\hat{\calY}^n.
\end{align}
\end{definition}
Define two distortion measures: $d_X:\calX\times\hat{\calX}\to[0,\infty)$ and $d_Y:\calY\times\hat{\calY}\to[0,\infty)$ such that for each $(x,y)\in\calX\times\calY$, there exists $(\hat{x},\hat{y})\in\hat{\calX}\times\hat{\calY}$ satisfying $d_{X}(x,\hat{x})=0$ and $d_{Y}(y,\hat{y})=0$. Define $\overline{d}_X:=\max_{x,\hat{x}}d_{X}(x,\hat{x})$ and $\underline{d}_X:=\min_{x,\hat{x}:d_{X}(x,\hat{x})>0}d_{X}(x,\hat{x})$. Similarly, we define $\overline{d}_Y$ and $\underline{d}_Y$. Let the average distortion between $x^n$ and $\hat{x}^n$ be defined as $d_X(x^n,\hat{x}^n):=\frac{1}{n}\sum_{i=1}^nd_X(x_i,\hat{x}_i)$ and the average distortion $d_Y(y^n,\hat{y}^n)$ be defined in a similar manner. Throughout the paper, we consider the case where $D_1>0$ and $D_2>0$ (but we will remark on how our results apply to the case where either or both $D_i=0$ ($i=1,2$)). The first-order fundamental limit is defined as follows.
\begin{definition}[First-order Region]
\label{deffirst}
A rate triplet $(R_0,R_1,R_2)$ is said to be $(D_1,D_2)$-achievable if there exists a sequence of $(n,M_0,M_1,M_2)$-codes such that
\begin{align}
\limsup_{n\to\infty}\frac{1}{n}\log M_0\leq R_0,\\*
\limsup_{n\to\infty}\frac{1}{n}\log M_1\leq R_1,\\*
\limsup_{n\to\infty}\frac{1}{n}\log M_2\leq R_2,
\end{align}
and
\begin{align}
\limsup_{n\to\infty} \mathbb{E}\left[d_X(X^n,\hat{X}^n)\right]\leq D_1,\\
\limsup_{n\to\infty} \mathbb{E}\left[d_Y(Y^n,\hat{Y}^n)\right]\leq D_2.
\end{align}
The closure of the set of all $(D_1,D_2)$-achievable rate triplets is the $(D_1,D_2)$-optimal rate region and denoted as $\calR(D_1,D_2|P_{XY})$.
\end{definition}

\subsection{Existing Results}
Gray and Wyner characterized the $(D_1,D_2)$-achievable rate region in \cite{gray1974source}. Let $\calP(P_{XY})$ be the set of all joint distributions $P_{XYW}$ such that the $\calX\times\calY$-marginal of $P_{XYW}$ is the source distribution $P_{XY}$ and $|\calW|\leq |\calX|\cdot|\calY|+2$. Denote the $\calX\times\calW$ marginal distribution as $P_{XW}$ and the $\calY\times\calW$ marginal distribution as $P_{YW}$.
\begin{theorem}[Gray-Wyner \cite{gray1974source}]
\label{gwregion}
The $(D_1,D_2)$-achievable rate region for lossy Gray-Wyner source coding is
\begin{align}
\!\!\!\calR(D_1,D_2|P_{XY})=\!\!\!\bigcup_{P_{XYW}\in\calP(P_{XY})}\!\!\!\left\{(R_0,R_1,R_2):R_0\geq I(X,Y;W), R_1\geq R_{X|W}(P_{XW},D_1), R_2\geq R_{Y|W}(P_{YW},D_2)\right\},
\end{align}
where $R_{X|W}(P_{XW},D_1)$ and $R_{Y|W}(P_{YW},D_2)$ are conditional rate-distortion functions \cite[pp.~275, Chapter 11]{el2011network}, i.e.,
\begin{align}
\!\!R_{X|W}(P_{XW},D_1)&=\min_{P_{\hat{X}|XW}:\mathbb{E}[d_{X}(X,\hat{X})]\leq D_1} \!\!I(X;\hat{X}|W),
\end{align}
and similarly for $R_{Y|W}(P_{YW},D_2)$.
\end{theorem}
An equivalent version of the first-order coding region for Gray-Wyner problem was given in \cite[Exercise 14.9]{el2011network} and states that
\begin{align}
\!\!\calR(D_1,D_2|P_{XY})=\!\!\!\!\!\!\!\!\bigcup_{\substack{P_{W|XW}P_{\hatX_1|XW}P_{\hatY|YW}:\\\mathbb{E}[d_X(X,\hatX)]\leq D_1,~\mathbb{E}[d_2(Y,\hatY)]\leq D_2}}\!\!\!\!\!\!\!\!\!\!\left\{(R_0,R_1,R_2):R_0\geq I(X,Y;W), R_1\geq I(X;\hatX|W), R_2\geq I(Y;\hatY|W)\right\}, \label{gwregion2}
\end{align}

\section{Second-order Asymptotics}
\label{secondorder}
\subsection{Definition of Second-Order Coding Region}
In this subsection, we define the second-order coding region for lossy Gray-Wyner problem. First, define the excess-distortion probability for distortion pair $(D_1,D_2)$ as 
\begin{align}
\label{defexcessprob}
\epsilon_n(D_1,D_2):=\Pr\left(d_X(X^n,\hat{X}^n)>D_1~\mathrm{or}~d_Y(Y^n,\hat{Y}^n)> D_2\right).
\end{align}
\begin{definition}[Second-Order Region]
\label{defsecond}
A triplet $(L_0,L_1,L_2)$ is said to be second-order $(R_0,R_1,R_2,D_1,D_2,\epsilon)$-achievable if there exists a sequence of $(n,M_0,M_1,M_2)$-codes such that
\begin{align}
\limsup_{n\to\infty}\frac{1}{\sqrt{n}}\left(\log M_0-nR_0\right)\leq L_0,\\
\limsup_{n\to\infty}\frac{1}{\sqrt{n}}\left(\log M_1-nR_1\right)\leq L_1,\\
\limsup_{n\to\infty}\frac{1}{\sqrt{n}}\left(\log M_2-nR_2\right)\leq L_2,
\end{align}
and
\begin{align}
\limsup_{n\to\infty}\epsilon_n(D_1,D_2)\leq \epsilon.
\end{align}
The closure of the set of all second-order $(R_0,R_1,R_2,D_1,D_2,\epsilon)$-achievable triplets is called the optimal second-order $(R_0,R_1,R_2,D_1,D_2,\epsilon)$ coding region and denoted as $\calL(R_0,R_1,R_2,D_1,D_2,\epsilon)$.
\end{definition}
The central goal for this section is to characterize $\calL(R_0,R_1,R_2,D_1,D_2,\epsilon)$. Note that in Definition~\ref{deffirst}, the expected distortion measure is considered whereas in Definition~\ref{defsecond}, the excess-distortion probability is considered. For the purposes of second-order asymptotics, error exponents and moderate deviations, the formulation in Definition~\ref{defsecond} is preferred since there is a probability to quantify.

\subsection{Tilted Information Density}

We now introduce the tilted information density which takes on a similar role as it did in the lossless case \cite{watanabe2015second}. Given distortion pair $(D_1,D_2)$ and rate pair $(R_1,R_2)$, let
\begin{align}
\rvR_0(R_1,R_2,D_1,D_2|P_{XY})
&:=\min\{R_0:(R_0,R_1,R_2)\in\calR(D_1,D_2|P_{XY})\}\label{minkey0}\\
&=\min_{P_{XYW}\in\calP(P_{XY})}\{I(X,Y;W):R_1\geq R_{X|W}(P_{XW},D_1),~R_2\geq R_{Y|W}(P_{YW},D_2)\}\label{minkey}\\
&=\min_{\substack{P_{W|XW}P_{\hatX_1|XW}P_{\hatY|YW}:\\\mathbb{E}[d_X(X,\hatX)]\leq D_1,~\mathbb{E}[d_2(Y,\hatY)]\leq D_2\\ I(X;\hatX|W)\leq R_1,~I(Y;\hatY|W)\leq R_2}}\!\!\!\!\!\!\!\!\!\!I(X,Y;W)\label{minkey2},
\end{align}
where \eqref{minkey} follows from Theorem \ref{gwregion} and \eqref{minkey2} follows from \eqref{gwregion2}.

Since $\calR(D_1,D_2|P_{XY})$ is a convex set \cite{gray1974source}, the minimization in \eqref{minkey} is attained when $R_1=R_{X|W}(P_{XW},D_1)$ and $R_2=R_{Y|W}(P_{YW},D_2)$ for some optimal test channel $P_{W|XY}$ unless $\rvR_0(R_1,R_2,D_1,D_2|P_{XY})=0$ or $\rvR_0(R_1,R_2,D_1,D_2|P_{XY})=\infty$. However, in the following, we assume $\rvR_0(R_1,R_2,D_1,D_2|P_{XY})>0$ is finite. Given distortion levels $(D_1,D_2)$, we are interested in the rate triplets $(R_0^*,R_1^*,R_2^*)$ such that $R_0^*=\rvR_0(R_0^*,R_1^*,D_1,D_2)$ throughout the section. Further, as in \cite{watanabe2015second}, we assume $\calR(D_1,D_2|P_{XY})$ is smooth at a rate triplet $(R_0^*,R_1^*,R_2^*)$ of our interest, i.e.,
\begin{align}
\lambda_i^*:=-\frac{\partial \rvR_0(R_1,R_2,D_1,D_2|P_{XY})}{\partial R_i}\bigg|_{(R_1,R_2)=(R_1^*,R_2^*)},\label{definelambdai}
\\
\gamma_i^*:=-\frac{\partial \rvR_0(R_1,R_2,D_1',D_2'|P_{XY})}{\partial D_i'}\bigg|_{(D_1',D_2')=(D_1,D_2)}\label{definegammai},
\end{align}
are well-defined for $i\in\{1,2\}$. Note that $\lambda_i^*,~\gamma_i^*\geq 0$ since $\rvR_0(R_1,R_2,D_1,D_2|P_{XY})$ is a non-increasing in $(R_1,R_2,D_1,D_2).$ Throughout the paper, we assume $(\lambda_1^*,\lambda_2^*,\gamma_1^*,\gamma_2^*)$ are strictly positive, i.e., we consider a rate triplet where $R_0^*=\rvR_0(R_1^*,R_2^*,D_1,D_2|P_{XY})$ is positive and finite.

Let $P_{W|XY}^*P_{\hatX|XW}^*P_{\hatY|YW}^*$ be the optimal test channel\footnote{The following tilted information density is still well-defined even if the optimal test channel is not unique due to similar arguments as \cite[Lemma 2]{watanabe2015second}} that achieves the $\rvR_0(R_1^*,R_2^*,D_1,D_2|P_{XY})$ in \eqref{minkey2}. Let $P_{\hatX|W}^*, P_{\hatY|W}^*, P_W^*$ be the induced (conditional) distributions. Define
\begin{align}
&\jmath(x,D_1|w)
:=\log \frac{1}{\sum_{\hatx} P_{\hatX|W}^*(\hatx|w)\exp\Big(\frac{\gamma_1^*}{\lambda_1^*}(D_1-d_X(x,\hatx))\Big)},\label{def:j1x}\\
&\jmath(y,D_2|w)
:=\log \frac{1}{\sum_{\haty} P_{\hatY|W}^*(\haty|w)\exp\Big(\frac{\gamma_2^*}{\lambda_2^*}(D_2-d_Y(y,\haty))\Big)}\label{def:j2y}.
\end{align}
\begin{definition}
For a rate triplet $(R_0^*,R_1^*,R_2^*)$, given distortion threshold pair $(D_1,D_2)$, the tilted information density for lossy Gray-Wyner source coding is defined as

\begin{align}
\jmath(x,y|R_1^*,R_2^*,D_1,D_2):=\log \frac{1}{\sum_{w}P_W^*(w)\exp\Big(\lambda_1^*(R_1^*-\jmath(x,D_1|w))+\lambda_2^*(R_2^*-\jmath(y,D_2|w))\Big)}\label{def:gwtilt}.
\end{align}

\end{definition}

We remark that there are two equivalent characterizations of the Gray-Wyner region, one defined in terms of conditional rate-distortion functions in Theorem \ref{gwregion} and the other defined solely in terms of (conditional) mutual information quantities in \eqref{gwregion2}. For the lossless Gray-Wyner problem~\cite{watanabe2015second}, the two regions are exactly the same. The tilted information densities derived based on these two regions are subtly different. We find that the tilted information density derived from the second region in \eqref{gwregion2} is more amenable to subsequent second-order analyses on the Pangloss plane (Lemma \ref{panglosstilted}). Thus the ``correct'' non-asymptotic fundamental quantity for the lossy Gray-Wyner problem is the tilted information density we identified based on the second Gray-Wyner region in \eqref{def:gwtilt}.

Next, we show that the tilted information density for lossy Gray-Wyner source coding has properties similarly like \cite[Properties 1-3]{kostina2012converse} and \cite[Lemma 1]{watanabe2015second}.
\begin{lemma}
\label{propertytilted}
The tilted information density $\jmath_{XY}(x,y|R_1^*,R_2^*,D_1,D_2,P_{XY})$ has the following properties:
\begin{align}
\rvR_0(R_1^*,R_2^*,D_1,D_2|P_{XY})&=\mathbb{E}_{P_{XY}}\left[\jmath_{XY}(X,Y|R_1^*,R_2^*,D_1,D_2,P_{XY})\right],
\end{align}
and for $(w,\hatx,\haty)$ such that $P_W^*(w)P_{\hatX|W}^*(\hatx|w)P_{\hatY|W}^*(\haty|w)>0$,
\begin{align}
\nn&\jmath_{XY}(x,y|R_1^*,R_2^*,D_1,D_2,P_{XY})\\
\nn&=\log \frac{P_{W|XY}^*(w|xy)}{P_W^*(w)}+\lambda_1^*\log\frac{P_{\hatX|XW}^*(\hatx|x,w)}{P_{\hatX|W}^*(\hatx|w)}-\lambda_1^*R_1^*+\lambda_2^*\log\frac{P_{\hatY|YW}^*(\haty|y,w)}{P_{\hatY|W}^*(\haty|w)}-\lambda_2^*R_2^*\\
&\qquad+\gamma_1^*(d_X(x,\hatx)-D_1)+\gamma_2^*(d_Y(y,\haty)-D_2).
\end{align}  
\end{lemma}
The proof of Lemma \ref{propertytilted} is similar to \cite[Lemma 1]{watanabe2015second}, \cite[Lemma 1]{kontoyiannis2000pointwise} and is provided in Appendix \ref{prooflemmatilted}.

In the following lemma, we relate the derivative of the minimum common rate function with the tilted information density where notation $\Gamma$ is defined in Section \ref{sec:notation} (See also \cite{watanabe2015second}). For any $Q_{XY}$, let $Q_{W|XY}^*Q_{\hatX|XW}^*Q_{\hatY|YW}^*$ be the optimal test channel for $\rvR_0(R_1^*,R_2^*,D_1,D_2|\Gamma(Q_{XY}))$ (see \eqref{minkey2}). Let $Q_W^*,Q_{\hatX|W}^*,Q_{\hatY|W}^*$ be the corresponding induced distributions.
\begin{lemma} 
\label{linkrjxy}
Suppose that for all $Q_{XY}$ in some neighborhood of $P_{XY}$, $\mathrm{supp}(Q_W^*)\subset\mathrm{supp}(P_W^*)$, $\mathrm{supp}(Q_{\hatX|W}^*)\subset\mathrm{supp}(P_{\hatX|W}^*)$ and $\mathrm{supp}(Q_{\hatY|W}^*)\subset\mathrm{supp}(P_{\hatX|W}^*)$. Then for $i\in[1:m-1]$,
\begin{align}
\nn&\frac{\partial \rvR_0(R_1^*,R_2^*,D_1,D_2|\Gamma(Q_{XY}))}{\partial \Gamma_i(Q_{XY})}\bigg|_{Q_{XY}=P_{XY}}\\
&=\jmath_{XY}(i|R_1^*,R_2^*,D_1,D_2,\Gamma(P_{XY}))-\jmath_{XY}(m|R_1^*,R_2^*,D_1,D_2,\Gamma(P_{XY})), \label{eqn:derivate_rd}
\end{align}
where $\jmath_{XY}(i|R_1^*,R_2^*,D_1,D_2,\Gamma(P_{XY}))$ is short for $\jmath_{XY}(x_i,y_i|R_1^*,R_2^*,D_1,D_2,\Gamma(P_{XY}))$ where $P_{XY}(x_i,y_i)=\Gamma_i(P_{XY})$, and similarly for $\jmath_{XY}(m|R_1^*,R_2^*,D_1,D_2,\Gamma(P_{XY}))$.

\end{lemma}
The proof of Lemma \ref{linkrjxy} is similar to the proof in \cite[Lemma 3]{watanabe2015second} and~\cite[Theorem 2.2]{kostina2013lossy} and provided in Appendix \ref{prooflinkrjxy}. In particular, we need to re-parametrize probability distributions on the simplex as in \cite[Lemma 3]{watanabe2015second}.

\subsection{Main Result}
Given a particular rate triplet $(R_0^*,R_1^*,R_2^*)\in\calR(D_1,D_2|P_{XY})$, we impose the following conditions:
\begin{enumerate}
\item \label{cond1} $R_0^*=\rvR_0(R_1^*,R_2^*,D_1,D_2|P_{XY})$ is positive and finite;
\item For $i=1,2$, $\lambda_i$ in \eqref{definelambdai} and $\gamma_i^*$ in \eqref{definegammai} are well-defined and positive;
\item \label{cond2} $(R_1, R_2, Q_{XY})\mapsto \rvR_0(R_1,R_2,D_1,D_2|Q_{XY})$ is twice differentiable in the neighborhood of $(R_1^*,R_2^*,P_{XY})$ and the derivative  is bounded (i.e., the spectral norm of the Hessian matrix is bounded).
\end{enumerate}
Let the {\em rate-dispersion  function} \cite{kostina2012fixed} be 
\begin{align}
\mathrm{V}(R_1^*,R_2^*,D_1,D_2|P_{XY}):=\mathrm{Var}\left[\jmath_{XY}(X,Y|R_1^*,R_2^*,D_1,D_2,P_{XY})\right].
\end{align}

\begin{theorem}
\label{mainresult}
Under conditions (\ref{cond1}) to (\ref{cond2}), the optimal second-order $(R_0^*,R_1^*,R_2^*,D_1,D_2,\epsilon)$ coding region is
\begin{align}
\calL(R_0^*,R_1^*,R_2^*,D_1,D_2,\epsilon)
=\left\{(L_0,L_1,L_2):L_0+\lambda_1^*L_1+\lambda_2^*L_2\geq \sqrt{\mathrm{V}(R_1^*,R_2^*,D_1,D_2|P_{XY})}\mathrm{Q}^{-1}(\epsilon)\right\}.
\end{align}
\end{theorem}
We observe that the rate-dispersion function $\mathrm{V}(R_1^*,R_2^*,D_1,D_2|P_{XY})$ is a fundamental quantity that governs the speed of convergence of the rates of optimal code to the rate triplet $(R_0^*,R_1^*,R_2^*)$. Theorem \ref{mainresult} is proved in Section \ref{secondorderproof}.

\begin{remark}
\label{remarkmain}
To obtain the corresponding results for $D_1=0$ or $D_2=0$, we need to define the conditional $D_i$-tilted information densities (cf. \eqref{def:j1x} and \eqref{def:j2y}) when $D_i=0$ (cf. \cite[Remark 1]{kostina2012fixed}). Define $\jmath(x,D_1|w):=-\log {P_{X|W}^*(x|w)}$ when $D_1=0$. Similarly, define $\jmath(y,D_2|w):=-\log{P_{Y|W}^*(y|w)}$ when $D_2=0$. Combining the techniques used in this paper and the lossless case in \cite{watanabe2015second}, it is not hard to verify that Theorem \ref{mainresult} is still valid when $D_1=0$ and/or $D_2=0$.
\end{remark}

\subsection{On the Pangloss Plane for the Lossy Gray-Wyner Problem} \label{sec:pang}
In general, it is not easy to calculate $\calL(R_0^*,R_1^*,R_2^*,D_1,D_2,\epsilon)$. Here we consider calculating $\calL(R_0,R_1,R_2,D_1,D_2,\epsilon)$ for a rate triplet $(R_0^*,R_1^*,R_2^*)$ on the Pangloss plane \cite{gray1974source}. It is shown in Theorem 6 in \cite{gray1974source} that $(R_0,R_1,R_2)$ is $(D_1,D_2)$-achievable if
\begin{align}
R_0+R_1+R_2&\geq R_{XY}(P_{XY},D_1,D_2),\label{panglossbd}\\
R_0+R_1&\geq R_{X}(P_X,D_1),\\
R_0+R_2&\geq R_{Y}(P_Y,D_2),
\end{align}
where $R_{XY}(P_{XY},D_1,D_2)$ is joint rate-distortion function and $R_{X}(P_X,D_1),R_{Y}(P_Y,D_2)$ are rate-distortion functions~\cite{el2011network}, i.e.,
\begin{align}
R_{XY}(P_{XY},D_1,D_2):=\min_{P_{\hat{X}\hat{Y}|XY}:\mathbb{E}[d_{X}(X,\hat{X})]\leq D_1,~\mathbb{E}[d_{Y}(Y,\hat{Y})]\leq D_2} I(XY;\hat{X}\hat{Y}).
\end{align}
The condition in \eqref{panglossbd} is called the Pangloss bound since the optimal performance is obtained when the receivers cooperate. The set of $D_1,D_2$-achievable rate triplets $(R_0,R_1,R_2)$ satisfying $R_0+R_1+R_2=R_{XY}(P_{XY},D_1,D_2)$ is called the Pangloss plane, denoted as $\calR_{\mathrm{pg}}(D_1,D_2|P_{XY})$, i.e.,
\begin{align}
\calR_{\mathrm{pg}}(D_1,D_2|P_{XY})
:=\left\{(R_0,R_1,R_2):(R_0,R_1,R_2)\in\calR(D_1,D_2|P_{XY}),~R_0+R_1+R_2=R_{XY}(P_{XY},D_1,D_2)\right\}.
\end{align}
Let $P_{\hat{X}\hat{Y}|XY}^*$ be the optimal conditional distribution achieving $R_{XY}(P_{XY},D_1,D_2)$. Let $P_{\hat{X}\hat{Y}}^*$ be induced by $P_{\hat{X}\hat{Y}|XY}^*$ and $P_{XY}$. Define the joint $(D_1,D_2)$-tilted information density as
\begin{align}
\imath_{XY}(x,y|D_1,D_2,P_{XY})
&:=-\log \mathbb{E}_{P_{\hat{X}\hat{Y}}^*}\left[\exp\left(\nu_1^*(D_1-d_{X}(x,\hat{X}))+\nu_2^*(D_2-d_{Y}(y,\hat{Y}))\right)\right]\label{defjrdtilted},
\end{align}
where
\begin{align}
\nu_1^*:&=-\frac{\partial R_{XY}(P_{XY},D,D_2)}{\partial D}\bigg|_{D=D_1},\label{defjrdnu1}\\
\nu_2^*:&=-\frac{\partial R_{XY}(P_{XY},D_1,D)}{\partial D}\bigg|_{D=D_2}\label{defjrdnu2}.
\end{align}
\begin{lemma}
\label{propertyjointrd}
The properties of $\imath_{XY}(x,y|D_1,D_2,P_{XY})$ include
\begin{itemize}
\item The joint rate-distortion function is the expectation of the joint tilted information density, i.e.,
\begin{align}
R_{XY}(P_{XY},D_1,D_2)=\mathbb{E}_{P_{XY}}\left[\imath_{XY}(X,Y|D_1,D_2,P_{XY})\right].
\end{align}
\item For $P_{\hat{X}\hat{Y}}^*$-almost every $(\hat{x},\hat{y})$, 
\begin{align}
\imath_{XY}(x,y|D_1,D_2,P_{XY})=\log \frac{P_{\hat{X}\hat{Y}|XY}^*(\hat{x},\hat{y}|x,y)}{P_{\hat{X}\hat{Y}}^*(\hat{x},\hat{y})}+\nu_1^*(d_{X}(x,\hat{x})-D_1)+\nu_2^*(d_{Y}(y,\hat{y})-D_2).
\end{align}

\end{itemize}
\end{lemma}
The proof of Lemma \ref{propertyjointrd} is provided in Appendix \ref{proofpropertyjointrd}. Lemma \ref{propertyjointrd} can be proved in a similar manner as \cite[Lemma 1]{watanabe2015second} and \cite[Lemma 1.4]{csiszar1974}. By considering a fixed rate triplet on the Pangloss plane, we can relate $\jmath_{XY}(x,y|R_1^*,R_2^*,D_1,D_2,P_{XY})$ to $\imath_{XY}(x,y|D_1,D_2,P_{XY})$.

\begin{lemma}
\label{panglosstilted}
When $(R_0^*,R_1^*,R_2^*)\in\calR_{\mathrm{pg}}(D_1,D_2|P_{XY})$ and $R_0^*>0$,
\begin{align}
\jmath_{XY}(x,y|R_1^*,R_2^*,D_1,D_2,P_{XY})=\imath_{XY}(x,y|D_1,D_2,P_{XY})-R_1^*-R_2^*\label{needproof}.
\end{align}
\end{lemma}
We defer the proof of Lemma \ref{panglosstilted} to Appendix \ref{proofpanglosstilted}. The proof of Lemma \ref{panglosstilted} invokes Lemma \ref{propertytilted}. Besides, we use an idea from \cite{viswanatha2014} in which it was shown that the following Markov chains hold for the optimal test channels $P_{W|XY}^*$ achieving $\calR(R_1^*,R_2^*,D_1,D_2|P_{XY})$ and  $P_{\hat{X}|XW}^*$ as well as $P_{\hat{Y}|YW}^*$ achieving conditional rate-distortion functions $R_{X|W}(P_{XW}^*,D_1)$ and $R_{Y|W}(P_{YW}^*,D_2)$: 
\begin{align}
\hat{X}&\to W\to \hat{Y} \\* 
(X,Y)&\to (\hat{X},\hat{Y})\to W\\* 
\hat{X}&\to(X,Y,W)\to \hat{Y}\\*  
\hat{X}&\to(XW)\to Y \\* 
\hat{Y}&\to (Y,W) \to X.
\end{align}
 Invoking Lemma \ref{panglosstilted}, for a rate triplet $(R_0^*,R_1^*,R_2^*)$ on the Pangloss plane, we can significantly simplify the calculation of $\calL(R_0^*,R_1^*,R_2^*,D_1,D_2,\epsilon)$.
\begin{proposition}
\label{proppangloss}
When $(R_0^*,R_1^*,R_2^*)\in\calR_{\mathrm{pg}}(D_1,D_2|P_{XY})$ and the conditions in Theorem \ref{mainresult} are satisfied, we have
\begin{align}
\calL(R_0^*,R_1^*,R_2^*,D_1,D_2,\epsilon)
=\left\{(L_0,L_1,L_2):L_0+L_1+L_2\geq \sqrt{\mathrm{V}(R_1^*,R_2^*,D_1,D_2|P_{XY})}\mathrm{Q}^{-1}(\epsilon)\right\},
\end{align}
where  the {\em rate-dispersion function}~\cite{kostina2012fixed} is 
\begin{align}
\mathrm{V}(R_1^*,R_2^*,D_1,D_2|P_{XY})=\mathrm{Var}[\jmath_{XY}(X,Y|R_1^*,R_2^*,D_1,D_2,P_{XY})]=\mathrm{Var}[\imath_{XY}(X,Y|D_1,D_2,P_{XY})].
\end{align}
\end{proposition}
\begin{remark}
\label{remarkprop}
To obtain the corresponding results for $D_1=0$ or $D_2=0$, we need to define the joint $(D_1,D_2)$-tilted information density correspondingly. Define
\begin{align}
\imath_{XY}(x,y|D_1,D_2,P_{XY}):=
\left\{
\begin{array}{lr}
-\log \mathbb{E}_{P_{\hat{Y}|X=x}^*}\big[\exp \big(\nu_2^*(D_2-d_{Y}(y,\hat{Y}))\big)\big]-\log P_{X}(x) & D_1=0,D_2>0,\\
-\log \mathbb{E}_{P_{\hat{X}|Y=y}^*}\big[\exp \big(\nu_1^*(D_1-d_{X}(x,\hat{X}))\big)\big]-\log P_{Y}(y) & D_1>0,D_2=0,\\
-\log P_{XY}(x,y) & D_1=0,D_2=0.\\
\end{array}
\right.
\end{align}
Combining the techniques used in this paper and the lossless case in \cite{watanabe2015second}, it is not hard to verify that Proposition \ref{proppangloss} is still valid when $D_1=0$ and/or $D_2=0$. We provide  a justification for $D_1=0$ and $D_2>0$ in Appendix \ref{justifyremark}.
\end{remark}

\subsection{A Numerical Example for Boundary Points on the Pangloss Plane}
We consider a doubly symmetric binary source (DSBS), where $\calX=\calY=\{0,1\}$, $P_{XY}(0,0)=P_{XY}(1,1)=\frac{1-p}{2}$ and $P_{XY}(0,1)=P_{XY}(1,0)=\frac{p}{2}$ for $p\in[0,\frac{1}{2}]$. We consider $\hat{\calX}=\hat{\calY}=\{0,1\}$ and Hamming distortion for both sources, i.e., $d_{X}(x,\hat{x})=1\{x=\hat{x}\}$ and $d_{Y}(y,\hat{y})=1\{y=\hat{y}\}$. Under this setting, we consider $R_1=R_2=R$ and $D_1=D_2=D$. 
Denote $h(\delta)=-\delta\log(\delta)-(1-\delta)\log(1-\delta)$ as the binary entropy function and define $f(x):=-x\log x$.  Define $p_1:=\frac{1}{2}-\frac{1}{2}\sqrt{1-2p}$. From Exercise 2.7.2 in \cite{berger1971rate}, we obtain
\begin{align}
R_{XY}(P_{XY},D,D)
&=\left\{
\begin{array}{lr}
1+h(p)-2h(D)&0\leq D\leq p_1,\\
f(1-p)-\frac{1}{2}\left(f(2D-p)+f(2(1-D)-p)\right) &p_1\leq D\leq \frac{1}{2}.
\end{array}\right.
\end{align}
It was shown in Example 2.5(A) in \cite{gray1974source} that for $0 \leq D\leq \Delta\leq p_1$, if we choose $R_0=R_{XY}(P_{XY},\Delta,\Delta)$, $R_1=R_2=h(\Delta)-h(D)$, then $(R_0,R_1,R_2)\in\calR_{\mathrm{pg}}(D,D|P_{XY})$. When $D\leq p_1$, the joint $(D,D)$-tilted information density is 
\begin{align}
\imath_{XY}(0,0|D,D,P_{XY})
=\imath_{XY}(1,1|D,D,P_{XY})&=\log\frac{1}{(2p-1)D-(2p-1)D^2+\frac{1}{2}(1-p)}-2h(D),\\
\imath_{XY}(0,1|D,D,P_{XY})=\imath_{XY}(1,0|D,D,P_{XY})&=\log\frac{1}{(2p-1)D^2-(2p-1)D+\frac{1}{2}p}-2h(D).
\end{align}
Hence,
\begin{align}
\mathrm{Var}[\imath_{XY}(X,Y|D,D,P_{XY})]
&=\sum_{x,y}P_{XY}(x,y)
\left(\imath_{XY}(x,y|D,D,P_{XY})-R_{XY}(P_{XY},D,D)\right)^2\\*
&\nn=(1-p)\left(\log\frac{1}{(2p-1)D-(2p-1)D^2+\frac{1}{2}(1-p)}-1-h(p)\right)^2\\*
&\qquad +
p\left(\log\frac{1}{(2p-1)D^2-(2p-1)D+\frac{1}{2}p}-1-h(p)\right)^2
\label{egvaluev}.
\end{align}
For a rate triplet $(R_0^*,R_1^*,R_2^*)\in\calR_{\mathrm{pg}}(D,D)$ satisfying the conditions in Theorem \ref{mainresult}, define 

\begin{align}
R_{\mathrm{sum}}(n)
&:=R_0^*+R_1^*+R_2^*+\min_{(L_0,L_1,L_2)\in\calL(R_0^*,R_1^*,R_2^*,D,D,\epsilon)}\frac{L_0+L_1+L_2}{\sqrt{n}}\\*
&=R_0^*+R_1^*+R_2^*+\sqrt{\frac{\mathrm{Var}[\imath_{XY}(X,Y|D,D,P_{XY})]}{n}}\mathrm{Q}^{-1}(\epsilon)\label{useprop},
\end{align}
where \eqref{useprop} follows from Proposition \ref{proppangloss} and \eqref{egvaluev}.

For $p=0.48$ and $D=0.15$, we plot $R_{\mathrm{sum}}(n)$ in Figure \ref{sumratesecondorder} for $\epsilon=0.01$ and $\epsilon=0.99$ where the blue line corresponds to the first-order sum rate $R_0^*+R_1^*+R_2^*$. This figure demonstrates the convergence of an approximation of the finite blocklength fundamental limit to the first-order fundamental limit.
\begin{figure}[t]
\centering
\includegraphics[width=12cm]{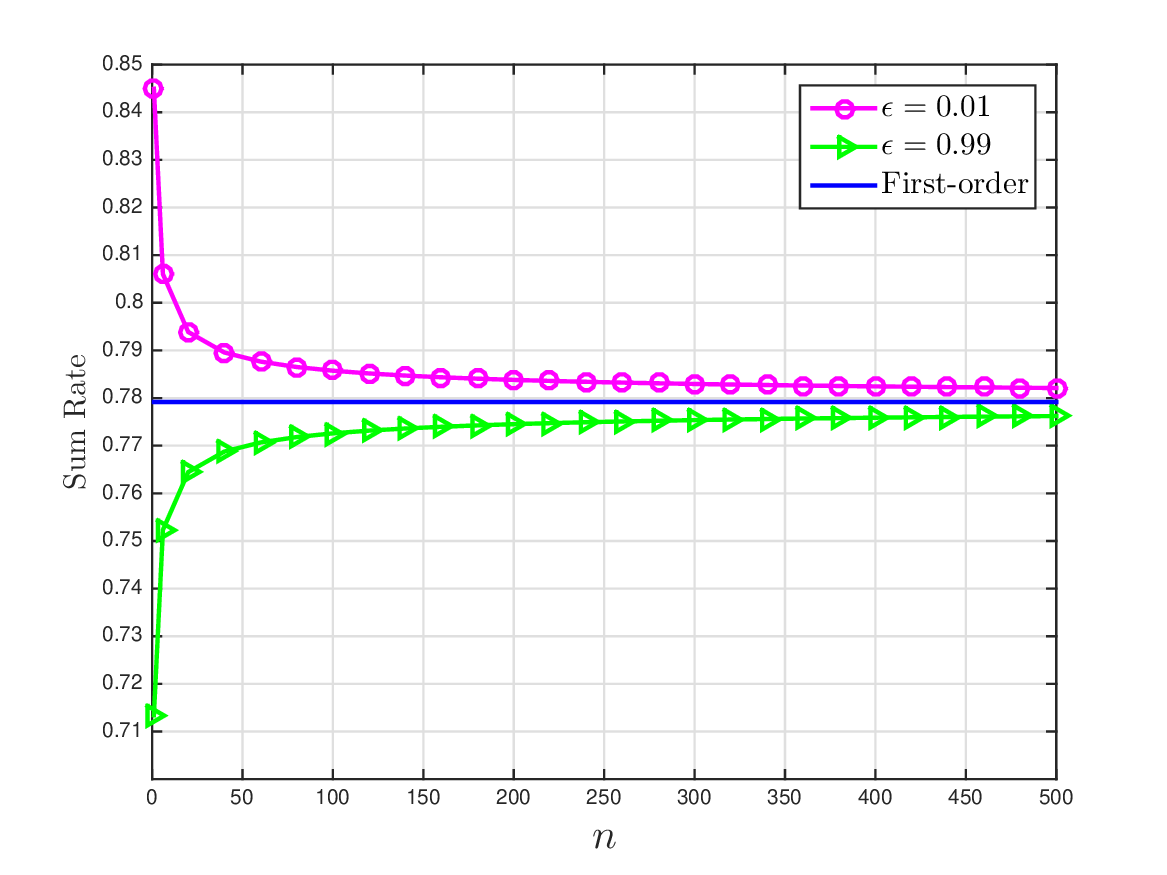}
\caption{Sum rate with $p=0.48$ and $D=0.15$.}
\label{sumratesecondorder}
\end{figure}

\section{Proof of Second-Order Asymptotics (Theorem \ref{mainresult}) } 
\label{secondorderproof}
\subsection{Achievability Proof}

In this part, we first prove that for any given joint type $Q_{XY} \in\calP_n(\calX\times\calY)$, there exists an $(n,M_0,M_1,M_2)$-code such that the excess-distortion probability is mainly due to the incorrect decoding of side information $W$. To do so, we present  a novel type covering lemma for discrete lossy Gray-Wyner problem. Using this result, we then prove an upper bound of the excess-distortion probability for the $(n,M_0,M_1,M_2)$-code. Finally, we establish the achievable second-order coding region by estimating this probability.

Define four constants
\begin{align}
c_0&=\left(3|\calX|\cdot|\calY|\cdot|\calW|+4\right),\\
c_0'&=c_0+|\calX|\cdot|\calY|\label{defc0p},\\ 
c_1&=\left(\frac{11\overline{d}_X}{\underline{d}_X}|\calX|\cdot|\calY|\cdot|\calW|+3|\calX|\cdot|\calW|\cdot|\hat{\calX}|+5\right)\label{defc1},\\
c_2&=\left(\frac{11\overline{d}_Y}{\underline{d}_Y}|\calX|\cdot|\calY|\cdot|\calW|+3|\calY|\cdot|\calW|\cdot|\hat{\calY}|+5\right)\label{defc2}.
\end{align}
We begin by presenting a type covering lemma that is suited to the needs of second-order analysis for the lossy Gray-Wyner problem. 
\begin{lemma}
\label{achievable}
Let $n$ satisfy $(n+1)^4>n\log|\calX|\cdot|\calY|$, $\log n\geq \frac{|\calX|\cdot|\calW|\cdot|\hat{\calX}|\log|\calX|\overline{d}_X}{D_1}$, $\log n\geq \frac{|\calY|\cdot|\calW|\cdot|\hat{\calY}|\log|\calY|\overline{d}_Y}{D_2}$, and $\log n\geq \log\frac{|\hat{\calX}|}{|\calY|}$. Given a joint type $Q_{XY}\in\calP_n(\calX\times\calY)$, for any rate pair $(R_1,R_2)\in\bbR_{++}^2$ such that $\rvR_0(R_1,R_2,D_1,D_2|Q_{XY})$ is achievable by some test channel, there exists a conditional type $Q_{W|XY}\in\calV_{n}(\calW,Q_{XY})$ such that the following holds:
\begin{itemize}
\item There exists a set $\calC_n\subset\calT_{Q_{W}}$ ($Q_W$ is induced by $Q_{XY}$ and $Q_{W|XY}$) such that
\begin{itemize}
\item For any $(x^n,y^n)\in\calT_{Q_{XY}}$, there exists a $w^n\in\calC_n$ whose joint type with $(x^n,y^n)$ is $Q_{XYW}$, i.e., $(x^n,y^n,w^n)\in\calT_{Q_{XYW}}$.
\item  The size of $\calC_n$ is upper bounded by
\begin{align}
\frac{1}{n}\log |\calC_n|\leq \rvR_0(R_1,R_2,D_1,D_2|Q_{XY})+c_0\frac{\log(n+1)}{n}.
\end{align}
\end{itemize}
\item 
 For each $w^n\in\calT_{Q_{W|XY}}(x^n,y^n)$, there exists sets $\calB_{\hat{X}}(w^n)\in\hat{\calX}^n$ and $\calB_{\hat{Y}}(w^n)\in\hat{\calY}^n$ satisfying
\begin{itemize}
\item For each $(x^n,y^n)\in\calT_{Q_{XY|W}}(w^n)$, there exists $\hat{x}^n\in\calB_{\hat{X}}(w^n)$ and $\hat{y}^n\in\calB_{\hat{Y}}(w^n)$ such that
$d_X(x^n,\hat{x}^n)\leq D_1$ and $d_{Y}(y^n,\hat{y}^n)\leq D_2$,
\item The sizes of $\calB_{\hat{X}}(w^n)$ and $\calB_{\hat{Y}}(w^n) $ are upper bounded as
\begin{align}
\frac{1}{n}\log |\calB_{\hat{X}}(w^n)|&\leq R_1+c_1\frac{\log n}{n},\\*
\frac{1}{n}\log |\calB_{\hat{Y}}(w^n)|&\leq R_2+c_2\frac{\log n}{n}.
\end{align}
\end{itemize}

\end{itemize} 
\end{lemma}
The proof of Lemma \ref{achievable} is given in Appendix \ref{proofachievable}. Lemma \ref{achievable} is proved by combining a few ideas from the literature: a type covering lemma for the conditional rate-distortion problem (modified from Lemma 4.1 in \cite{csiszar2011information} for the standard rate-distortion problem and Lemma 8 in \cite{no2016} for the successive refinement problem),  a type covering lemma for the common side information for the Gray-Wyner problem (Lemma 4 in \cite{watanabe2015second}) and finally, a uniform continuity lemma for the conditional rate-distortion function (modified from \cite{no2016,palaiyanur2008uniform}).

The proof of Lemma \ref{achievable} adopts similar idea as the proof of the first-order coding region~\cite{gray1974source}. The main idea is that we first send the common information via the common link carry $S_0$ and then we consider two conditional rate-distortion problems on the two private links carrying $S_1,S_2$ using the common information as the side information.

Invoking Lemma \ref{achievable}, we show that there exists an $(n,M_0,M_1,M_2)$-code whose excess-distortion probability can be upper bounded as follows. Recall the definitions of $c_0'$ in \eqref{defc0p}, $c_1$ in \eqref{defc1} and $c_2$ in \eqref{defc2}. Define three rates
\begin{align}
R_{0,n}&=\frac{1}{n}\log M_0-c_0'\frac{\log (n+1)}{n},\\*
R_{1,n}&=\frac{1}{n}\log M_1-c_1\frac{\log n}{n},\\*
R_{2,n}&=\frac{1}{n}\log M_2-c_2\frac{\log n}{n}.
\end{align}
\begin{lemma}
\label{upperboundexcessp}
There exists an $(n,M_0.M_1,M_2)$-code such that
\begin{align}
\epsilon_n(D_1,D_2)\leq \Pr\left(R_{0,n}<\rvR_0(R_{1,n},R_{2,n},D_1,D_2|\hat{T}_{X^nY^n})\right). \label{eqn:ed_bd}
\end{align}
\end{lemma}
The proof of Lemma \ref{upperboundexcessp} is similar to \cite[Lemma 5]{watanabe2015second} and given in Appendix \ref{proofupperboundexcessp}.

Given two probability mass functions $P$ and $Q$ on a common alphabet $\calX$, define  the $\ell_\infty $ distance $\|P-Q\|_{\infty}:=\max_{x\in\calX}|P(x)-Q(x)|$. Then, define the typical set for joint types as
\begin{align}
\calA_{n}(P_{XY}):=\left\{Q_{XY}\in\calP_n(\calX\times\calY):\|\Gamma(Q_{XY})-\Gamma(P_{XY})\|_{\infty}\leq \sqrt{\frac{\log n}{n}}\right\}\label{eqn:typ_set1},
\end{align}
where the notation $\Gamma$ is defined in Section \ref{sec:notation}.

From Lemma 22 in \cite{tan2014state}, we know
\begin{align}
\Pr\left(\hat{T}_{X^nY^n}\notin \calA_{n}(P_{XY})\right)\leq \frac{2|\calX|\cdot|\calY|}{n^2}.
\end{align}
For a rate triplet $(R_0^*,R_1^*,R_2^*)$ satisfying conditions in Theorem \ref{mainresult}, we choose 
\begin{align}
\frac{1}{n}\log M_0&=\rvR_0(R_1^*,R_2^*,D_1,D_2|P_{XY})+\frac{L_0}{\sqrt{n}}+c_0'\frac{\log (n+1)}{n},\\*
\frac{1}{n}\log M_1&=R_1^*+\frac{L_1}{\sqrt{n}}+c_1\frac{\log n}{n},\\*
\frac{1}{n}\log M_2&=R_2^*+\frac{L_2}{\sqrt{n}}+c_2\frac{\log n}{n}.
\end{align}
Hence,
\begin{align}
R_{i,n}=R_{i}^*+\frac{L_i}{\sqrt{n}},~i=0,1,2.
\end{align}
From the conditions in Theorem \ref{mainresult}, we know that the second derivatives of $\rvR_0(R_1,R_2,D_1,D_2|P_{XY})$ with respect to $(R_1,R_2,P_{XY})$ are bounded  around a neighborhood of $(R_1^*,R_2^*,P_{XY})$. Hence, for any $\hat{T}_{x^ny^n}\in\calA_{n}(P_{XY})$, for large $n$, applying Taylor's expansion for $\rvR_0(R_{1,n},R_{2,n},D_1,D_2|\hat{T}_{x^ny^n})$ and invoking Lemma \ref{linkrjxy}, we obtain:
\begin{align}
\nn &\rvR_0(R_{1,n},R_{2,n},D_1,D_2|\hat{T}_{x^ny^n})\\*
\nn&=\rvR_0(R_1^*,R_2^*,D_1,D_2|P_{XY})-\lambda_1^*\frac{L_1}{\sqrt{n}}-\lambda_2^*\frac{L_2}{\sqrt{n}}+O\left((R_{1,n}-R_1^*)^2+(R_{2,n}-R_2^*)^2\right)+O\left(\|\Gamma(\hatT_{x^ny^n})-\Gamma(P_{XY})\|^2\right)\\
&\qquad+\sum_{i=1}^m \left(\Gamma_i(\hatT_{x^ny^n})-\Gamma_i(P_{XY})\right)\Bigg(\jmath(i|R_1^*,R_2^*,D_1,D_2,\Gamma(P_{XY}))-\jmath(m|R_1^*,R_2^*,D_1,D_2,\Gamma(P_{XY}))\Bigg)\\
&= \nn \rvR_0(R_1^*,R_2^*,D_1,D_2|P_{XY})-\lambda_1^*\frac{L_1}{\sqrt{n}}-\lambda_2^*\frac{L_2}{\sqrt{n}}+\sum_{x,y}\left(\hat{T}_{x^ny^n}(x,y)-P_{XY}(x,y)\right)\jmath_{XY}(x,y|R_1^*,R_2^*,D_1,D_2,P_{XY})\\*
&\qquad +O\left(\frac{\log n}{n}\right)\\*
&\leq \sum_{x,y}Q_{XY}(x,y)\jmath_{XY}(x,y|R_1^*,R_2^*,D_1,D_2,P_{XY})-\lambda_1^*\frac{L_1}{\sqrt{n}}-\lambda_2^*\frac{L_2}{\sqrt{n}}+O\left(\frac{\log n}{n}\right)\label{minr0}\\*
&=\frac{1}{n}\sum_{i=1}^n\jmath_{XY}(x_i,y_i|R_1^*,R_2^*,D_1,D_2,P_{XY})-\lambda_1^*\frac{L_1}{\sqrt{n}}-\lambda_2^*\frac{L_2}{\sqrt{n}}+O\left(\frac{\log n}{n}\right)\label{taylorfisrtt},
\end{align}
where \eqref{minr0} follows from Lemma \ref{propertytilted} and the definition of the typical set $\calA_{n}(P_{XY})$ in \eqref{eqn:typ_set1}. Define $\xi_n=\frac{\log n}{n}$.

Invoking Lemma \ref{upperboundexcessp}, we can upper bound the excess-distortion probability 
as follows:
\begin{align}
&\nn \epsilon_n(D_1,D_2)\\
&\leq \Pr\left(R_{0,n}<\rvR_0(R_{1,n},R_{2,n},D_1,D_2|\hat{T}_{X^nY^n})\right)\\
&\leq \Pr\left(\hat{T}_{X^nY^n}\in\calA_{n}(P_{XY}), R_{0,n}<\rvR_0(R_{1,n},R_{2,n},D_1,D_2|\hat{T}_{X^nY^n})\right)+\Pr\left(\hat{T}_{X^nY^n}\notin\calA_{n}(P_{XY})\right)\\
&\leq \Pr\left(R_{0,n}<\frac{1}{n}\sum_{i=1}^n \jmath_{XY}(X_i,Y_i|R_1^*,R_2^*,D_1,D_2,P_{XY})-\lambda_{1}^*
\frac{L_1}{\sqrt{n}}-\lambda_{2}^*\frac{L_2}{\sqrt{n}}+O(\xi_n)\right)+\frac{2|\calX|\cdot|\calY|}{n^2}\\
&=\Pr\left(\frac{L_0}{\sqrt{n}}+\lambda_{1}^*
\frac{L_1}{\sqrt{n}}+\lambda_{2}^*\frac{L_2}{\sqrt{n}}+O(\xi_n)<\frac{1}{n}\sum_{i=1}^n \left(\jmath_{XY}(X_i,Y_i|R_1^*,R_2^*,D_1,D_2,P_{XY})-\rvR_0(R_1^*,R_2^*,D_1,D_2|P_{XY})\right)\right)\nn\\*
&\qquad\qquad +\frac{2|\calX|\cdot|\calY|}{n^2}\\
&\leq \rmQ\left(\frac{L_0+\lambda_1^*L_1+\lambda_2^*L_2+O\left(\sqrt{n}\xi_n\right)}{\sqrt{\mathrm{V}(R_1^*,R_2^*,D_1,D_2|P_{XY})}}\right)+\frac{6\mathrm{T}(R_1^*,R_2^*,D_1,D_2)}{\sqrt{n}\mathrm{V}^{3/2}(R_1^*,R_2^*,D_1,D_2)}+\frac{2|\calX|\cdot|\calY|}{n^2},\label{berryessen}
\end{align}
where \eqref{berryessen} follows from the Berry-Esseen Theorem and  $\mathrm{T}(R_1^*,R_2^*,D_1,D_2)$ is third absolute moment of the tilted information density  $\jmath_{XY}(X,Y|R_1^*,R_2^*,D_1,D_2,P_{XY})$. From the conditions in Theorem \ref{mainresult}, we conclude that $\mathrm{T}(R_1^*,R_2^*,D_1,D_2)$ is finite.
Therefore, if $(L_0,L_1,L_2)$ satisfies
\begin{align}
L_0+\lambda_1^*L_1+\lambda_2^*L_2 \geq \sqrt{\mathrm{V}(R_1^*,R_2^*,D_1,D_2|P_{XY})}\mathrm{Q}^{-1}(\epsilon),
\end{align}
then $\limsup_{n\to\infty}\epsilon_n(D_1,D_2)\leq \epsilon$.

\subsection{Converse Proof}
In this part, we prove an outer bound for the second-order region under conditions stated in Theorem \ref{mainresult}. We follow the method in \cite{watanabe2015second} closely. First, we invoke the strong converse in \cite{wei2009strong} to establish a type-based strong converse. Second, we prove a lower bound on excess-distortion probability $\epsilon_n(D_1,D_2)$ by using the type-based strong converse. Finally, we use Taylor expansion and apply Berry-Esseen Theorem to obtain a outer region expressed essentially using $\mathrm{V}(R_1^*,R_2^*,D_1,D_2|P_{XY})$, i.e., the variance of $\jmath_{XY}(X,Y|R_1^*,R_2^*,D_1,D_2,P_{XY})$ for a rate triplet $(R_0^*,R_1^*,R_2^*)$.
 
We now consider an $(n,M_0,M_1,M_2)$-code for the correlated source $(X^n,Y^n)$ with joint distribution $U_{\calT_{Q_{XY}}}(x^n,y^n)=|\calT_{Q_{XY}}|^{-1}$, the uniform distribution over the type class $\calT_{Q_{XY}}$.
\begin{lemma}
\label{typestrongconverse}
If the non-excess-distortion probability satisfies
\begin{align}
\label{assumption}
\Pr\left(d_{X}(X^n,\hat{X}^n)\leq D_1, d_{Y}(Y^n,\hat{Y}^n)\leq D_2|(X^n,Y^n)\in\calT_{Q_{XY}}\right)\geq \exp(-n\alpha)
\end{align}
for some positive number $\alpha$, then for $n$ large enough such that $\log n\geq\max\{\overline{d}_X,\overline{d}_Y\}\log|\calX|$, a conditional distribution $Q_{W|XY}$ with $|\calW|\leq |\calX|\cdot|\calY|+2$ such that
\begin{align}
\frac{1}{n}\log M_0&\geq I(X,Y;W)-\frac{\Big(|\calX|\cdot|\calY|+1\Big)\log(n+1)}{n}-\alpha,\\
\frac{1}{n}\log M_1&\geq R_{X|W}(Q_{XW},D_1)-\frac{\log n}{n},\\
\frac{1}{n}\log M_2&\geq R_{Y|W}(Q_{YW},D_2)-\frac{\log n}{n}.
\end{align}
where $(X,Y,W)\sim Q_{XY}\times Q_{W|XY}$.
\end{lemma}

The proof of Lemma \ref{typestrongconverse} is given in Appendix \ref{prooftypeconverse}. The proof of Lemma \ref{typestrongconverse} is similar to \cite[Lemma 6]{watanabe2015second} but we need to also combine this with the (weak) converse proof for lossy Gray-Wyner problem under the expected distortion criterion in \cite{gray1974source}. See Definition \ref{deffirst}.

We then prove a lower bound on the excess-distortion probability $\epsilon_n(D_1,D_2)$
in \eqref{defexcessprob}. 
Define the constant $c=\frac{|\calX|\cdot|\calY|+2}{n}$ and the three quantities
\begin{align}
R_{0,n}&:=\frac{1}{n}\log M_0+c\frac{\log (n+1)}{n},\label{eq1}\\
R_{1,n}&:=\frac{1}{n}\log M_1+\frac{\log n}{n},\label{eq2}\\
R_{2,n}&:=\frac{1}{n}\log M_2+\frac{\log n}{n}.\label{eq3}
\end{align}
\begin{lemma}
\label{lowerboundexcessp}
For any $(n,M_0,M_1,M_2)$-code such that $\log n\geq\max\{\overline{d}_X,\overline{d}_Y\}\log|\calX|$
\begin{align}
\epsilon_n(D_1,D_2)\geq \Pr\left(R_{0,n}< \rvR_0(R_{1,n},R_{2,n},D_1,D_2|\hat{T}_{X^nY^n})\right)-\frac{1}{n}.
\end{align} 
\end{lemma}
The proof of Lemma \ref{lowerboundexcessp} is similar to \cite[Lemma 7]{watanabe2015second} and given in Appendix \ref{prooflbexcess}.

Choose $(M_0,M_1,M_2)$ such that
\begin{align}
\frac{1}{n}\log M_0&=R_0^*+\frac{L_0}{\sqrt{n}}-c\frac{\log (n+1)}{n},\\
\frac{1}{n}\log M_1&=R_1^*+\frac{L_1}{\sqrt{n}}-\frac{\log n}{n},\\
\frac{1}{n}\log M_2&=R_2^*+\frac{L_2}{\sqrt{n}}-\frac{\log n}{n}.
\end{align}

Hence, according to \eqref{eq1} to \eqref{eq3} in Lemma \ref{lowerboundexcessp}, for $i\in[0:2]$,
\begin{align}
R_{i,n}&=R_i^*+\frac{L_i}{\sqrt{n}}.
\end{align} 
Invoking Lemma \ref{lowerboundexcessp}, in a similar manner as the achievability proof, we obtain
\begin{align}
\nn&\epsilon_n(D_1,D_2)\\*
&\geq \Pr\left(R_{0,n}<\rvR_0(R_{1,n},R_{2,n},D_1,D_2|\hat{T}_{X^nY^n}),\right)-\frac{1}{n}\\
&\geq \Pr\left(R_{0,n}<\rvR_0(R_{1,n},R_{2,n},D_1,D_2|\hat{T}_{X^nY^n}),~\hat{T}_{X^nY^n}\in \calA_{n}(P_{XY})\right)-\frac{1}{n}\\
&\geq\Pr\left(R_0^*+\frac{L_0}{\sqrt{n}}<\frac{1}{n}\sum_{i=1}^n\jmath_{XY}(X_i,Y_i|R_1^*,R_2^*,D_1,D_2,P_{XY})-\lambda_1^*\frac{L_1}{\sqrt{n}}-\lambda_2^*\frac{L_2}{\sqrt{n}}+O(\xi_n),~\hat{T}_{X^nY^n}\in \calA_{n}(P_{XY})\right)-\frac{1}{n}\\ \nn&\geq \Pr\left(R_0^*+\frac{L_0}{\sqrt{n}}<\frac{1}{n}\sum_{i=1}^n\jmath_{XY}(X_i,Y_i|R_1^*,R_2^*,D_1,D_2,P_{XY})-\lambda_1^*\frac{L_1}{\sqrt{n}}-\lambda_2^*\frac{L_2}{\sqrt{n}}+O(\xi_n)\right)\\
&\qquad -\Pr\left(\hat{T}_{X^nY^n}\notin \calA_n(P_{XY})\right)-\frac{1}{n} \label{eqn:rev_union_bd}\\
\nn&=\Pr\left(\frac{L_0}{\sqrt{n}}+\lambda_1^*\frac{L_1}{\sqrt{n}}+\lambda_2^*\frac{L_2}{\sqrt{n}}+O(\xi_n)<\frac{1}{n}\sum_{i=1}^n\jmath_{XY}(X_i,Y_i|R_1^*,R_2^*,D_1,D_2,P_{XY})-\rvR_0(R_1^*,R_2^*,D_1,D_2|P_{XY})\right)\\
&\qquad -\frac{2|\calX|\cdot|\calY|}{n^2}-\frac{1}{n}\\
&\geq \rmQ\left(\frac{L_0+\lambda_1^*L_1+\lambda_2^*L_2+O\left(\sqrt{n}\xi_n\right)}{\sqrt{\mathrm{V}(R_1^*,R_2^*,D_1,D_2|P_{XY})}}\right)-\frac{6\mathrm{T}(R_1^*,R_2^*,D_1,D_2)}{\sqrt{n}\mathrm{V}^{3/2}(R_1^*,R_2^*,D_1,D_2)}-\frac{2|\calX|\cdot|\calY|}{n^2}-\frac{1}{n},
\end{align} 
where \eqref{eqn:rev_union_bd} follows from the fact that $\Pr(\calE\cap\calF)\ge\Pr(\calE)-\Pr(\calF^c)$. 
Hence, if $(L_0,L_1,L_2)$ satisfies
\begin{align}
L_0+\lambda_1^*L_1+\lambda_2^*L_2< \sqrt{\mathrm{V}(R_1^*,R_2^*,D_1,D_2|P_{XY})}\mathrm{Q}^{-1}(\epsilon),
\end{align} 
then $\liminf_{n\to\infty}\epsilon_n(D_1,D_2)>\epsilon$. Therefore, for sufficiently large $n$, any second-order $(R_0^*,R_1^*,R_2^*,D_1,D_2,\epsilon)$-achievable triplet $(L_0,L_1,L_2)$ must satisfy
\begin{align}
\label{outerregion}
L_0+\lambda_1^*L_1+\lambda_2^*L_2\geq \sqrt{\mathrm{V}(R_1^*,R_2^*,D_1,D_2|P_{XY})}\mathrm{Q}^{-1}(\epsilon).
\end{align}

\section{Large Deviations Analysis}
\label{largedeviations}
In this section, we define the error exponent and present the results together with the proof.
\begin{definition} \label{defee}
{\em  A  number $E\ge 0$ is said to be an {\em $(R_0,R_1,R_2,D_1,D_2)$-achievable error exponent} if there exists a sequence of $(n,M_0,M_1,M_2)$-codes such that,
\begin{align}
\limsup_{n\to\infty}\frac{1}{n}\log M_i\leq  R_i,~i=0,1,2,
\end{align}
and
\begin{align}
\liminf_{n\to\infty}-\frac{\log \epsilon_n(D_1,D_2)}{n}\geq E.
\end{align}
The supremum of all $(R_0,R_1,R_2,D_1,D_2)$-achievable error exponent is denoted as $E^*(R_0,R_1,R_2|D_1,D_2)$. }
\end{definition}
Define the function
\begin{align}
\label{defeef}
F(P_{XY},R_0,R_1,R_2,D_1,D_2):=\inf_{Q_{XY}:\rvR_0(R_1,R_2,D_1,D_2|Q_{XY})\geq R_0} D(Q_{XY}\|P_{XY}).
\end{align}
\begin{theorem}
\label{eegray}
The optimal error exponent for discrete lossy Gray-Wyner problem is
\begin{align}
E^*(R_0,R_1,R_2|D_1,D_2)&=F(P_{XY},R_0,R_1,R_2,D_1,D_2).
\end{align}
\end{theorem}

\begin{proof}
The proof of Theorem \ref{eegray} follows \cite{Marton74} closely. 

The achievability part follows from Lemma \ref{achievable}. We consider a sequence of $(n,M_0,M_1,M_2)$-codes where
\begin{align}
\frac{1}{n}\log M_0&=R_0+c_0\frac{\log (n+1)}{n},\\
\frac{1}{n}\log M_i&=R_i+c_i\frac{\log n}{n},~i=1,2.
\end{align}
Given $(R_0,R_1,R_2)$, we define the set 
\begin{align}
\calU_n=\bigcup_{Q_{XY}\in\calP_n(\calX\times\calY):\rvR_0(R_1,R_2,D_1,D_2|Q_{XY})\geq R_0}\calT_{Q_{XY}}.
\end{align}
Invoking Lemma \ref{achievable}, we know that for type $Q_{XY}$ such that $\rvR_0(R_1,R_2,D_1,D_2|Q_{XY})\leq R_0$, the excess-distortion probability is zero. Hence, by Sanov's theorem~\cite[Ch.\ 11]{cover2012elements},
\begin{align}
\epsilon_n(D_1,D_2)
&\leq (n+1)^{|\calX|\cdot|\calY|}\exp \left(-n\min_{Q_{XY}:\rvR_0(R_1,R_2,D_1,D_2|Q_{XY})\geq R_0}D(Q_{XY}\|P_{XY}) \right).
\end{align}
Therefore, we obtain
\begin{align}
\liminf_{n\to\infty}-\frac{\log \epsilon_n(D_1,D_2)}{n}\geq \min_{Q_{XY}:\rvR_0(R_1,R_2,D_1,D_2|Q_{XY})\geq R_0}D(Q_{XY}\|P_{XY}).
\end{align}
The proof for achievability part is now complete.

The converse part follows from strong converse \cite{wei2009strong} and the change-of-measure technique  in~\cite{csiszar2011information, haroutunian68}.  Define the set 
\begin{equation}
\calD_n:=\left\{(x^n,y^n):d_X(x^n,\phi_1(f_0(x^n,y^n),f_1(x^n,y^n)))>D_1~\mathrm{or}~d_Y(y^n,\phi_1(f_0(x^n,y^n),f_2(x^n,y^n)))>D_2\right\}.
\end{equation}
 Given rate triplet $(R_0,R_1,R_2)$, suppose the source has distribution $Q_{XY}$ such that $(R_0,R_1,R_2)\notin\calR(D_1,D_2|Q_{XY})$.  Since $\calR(D_1,D_2|Q_{XY})$ is closed, this means that   $\rvR_0(R_1,R_2,D_1,D_2|Q_{XY})\geq R_0+\delta$ for some $\delta>0$. By invoking the strong converse for the lossy Gray-Wyner problem~\cite{wei2009strong}, we obtain that
\begin{align}
Q_{XY}^n(\calD_n)\geq 1-\beta_n,
\end{align}
where $\beta_n\to0$ as $n\to\infty$. Then with a standard change-of-measure technique \cite{csiszar2011information,haroutunian68}, we obtain that for any $(n,M_0,M_1,M_2)$-code,
\begin{align}
\epsilon_n(D_1,D_2)=P_{XY}^n(\calD_n)\geq \exp\left(\frac{-nD(Q_{XY}\|P_{XY})-h(\beta_n)}{1-\beta_n}\right),
\end{align}
where $h(\delta)$ is the binary entropy function. Hence, 
\begin{align}
\label{converseee}
\limsup_{n\to\infty}-\frac{\log\epsilon_n(D_1,D_2)}{n}\leq D(Q_{XY}\|P_{XY}).
\end{align}
Minimizing \eqref{converseee} over all auxiliary distributions $Q_{XY}$ such that $\rvR_0(R_1,R_2,D_1,D_2|Q_{XY})\geq R_0+\delta$ and finally letting $\delta\to0$ (using the convexity and hence continuity of the exponent in the rate $R_0$), we complete the proof.
\end{proof}

\section{Moderate Deviations Analysis}
\label{moderatedeviations}
In this section, we define the moderate deviations constant   and present the main results as well as the proof. Fix a rate triplet $(R_0^*,R_1^*,R_2^*)\in\calR(D_1,D_2|P_{XY})$.
\begin{definition}[Moderate Deviations Constant]
{\em  
Consider any correlated source with joint probability mass function  $P_{XY}$ and any positive sequence $\{\rho_n\}_{n=1}^{\infty}$ satisfying
\begin{align}
 \lim_{n\to\infty}\rho_n&=0,\label{eqn:cond_rho1}\\
 \lim_{n\to\infty}n\rho_n^2&=\infty. \label{eqn:cond_rho2}
\end{align}
Let $\theta_i$ for $ i=[0:2]$ be three positive numbers. 
A number $\nu\ge 0$ is said to be a {\em $(R_0^*,R_1^*,R_2^*, D_1, D_2)$-achievable moderate deviations constant (with respect to $\{ \rho_n\}_{n=1}^{\infty}$)} if there exists a sequence of $(n,M_0,M_1,M_2)$-codes such that
\begin{align}
\limsup_{n\to\infty}\frac{1}{n\rho_n}(\log M_i-nR_i^*)\leq \theta_i,~i=0,1,2,
\end{align}
and
\begin{align}
\liminf_{n\to\infty}-\frac{\log \epsilon_n(D_1,D_2)}{n\rho_n^2}\geq \nu.
\end{align}
The supremum of all $(R_0^*,R_1^*,R_2^*, D_1, D_2)$-achievable moderate deviations constants is denoted as $\nu^*(R_0^*,R_1^*,R_2^* |  D_1, D_2)$.
}
\end{definition}

Define
\begin{align}
\theta&=\theta_0+\lambda_1^*\theta_1+\lambda_2^*\theta_2.
\end{align}
We are now ready to present the result on moderate deviations.
\begin{theorem}
\label{theoremmdc}
Given a rate triplet $(R_0^*,R_1^*,R_2^*)\in\calR(D_1,D_2|P_{XY})$ satisfying $\mathrm{V}(R_1^*,R_2^*,D_1,D_2|P_{XY})>0$ and the conditions in Theorem \ref{mainresult}, the moderate deviations constant is
\begin{align}
\nu^*(R_0^*,R_1^*,R_2^* | D_1, D_2)=\frac{\theta^2\log e}{2\mathrm{V}(R_1^*,R_2^*,D_1,D_2|P_{XY})}.
\end{align}
\end{theorem}
We observe that similarly to second-order asymptotics (Theorem \ref{mainresult}), the rate-dispersion function $\mathrm{V}(R_1^*,R_2^*,D_1,D_2|P_{XY})$ is a fundamental quantity that governs the speed of convergence of the excess-distortion probability to zero.

The proof of Theorem \ref{theoremmdc} can be done in similar manner as \cite{tan2012moderate} using Euclidean information theory \cite{borade2008}. Here we provide an alternative (and more direct) proof using the moderate deviations principle/theorem. See Dembo and Zeitouni \cite[Theorem~3.7.1]{dembo2009large}.
\subsection{Achievability}

For $i=0,1,2$, let 
\begin{align}
R_{i,n}&:=\frac{1}{n}\log M_i=R_i^*+\theta_i\rho_n\label{defmdcrin}.
\end{align}

Recall the definitions of $c_0'$ in \eqref{defc0p}, $c_1$ in \eqref{defc1} and $c_2$ in \eqref{defc2}.

Define
\begin{align}
\rho_{0,n}'&=\theta_0\rho_n-c_0'\frac{\log (n+1)}{n},\\
\rho_{i,n}'&=\theta_i\rho_n-c_i\frac{\log n}{n},~i=1,2\\
R_{i,n}'&=R_{i,n}-c_i\frac{\log n}{n}=R_i^*+\rho_{i,n}',~i=1,2,3.
\end{align} 
Define the typical set 
\begin{align}
\calA_{n}'(P_{XY})
:=\left\{Q_{XY}\in\calP_n(\calX\times\calY):\left\|Q_{XY}-P_{XY}\right\|_{1}\leq \frac{\theta\rho_n}{\sqrt{\mathrm{V}(R_1^*,R_2^*,D_1,D_2|P_{XY})}}\right\}. \label{eqn:typ}
\end{align}
Invoking Lemma \ref{upperboundexcessp} with $\frac{1}{n}\log M_i=R_i^*+\theta_i\rho_n$ for $i=0,1,2$, we obtain
\begin{align}
\epsilon_n(D_1,D_2)
&\leq \Pr\left(R_{0,n}'<\rvR_0(R_{1,n}',R_{2,n}',D_1,D_2|\hat{T}_{X^nY^n})\right)\\
&\leq \Pr\left(R_{0,n}'<\rvR_0(R_{1,n}',R_{2,n}',D_1,D_2|\hat{T}_{X^nY^n}),\hatT_{X^nY^n}\in\calA_n'(P_{XY})\right)+\Pr\left(\hat{T}_{X^nY^n}\notin\calA_n'(P_{XY})\right).
\end{align}
According to Weissman {\em et al.}~\cite{weissman2003inequalities}, we obtain
\begin{align}
\Pr\left(\hat{T}_{X^nY^n}\notin\calA_n'(P_{XY})\right)\leq \exp(|\calX|)\exp\left(-\frac{ n\rho_n^2\theta^2\log e}{2\mathrm{V}(R_1^*,R_2^*,D_1,D_2|P_{XY})}\right). \label{eqn:weiss}
\end{align}
For any $(x^n,y^n)$ such that $\hat{T}_{x^ny^n}\in\calA_n'(P_{XY})$, for $n$ large enough, applying Taylor's expansion similarly as \eqref{taylorfisrtt}, we obtain
\begin{align}
\nn&\rvR_0(R_{1,n}',R_{2,n}',D_1,D_2|\hat{T}_{x^ny^n})-R_{0,n}'\\
&=-\lambda_1^*\rho_{1,n}'-\lambda_2^*\rho_{2,n}'-\rho_{0,n}'+\sum_{x,y}\left(\hat{T}_{x^ny^n}(x,y)-P_{XY}(x,y)\right)\jmath_{XY}(x,y|R_1^*,R_2^*,D_1,D_2,P_{XY})\\
&\qquad+O\left(\rho_{1,n}'^2+\rho_{2,n}'^2+\rho_{0,n}'^2+\left\|\hat{T}_{x^ny^n}-P_{XY}\right\|^2\right)\\
&=-\left(\theta_0+\lambda_1^*\theta_1+\lambda_2^*\theta_2\right)\rho_n+\frac{1}{n}\sum_{i=1}^n\jmath_{XY}(x_i,y_i|R_1^*,R_2^*,D_1,D_2,P_{XY})-\rvR_0(R_1^*,R_2^*,D_1,D_2|P_{XY})+o\left(\rho_n\right)\label{assumptionuse},
\end{align}
where \eqref{assumptionuse} holds because (i) according to  \eqref{eqn:cond_rho2}, we have $\frac{\log( n+1)}{n}=o(\rho_n)$ and also $\rho_{i,n}'^2=O(\rho_n^2)=o(\rho_n),~i=0,1,2$; and (ii) since $\hat{T}_{x^ny^n}\in\calA_n'(P_{XY})$, we have $O\big(\|\hat{T}_{x^ny^n}-P_{XY}\|^2\big)=O(\rho_n^2)=o(\rho_n)$.

Hence, for $n$ large enough,
\begin{align}
\nn&\Pr\left(R_{0,n}'<\rvR_0(R_{1,n}',R_{2,n}',D_1,D_2|\hat{T}_{X^nY^n}),\hat{T}_{X^nY^n}\in\calA_n'(P_{XY})\right)\\
&\leq \Pr\left(\sum_{i=1}^n\left(\jmath_{XY}(X_i,Y_i|R_1^*,R_2^*,D_1,D_2,P_{XY})-\rvR_0(R_1^*,R_2^*,D_1,D_2|P_{XY})\right)>n\left(\theta \rho_n+o\left(\rho_n\right)\right)\right)\label{termendsec}.
\end{align}
We bound the term in \eqref{termendsec} at the end of this section. We show that this term is of the same order as that in \eqref{eqn:weiss}. Hence, the  moderate deviations constant  is lower bounded by ${\theta^2\log e}/{(2\mathrm{V}(R_1^*,R_2^*,D_1,D_2|P_{XY}) )}$.

\subsection{Converse}
Recall the definitions of $R_{i,n}$ in \eqref{defmdcrin} for $i=0,1,2$. To prove the converse part, we first define 
\begin{align}
\rho_{0,n}'&:=\theta_0\rho_n+\frac{(|\calX|\cdot|\calY|+2)\log (n+1)}{n},\\
R_{0,n}'&:=R_{0,n}+\frac{(|\calX|\cdot|\calY|+2)\log n}{n}=R_0^*+\rho_{0,n}'.
\end{align}
In a similar manner as the proof of Lemma \ref{lowerboundexcessp} in Appendix \ref{prooflbexcess}, we can show that
\begin{align}
\epsilon_n(D_1,D_2)
&\geq \left(1-\frac{1}{n}\right)\sum_{Q_{XY}\in\calP_{n}(\calX\times\calY):(R_{0,n}',R_{1,n},R_{2,n})\notin\calR(D_1,D_2|Q_{XY})}P_{XY}^n(\calT_{Q_{XY}})\\
&\geq \frac{1}{2}\sum_{Q_{XY}\in\calP_{n}(\calX\times\calY):R_{0,n}'<\rvR_0(R_{1,n},R_{2,n},D_1,D_2|Q_{XY})} P_{XY}^n(\calT_{Q_{XY}})\\
&\geq 
\frac{1}{2}\sum_{Q_{XY}\in\calP_{n}(\calX\times\calY):R_0^*+\rho_{0,n}'<\rvR_0(R_{1,n},R_{2,n},D_1,D_2|Q_{XY})} P_{XY}^n(\calT_{Q_{XY}})\\
&=\frac{1}{2}\Pr\left(R_0^*+\rho_{0,n}'<\rvR_0(R_{1,n},R_{2,n},D_1,D_2|\hat{T}_{X^nY^n})\right).
\end{align}
Applying Taylor's expansion for $\hatT_{x^ny^n}\in\calA_n'(P_{XY})$ (this typical set was defined in \eqref{eqn:typ}) in a similar manner as \eqref{assumptionuse}, we obtain
\begin{align}
\nn&\rvR_0(R_{1,n},R_{2,n},D_1,D_2| \hatT_{x^ny^n})-\left(R_0^*+\rho_{0,n}'\right)\\*
&=-\left(\theta_0+\lambda_1^*\theta_1+\lambda_2^*\theta_2\right)\rho_n+\frac{1}{n}\sum_{i=1}^n\jmath_{XY}(x_i,y_i|R_1^*,R_2^*,D_1,D_2,P_{XY})-\rvR_0(R_1^*,R_2^*,D_1,D_2|P_{XY})+o\left(\rho_n\right).
\end{align}
Hence,
\begin{align}
\epsilon_n(D_1,D_2)
&\geq \frac{1}{2}\Pr\left(R_0^*+\rho_{0,n}'<\rvR_0(R_{1,n},R_{2,n},D_1,D_2|\hat{T}_{X^nY^n}),\hat{T}_{X^nY^n}\in\calA_n'(P_{XY})\right)\\*
\nn&\geq \frac{1}{2}\Pr\left(\sum_{i=1}^n\left(\jmath_{XY}(X_i,Y_i|R_1^*,R_2^*,D_1,D_2,P_{XY})-\rvR_0(R_1^*,R_2^*,D_1,D_2|P_{XY})\right)>n\left(\theta \rho_n+o\left(\rho_n\right)\right)\right)\\*
&\qquad -\frac{1}{2}\Pr\left(\hat{T}_{X^nY^n}\notin\calA_n'(P_{XY})\right),\label{eqn:lower_mdp}
\end{align}  
where \eqref{eqn:lower_mdp} follows from the same reasoning as \eqref{eqn:rev_union_bd}. Note that the second term in \eqref{eqn:lower_mdp} is of the same order as \eqref{eqn:weiss}.

Invoking \cite[Theorem 3.7.1]{dembo2009large} and the fact that $\rho_n\to 0$,  
\begin{align}
\nn&\lim_{n\to\infty}-\frac{1}{n\rho_n^2}\log\Pr\left(\sum_{i=1}^n\left(\jmath_{XY}(X_i,Y_i|R_1^*,R_2^*,D_1,D_2,P_{XY})-\rvR_0(R_1^*,R_2^*,D_1,D_2|P_{XY})\right)>n\left(\theta \rho_n+o\left(\rho_n\right)\right)\right) \\
&=\frac{\theta^2\log e}{2\mathrm{V}(R_1^*,R_2^*,D_1,D_2|P_{XY})}.
\end{align}
Note that this calculation applies to both \eqref{termendsec} and \eqref{eqn:lower_mdp}. This completes the proof.

\section{Conclusion}
\label{conclusion}
In this paper, we derived the second-order coding region, the error exponent and moderate deviations constant for the discrete lossy Gray-Wyner problem under   mild conditions on the source. In general, it is not easy to calculate the second-order coding region but we provide an example where the second-order region calculation can be simplified. The proofs make use of a  novel type covering lemma that is suited to the discrete lossy Gray-Wyner problem. We also establish new results on the  uniform continuity of conditional rate-distortion function that may be of independent interest elsewhere. We hope the solution to this problem may lead to the solution of second-order regions for other multi-terminal lossy source coding problems \cite{watanabe2015,el2011network}.  

In the future, we aim to solve for the second-order asymptotics of  the lossy Gray-Wyner problem with correlated Gaussian sources and the quadratic  distortion measure~\cite[Example~2.5(B)]{gray1974source}.  Lastly, we hope to derive the exact asymptotics for this problem~\cite{altug14a,altug14b}.

\appendix
\subsection{Proof of Lemma \ref{propertytilted}}
\label{prooflemmatilted}
For given test channel $P_{W|XY}P_{\hatX|XW}P_{\hatY|YW}$, let $P_{XYW}$, $P_{XW}, P_{YW}$ be induced joint and marginal distributions. 
For any $Q_{W}\in\calP(\calW)$, define
\begin{align}
&\nn F(P_{W|XY},Q_{W},Q_{\hat{X}|W},Q_{\hat{Y}|W},D_1,D_2)\\
\nn&:=D(P_{W|XY}\|Q_W|P_{XY})+\lambda_1^*\left(D(P_{\hatX|XW}\|Q_{\hatX|W}|P_{XW})-R_1^*\right)+\lambda_2^*\left(D(P_{\hatY|YW}\|Q_{\hatY|W}|P_{YW})-R_2^*\right)\\
&+\gamma_1^*(\mathbb{E}[d_X(X,\hatX)-D_1])+\gamma_2^*(\mathbb{E}[d_Y(Y,\hatY)]-D_2)\\
\nn&=I(X,Y;W)+D(P_W\|Q_W)+\lambda_1^*\left(I(X;\hatX|W)+D(P_{\hatX|W}\|Q_{\hatX|W}|P_W)-R_1^*\right)\\
&\qquad+\lambda_2^*\left(I(Y;\hatY|W)+D(P_{\hatY|W}\|Q_{\hatY|W}|P_W)-R_2^*\right)+\gamma_1^*(\mathbb{E}[d_X(X,\hatX)-D_1])+\gamma_2^*(\mathbb{E}[d_Y(Y,\hatY)]-D_2).
\end{align}
We can relate $\rvR_0(R_1^*,R_2^*,D_1,D_2|P_{XY})$ to $F(P_{W|XY},Q_{W},Q_{\hat{X}|W},Q_{\hat{Y}|W})$ as follows:
\begin{align}
\label{linkrf}
\rvR_0(R_1^*,R_2^*,D_1,D_2|P_{XY})=\min_{P_{W|XY}P_{\hatX|XW}P_{\hatY|YW}}\min_{Q_W}\min_{Q_{\hat{X}|W}}\min_{Q_{\hat{Y}|W}}F(P_{W|XY},P_{\hatX|XW},P_{\hatY|YW},Q_{W},Q_{\hat{X}|W},Q_{\hat{Y}|W},D_1,D_2).
\end{align}

For given $Q_W$, $Q_{\hat{X}|W}$, $Q_{\hat{Y}|W}$, $\lambda_1,\lambda_2>0$ and $\gamma_1,\gamma_2\geq 0$, define
\begin{align}
&\Lambda(x,w,\hatx,Q_{\hatX|W}|\lambda_1,\gamma_1)
:=\log \frac{1}{\sum_{\hatx} Q_{\hatX|W}(\hatx|w)\exp\Big(\frac{\gamma_1}{\lambda_1}(D_1-d_X(x,\hatx))\Big)}\label{def:jmathx}\\
&\Lambda(y,w,\haty,Q_{\hatY|W}|\lambda_2,\gamma_2)
:=\log \frac{1}{\sum_{\haty} Q_{\hatY|W}(\haty|w)\exp\Big(\frac{\gamma_2}{\lambda_2}(D_2-d_Y(y,\haty))\Big)}\label{def:jmathy}\\
\nn&\Lambda(x,y|Q_W,Q_{\hat{X}|W},Q_{\hat{Y}|W},\lambda_1,\lambda_2,\gamma_1,\gamma_2)\\
&:=\log \frac{1}{\sum_{w}Q_W(w)\exp\Big(\lambda_1(R_1-\jmath(x,w,\hatx,Q_{\hatX|W}|\lambda_1,\gamma_1))+\lambda_2(R_2-\jmath(y,w,\haty,Q_{\hatY|W}|\lambda_2,\gamma_2))\Big)}\label{def:jmathxy}
\end{align}
We show the relationship between $F(P_{W|XY},Q_{W},Q_{\hat{X}|W},Q_{\hat{Y}|W})$ and $\Lambda(x,y|Q_W,Q_{\hat{X}|W},Q_{\hat{Y}|W},\lambda_1,\lambda_2)$ in the following Lemma.
\begin{lemma}
\label{linkflambda}
For any $Q_W,Q_{\hat{X}|W},Q_{\hat{Y}|W}$,
\begin{align}
\!\!\!\min_{P_{W|XY}\!\!\!P_{\hatX|XW}P_{\hatY|YW}}F(P_{W|XY},P_{\hatX|XW},P_{\hatY|YW},Q_{W},Q_{\hat{X}|W},Q_{\hat{Y}|W},D_1,D_2)=\mathbb{E}_{P_{XY}}\left[\Lambda(X,Y|Q_W,Q_{\hat{X}|W},Q_{\hat{Y}|W},\lambda_1^*,\lambda_2^*)\right],
\end{align}
where the minimization is achieved by $P_{W|XY}P_{\hatX|XW}P_{\hatY|YW}$ s.t. 
\begin{align}
P_{\hatX|XW}&=\frac{Q_{\hatX|W}(\hatx|w)\exp\Big(\frac{\gamma_1^*}{\lambda_1^*}(D_1-d_X(x,\hatx))\Big)}{\sum_{\hatx} Q_{\hatX|W}(\hatx|w)\exp\Big(\frac{\gamma_1^*}{\lambda_1^*}(D_1-d_X(x,\hatx))\Big)},\\
P_{\hatY|YW}&=\frac{Q_{\hatY|W}(\haty|w)\exp\Big(\frac{\gamma_2^*}{\lambda_2^*}(D_2-d_Y(y,\haty))\Big)}{\sum_{\haty} Q_{\hatY|W}(\haty|w)\exp\Big(\frac{\gamma_2^*}{\lambda_2^*}(D_2-d_Y(y,\haty))\Big)},\\
P_{W|XY}(w|xy)&=\frac{Q_W(w)\exp\Big(\lambda_1^*(R_1-\Lambda(x,w,\hatx,Q_{\hatX|W}|\lambda_1^*,\gamma_1^*))+\lambda_2^*(R_2-\Lambda(y,w,\haty,Q_{\hatY|W}|\lambda_2^*,\gamma_2^*))\Big)}{\sum_{w}Q_W(w)\exp\Big(\lambda_1^*(R_1-\Lambda(x,w,\hatx,Q_{\hatX|W}|\lambda_1^*,\gamma_1^*))+\lambda_2^*(R_2-\Lambda(y,w,\haty,Q_{\hatY|W}|\lambda_2^*,\gamma_2^*))\Big)}.
\end{align}
\end{lemma}
\begin{proof}
The proof is similar to the proof of Lemma 1 in \cite{watanabe2015second}. Invoking the log-sum inequality, we obtain
\begin{align}
&\nn F(P_{W|XY},P_{\hatX|XW},P_{\hatY|YW},Q_{W},Q_{\hat{X}|W},Q_{\hat{Y}|W},D_1,D_2)\\
\nn &=\sum_{x,y,w,}P_{XY}(x,y)P_{W|XY}(w|xy)\log\frac{P_{W|XY}(w|xy)}{Q_{W}(w)}\\
\nn &\qquad+\sum_{\hatx} P_{\hatX|XW}(\hatx|x,w)\Bigg(\lambda_1^*\left(\log \frac{P_{\hatX|XW}(\hatx|x,w)}{Q_{\hatX|W}(\hatx|w)}-R_1^*\right)+\gamma_1^*(d_X(x,\hatx)-D_1)\Bigg)\\
&\qquad +\sum_{\haty}P_{\hatY|YW}(\haty|y,w)\Bigg(\lambda_2^*\left(\log \frac{P_{\hatY|YW}(\hatx|x,w)}{Q_{\hatY|W}(\haty|w)}-R_2^*\right)+\gamma_2^*(d_Y(y,\haty)-D_2)\Bigg)\\
\nn&\geq\sum_{x,y,w}P_{XY}(x,y)P_{W|XY}(w|xy))\log\frac{P_{W|XY}(w|xy)}{Q_{W}(w)}+\lambda_1^*\log \frac{1}{\sum_{\hatx} Q_{\hatX|W}(\hatx|w)\exp\Big(\frac{\gamma_1^*}{\lambda_1^*}(D_1-d_X(x,\hatx))\Big)}-\lambda_1^*R_1\\
&\qquad+\lambda_2^*\log \frac{1}{\sum_{\haty} Q_{\hatY|W}(\haty|w)\exp\Big(\frac{\gamma_2^*}{\lambda_2^*}(D_2-d_Y(y,\haty))\Big)}-\lambda_2^*R_2\\
&\geq \sum_{x,y}P_{XY}(x,y)\log \frac{1}{\sum_{w}Q_W(w)\exp\Big(\lambda_1^*(R_1-\Lambda(x,w,\hatx,Q_{\hatX|W}|\lambda_1^*,\gamma_1^*))+\lambda_2^*(R_2-\Lambda(y,w,\haty,Q_{\hatY|W}|\lambda_2^*,\gamma_2^*))\Big)}\label{logsumfinal1}\\
&=\mathbb{E}[\Lambda(X,Y|Q_W,Q_{\hatX|W},Q_{\hatY|W},\lambda_1^*,\lambda2^*,\gamma_1^*,\gamma_2^*)]\label{logsumfinal2},
\end{align}  
where \eqref{logsumfinal1} follows from \eqref{def:jmathx} and \eqref{def:jmathy} while \eqref{logsumfinal2} follows from \eqref{def:jmathxy}.
\end{proof}
Invoking \eqref{linkrf}, we obtain
\begin{align}
\nn&\rvR_0(R_1^*,R_2^*,D_1,D_2|P_{XY})\\*
&=\min_{P_{W|XY}P_{\hatX|XW}P_{\hatY|YW}}
\min_{Q_W}\min_{Q_{\hat{X}|W}}\min_{Q_{\hat{Y}|W}}
F(P_{W|XY},P_{\hatX|XW},P_{\hatY|YW},Q_{W},Q_{\hat{X}|W},Q_{\hat{Y}|W},D_1,D_2)\\*
&\leq \min_{P_{W|XY}P_{\hatX|XW}P_{\hatY|YW}}F(P_{W|XY},P_{\hatX|XW},P_{\hatY|YW},P_{W}^*,P_{\hat{X}|W}^*,P_{\hat{Y}|W}^*,D_1,D_2)\\*
&\leq F(P_{W|XY}^*,P_{\hatX|XW}^*,P_{\hatY|YW}^*,P_{W}^*,P_{\hat{X}|W}^*,P_{\hat{Y}|W}^*,D_1,D_2)\\*
&=\rvR_0(R_1^*,R_2^*,D_1,D_2|P_{XY}).
\end{align}
Invoking Lemma \ref{linkflambda}, we obtain
\begin{align}
\rvR_0(R_1^*,R_2^*,D_1,D_2|P_{XY})&=\mathbb{E}_{P_{XY}}\left[\Lambda(X,Y|P^*_W,P^*_{\hat{X}|W},P^*_{\hat{Y}|W},\lambda_1^*,\lambda_2^*)\right],
\end{align}
and for $(w,\hatx,\haty)$ s.t. $P_W^*(w)P_{\hatX|W}^*(\hatx|w)P_{\hatY|W}^*(\haty|w)>0$,
\begin{align}
&\Lambda(x,y|P^*_W,P^*_{\hat{X}|W},P^*_{\hat{Y}|W},\lambda_1^*,\lambda_2^*)\\
&=\log\frac{P_{W|XY}^*(w|xy)}{P_{W}^*(w)}+\lambda_1^*\left(\Lambda(x,w,\hatx,P_{\hatX|W}^*|\lambda_1^*,\gamma_1^*)-R_1^*\right)+\lambda_2^*\left(\Lambda(y,w,\haty,P_{\hatY|W}^*|\lambda_2^*,\gamma_2^*)-R_2^*\right)\\
\nn&=\log\frac{P_{W|XY}^*(w|xy)}{P_{W}^*(w)}+\lambda_1^*\Bigg(\log \frac{P_{\hatX|XW}^*(\hatx|x,w)}{P_{\hatX|W}^*(\hatx|w)}+\frac{\gamma_1^*}{\lambda_1^*}(d_X(x,\hatx)-D_1-R_1^*)\Bigg)\\
&\qquad\qquad+\lambda_2^*\Bigg(\log \frac{P_{\hatY|YW}^*(\haty|y,w)}{P_{\hatY|W}^*(\haty|w)}+\frac{\gamma_2^*}{\lambda_2^*}(d_Y(y,\haty)-D_2)-R_2^*\Bigg)\\
\nn&=\log \frac{P_{W|XY}^*(w|xy)}{P_W^*(w)}+\lambda_1^*\log\frac{P_{\hatX|XW}^*(\hatx|x,w)}{P_{\hatX|W}^*(\hatx|w)}-\lambda_1^*R_1^*+\lambda_2^*\log\frac{P_{\hatY|YW}^*(\haty|y,w)}{P_{\hatY|W}^*(\haty|w)}-\lambda_2^*R_2^*\\
&\qquad+\gamma_1^*(d_X(x,\hatx)-D_1)+\gamma_2^*(d_Y(y,\haty)-D_2).
\label{averageobjrate}
\end{align}

Note that
\begin{align}
\jmath_{XY}(x,y|R_1^*,R_2^*,D_1,D_2,P_{XY}):=\Lambda(x,y|P^*_W,P^*_{\hat{X}|W},P^*_{\hat{Y}|W},\lambda_1^*,\lambda_2^*).
\end{align}
Hence, according to the properties of $\Lambda(x,y|P^*_W,P^*_{\hat{X}|W},P^*_{\hat{Y}|W},\lambda_1^*,\lambda_2^*)$, the properties of $\jmath_{XY}(x,y|R_1^*,R_2^*,D_1,D_2)$ in Lemma \ref{propertytilted} are proved.

\subsection{Proof of Lemma \ref{linkrjxy}}
\label{prooflinkrjxy}
Recall $Q_{W|XY}^*Q_{\hatX|XW}^*Q_{\hatY|YW}^*$ is an optimal test channel achieving the objective function $\rvR_0(R_1^*,R_2^*,D_1,D_2|\Gamma(Q_{XY}))$ (see \eqref{minkey2}) and  $Q_{W}^*, Q_{\hatX|W}^*,Q_{\hatY|W}^*$ are the induced (conditional) distributions. Invoking Lemma \ref{propertytilted}, we obtain
\begin{align}
\rvR_0(R_1^*,R_2^*,D_1,D_2|\Gamma(Q_{XY}))
&=\sum_{k=1}^m \Gamma_k(Q_{XY})\jmath_{XY}(k|R_1^*,R_2^*,D_1,D_2,\Gamma(Q_{XY})).
\end{align}
Hence,
\begin{align}
\frac{\partial \rvR_0(R_1^*,R_2^*,D_1,D_2|\Gamma(Q_{XY}))}{\partial \Gamma_i(Q_{XY})}\bigg|_{Q_{XY}=P_{XY}}
\nn&=\jmath_{XY}(i|R_1^*,R_2^*,D_1,D_2,\Gamma(P_{XY}))-\jmath_{XY}(m|R_1^*,R_2^*,D_1,D_2,\Gamma(P_{XY}))\\
&\qquad+\frac{\partial }{\partial \Gamma_i(Q_{XY})}\Bigg(\sum_{k=1}^m \Gamma_k(P_{XY})\jmath_{XY}(k|R_1^*,R_2^*,D_1,D_2,\Gamma(Q_{XY}))\Bigg)\label{derivativeb}.
\end{align}
We now focus on the third term in \eqref{derivativeb}. For $t=1,2$, denote
\begin{align}
\lambda_{t,Q}^*&=-\frac{\partial \rvR_0(R_1,R_2,D_1,D_2|\Gamma(Q_{XY}))}{\partial R_i}\bigg|_{(R_1,R_2)=(R_1^*,R_2^*)},\\
\gamma_{t,Q}^*&=-\frac{\partial \rvR_0(R_1,R_2,D_1',D_2'|\Gamma(Q_{XY}))}{\partial D_i'}\bigg|_{(D_1',D_2')=(D_1,D_2)}.
\end{align}
Invoking Lemma \ref{propertytilted}, we obtain
\begin{align}
\nn&\sum_{k=1}^m \Gamma_k(P_{XY})\jmath_{XY}(k|R_1^*,R_2^*,D_1,D_2,\Gamma(Q_{XY}))\\
&=\sum_{k=1}^m P_{XY}(x_k,y_k)\jmath_{XY}(k|R_1^*,R_2^*,D_1,D_2,\Gamma(Q_{XY}))\label{useparametric},\\
&\nn=\sum_{k=1}^m P_{XY}(x_k,y_k)\sum_{w,\hatx,\haty} P_{W|XY}^*(w|x_k,y_k)P_{\hatX|XW}^*(\hatx|x_k,w)P_{\hatY|YW}^*(\haty|y_k,w)\Bigg[\log\frac{Q_{W|XY}^*(w|x_k,y_k)}{Q_{W}^*(w)}\\
\nn&\qquad+\lambda_{1,Q}^*\Bigg(\log\frac{Q_{\hatX|XW}^*(\hatx|x_k,w)}{Q_{\hatX|W}^*(\hatx|w)}-R_1^*\Bigg)+\lambda_{2,Q}^*\left(\log\frac{Q_{\hatY|YW}^*(\haty|y_k,w)}{Q_{\hatY|W}^*(\haty|w)}-R_2^*\right)\vphantom{[\log\frac{Q_{W|XY}^*(w|xy)}{Q_{W}^*(w)}}\\
&\qquad+\gamma_{1,Q}^*(d_X(x_k,\hatx)-D_1)+\gamma_{2,Q}^*(d_Y(y_k,\haty)-D_2)\Bigg]\label{expectpq}.
\end{align}
where \eqref{useparametric} follows from the notation introduced in Section \ref{sec:notation} and where $P_{W|XY}^*P_{\hatX|XW}^*P_{\hatY|YW}^*$ achieve $\rvR_0(R_1^*,R_2^*,D_1,D_2|\Gamma(P_{XY}))$. If we let $Q_{XY}=P_{XY}$, then $Q_{W|XY}^*=P_{W|XY}^*$, $Q_{\hat{X}|XW}^*=P_{\hat{X}|XW}^*$ and $Q_{\hat{Y}|YW}^*=P_{\hat{Y}|YW}^*$. In the following, for ease of notation, we will use $Q=P$ for all the above relations.

We claim that
\begin{align}
\frac{\partial }{\partial \Gamma_i(Q_{XY})}\Bigg(\sum_{k=1}^m \Gamma_k(P_{XY})\jmath_{XY}(k|R_1^*,R_2^*,D_1,D_2,\Gamma(Q_{XY}))\Bigg)
&=0\label{derivetiveb2}.
\end{align}

Equation \eqref{derivetiveb2} holds for the following reasons:
\begin{itemize}
\item In a similar manner as \cite[Theorem 2.2]{kostina2013lossy} and noting that $\Gamma_m(Q_{XY})=1-\sum_{i=1}^{m-1}\Gamma_i(Q_{XY})$ and $\Gamma_i(Q_{XY})=Q_{XY}(x_i,y_i)$, we obtain
\begin{align}
\nn&\frac{\partial }{\partial \Gamma_i(Q_{XY})}\Bigg(\sum_{k=1}^m P_{XY}(x_k,y_k)\sum_w P_{W|XY}^*(w|x_k,y_k) \log \frac{Q_{W|XY}^*(w|x_k,y_k)}{Q_W^*(w)}\Bigg)\Bigg|_{Q=P}\\
\nn&=\frac{\partial }{\partial \Gamma_i(Q_{XY})}\Bigg(\sum_{k=1}^m P_{XY}(x_k,y_k)\sum_w P_{W|XY}^*(w|x_k,y_k) \Big(\log Q_{XY|W}^*(x_k,y_k|w)-\log Q_{XY}(x_k,y_k)\Big)\Bigg)\Bigg|_{Q=P}\\
\nn&=\frac{\partial }{\partial \Gamma_i(Q_{XY})}\Bigg\{\sum_w \sum_{k=1}^m P_W^*(w)P_{XY|W}^*(x_k,y_k|w) \Bigg(\frac{Q_{XY|W}^*(x_k,y_k|w)}{P_{XY|W}^*(x_k,y_k|w)}\Bigg)\Bigg\}\Bigg|_{Q=P}\log e\\
&\qquad-\frac{\partial }{\partial \Gamma_i(Q_{XY})}\Bigg(-\sum_{k=1}^{m-1}P_{XY}(x_k,y_k)\log Q_{XY}(x_k,y_k)-P_{XY}(x_m,y_m)\log Q_{XY}(x_m,y_m)\Bigg)\\
&=0\label{derivecommon}.
\end{align}
\item In a similar manner as \cite[(64)-(70)]{watanabe2015second}, we obtain
\begin{align}
\nn&\frac{\partial }{\partial \Gamma_i(Q_{XY})}\Bigg(\sum_{k=1}^m P_{XY}(x_k,y_k)\sum_{w,\hatx}P_{W|XY}^*(w|x_k,y_k)P_{\hatX|XW}^*(\hatx|x_k,w)\lambda_{1,Q}^*\left(\log
\frac{Q_{\hatX|XW}^*(\hatx|x_k,w)}{Q_{\hatX|W}^*(\hatx|w)}-R_1^*\right)\Bigg)\\
\nn&=\frac{\partial \lambda_{1,Q}^*}{\partial \Gamma_i(Q_{XY})}\Bigg|_{Q=P}\sum_{k=1}^m P_{XY}(x_k,y_k)\sum_{w,\hatx}P_{W|XY}^*(w|x_k,y_k)P_{\hatX|XW}^*(\hatx|x_k,w)\left(\log\frac{P_{\hatX|XW}^*(\hatx|x_k,w)}{P_{\hatX|W}^*(\hatx|w)}-R_1^*\right)\\
&\qquad+\lambda_{1,Q}^*\sum_{k=1}^m P_{XY}(x_k,y_k)\sum_{w,\hatx}P_{W|XY}^*(w|x_k,y_k)P_{\hatX|XW}^*(\hatx|x_k,w)\frac{\partial }{\partial \Gamma_i(Q_{XY})}\left(\log\frac{Q_{\hatX|XW}^*(\hatx|x_k,w)}{Q_{\hatX|W}^*(\hatx|w)}\right)\Bigg|_{Q=P}\\
&=0\label{derivelambda1q},
\end{align}
where \eqref{derivelambda1q} follows because: (i) under the optimal test channel, we have $R_1^*=I(X;\hatX|W)$; (ii) the second term in \eqref{derivelambda1q} equals to 0, which results from a similar manner to \eqref{derivecommon}.

Symmetrically, we have
\begin{align}
\nn&\frac{\partial }{\partial \Gamma_i(Q_{XY})}\Bigg(\sum_{k=1}^m P_{XY}(x_k,y_k)\sum_{w,\haty}P_{W|XY}^*(w|x_k,y_k)P_{\hatY|YW}^*(\haty|y_k,w)\lambda_{2,Q}^*\left(\log
\frac{Q_{\hatY|YW}^*(\haty|y_k,w)}{Q_{\hatY|W}^*(\haty|w)}-R_2^*\right)\Bigg)\\*
&=0.
\end{align}
\item Under the optimal test channel, $\mathbb{E}[d_X(X,\hatX)]=D_1$ and $\mathbb{E}[d_Y(Y,\hatY)]=D_2$. Thus, the last two terms in \eqref{expectpq} are zero.
\end{itemize}

\subsection{Proof of Lemma \ref{propertyjointrd}}
\label{proofpropertyjointrd}
We offer the proof resembling \cite[Lemma 1]{watanabe2015second}. Note that $R_{XY}(P_{XY},D_1,D_2)$ is convex and non-increasing in $(D_1,D_2)$.
Hence,
\begin{align}
R_{XY}(P_{XY},D_1,D_2)&=\inf_{P_{\hat{X}\hat{Y}|XY}:\mathbb{E}[d_{X}(X,\hat{X})]\leq D_1,~\mathbb{E}[d_{Y}(Y,\hat{Y})]\leq D_2} I(XY;\hat{X}\hat{Y})
\end{align}
is a convex optimization problem. The dual problem is given by 
\begin{align}
\label{dualproblem}
R_{XY}(P_{XY},D_1,D_2)=\max_{\nu_1\geq 0,\nu_2\geq 0} \inf_{P_{\hat{X}\hat{Y}|XY}} I(XY;\hat{X}\hat{Y})+\nu_1(\mathbb{E}[d_{X}(X,\hat{X})]-D_1)+\nu_2(\mathbb{E}[d_{Y}(Y,\hat{Y})]-D_2).
\end{align}
Hence, the dual optimal values are $\nu_1^*$ in \eqref{defjrdnu1} and $\nu_2^*$ in \eqref{defjrdnu2}.

Given $Q_{\hat{X}\hat{Y}}$, define
\begin{align}
\nn&F(P_{\hat{X}\hat{Y}|XY},Q_{\hat{X}\hat{Y}},D_1,D_2)\\
&:=D(P_{\hat{X}\hat{Y}|XY}\|Q_{\hat{X}\hat{Y}}|P_{XY})+\nu_1^*\left(\mathbb{E}[d_{X}(X,\hat{X})]-D_1\right)+\nu_2^*\left(\mathbb{E}[d_{Y}(Y,\hat{Y})]-D_2\right)\\
&=I(XY;\hat{X}\hat{Y})+D(P_{\hat{X}\hat{Y}}\|Q_{\hat{X}\hat{Y}})+\nu_1^*\left(\mathbb{E}[d_{X}(X,\hat{X})]-D_1\right)+\nu_2^*\left(\mathbb{E}[d_{Y}(Y,\hat{Y})]-D_2\right).
\end{align}
Then considering the dual problem of $R_{XY}(P_{XY},D_1,D_2)$ in \eqref{dualproblem}, we obtain
\begin{align}
R_{XY}(P_{XY},D_1,D_2)=\inf_{Q_{\hat{X}\hat{Y}}}\inf_{P_{\hat{X}\hat{Y}|XY}}F(P_{\hat{X}\hat{Y}|XY},Q_{\hat{X}\hat{Y}},D_1,D_2).
\end{align}
For $\nu_1>0$ and $\nu_2>0$, define
\begin{align}
\Lambda(x,y|Q_{\hat{X}\hat{Y}},\nu_1,\nu_2)
:=-\log \mathbb{E}_{Q_{\hat{X}\hat{Y}}}\left[\exp\left(\nu_1(D_1-d_{X}(x,\hat{X}))-\nu_2 (D_2-d_{Y}(y,\hat{Y}))\right)\right].
\end{align}
We can relate $F(P_{\hat{X}\hat{Y}|XY},Q_{\hat{X}\hat{Y}},D_1,D_2)$ with $\Lambda(x,y|Q_{\hat{X}\hat{Y}},\nu_1,\nu_2)$ as follows.
\begin{align}
\min_{P_{\hat{X}\hat{Y}|XY}}F(P_{\hat{X}\hat{Y}|XY},Q_{\hat{X}\hat{Y}},D_1,D_2)
&=\mathbb{E}_{P_{XY}}\left[\Lambda(X,Y|Q_{\hat{X}\hat{Y}},\nu_1^*,\nu_2^*)\right],
\end{align}
where the minimization is achieved by the optimal test channel 
\begin{align}
P_{\hat{X}\hat{Y}|XY}^{*(Q_{\hat{X}\hat{Y}})}(\hat{x},\hat{y}|x,y)
&=Q_{\hat{X}\hat{Y}}(\hat{x},\hat{y})\exp\left(\Lambda(x,y|Q_{\hat{X}\hat{Y}},\nu_1^*,\nu_2^*)+\nu_1^*(D_1-d_{X}(x,\hatx))-\nu_2^* (D_2-d_Y(y,\haty))\right).
\end{align}
\begin{proof}
Invoking the log-sum inequality, we have
\begin{align}
\nn&F(P_{\hat{X}\hat{Y}|XY},Q_{\hat{X}\hat{Y}},D_1,D_2)\\
&=\sum_{x,y,\hat{x},\hat{y}}P_{XY}(x,y)P_{\hat{X}\hat{Y}|XY}(\hat{x},\hat{y}|x,y)\left(\log\frac{P_{\hat{X}\hat{Y}|XY}(\hat{x},\hat{y}|x,y)}{Q_{\hat{X}\hat{Y}}(\hat{x},\hat{y})}+\nu_1^*(d_{X}(x,\hatx)-D_1)+\nu_2^*(d_{Y}(y,\hat{y})-D_2)\right)\\
&=\sum_{x,y}P_{XY}(x,y)\sum_{\hat{x},\hat{y}}P_{\hat{X}\hat{Y}|XY}(\hat{x},\hat{y}|x,y)\log\frac{P_{\hat{X}\hat{Y}|XY}(\hat{x},\hat{y}|x,y)}{Q_{\hat{X}\hat{Y}}(\hat{x},\hat{y})\exp\left(\nu_1(D_1-d_{X}(x,\hatx))-\nu_2 (D_2-d_{Y}(y,\haty))\right)}\\
&\geq \sum_{x,y}P_{XY}(x,y)\log\frac{\sum_{\hat{x},\hat{y}}P_{\hat{X}\hat{Y}|XY}(\hat{x},\hat{y}|x,y)}{\sum_{\hat{x},\hat{y}}Q_{\hat{X}\hat{Y}}(\hat{x},\hat{y})\exp\left(\nu_1^*(D_1-d_{X}(x,\hatx))-\nu_2^* (D_2-d_{Y}(y,\haty))\right)}\\
&=\mathbb{E}_{P_{XY}}\left[\Lambda(X,Y|Q_{\hat{X}\hat{Y}},\nu_1^*,\nu_2^*)\right],
\end{align}
with equality if and only if $P_{\hat{X}\hat{Y}|XY}$ is $P_{\hat{X}\hat{Y}|XY}^{*(Q_{\hat{X}\hat{Y}})}$.
\end{proof}

Let $P_{\hat{X}\hat{Y}|XY}^*$ be the optimal channel achieving $R_{XY}(P_{XY},D_1,D_2)$ and $P_{\hat{X}\hat{Y}}^*$ be induced by $P_{\hat{X}\hat{Y}|XY}^*$ and $P_{XY}$. Then,
\begin{align}
R_{XY}(P_{XY},D_1,D_2)
&=\inf_{Q_{\hat{X}\hat{Y}}}\inf_{P_{\hat{X}\hat{Y}|XY}}F(P_{\hat{X}\hat{Y}|XY},Q_{\hat{X}\hat{Y}},D_1,D_2)\\
&\leq \inf_{Q_{\hat{X}\hat{Y}}}F(P_{\hat{X}\hat{Y}|XY}^*,Q_{\hat{X}\hat{Y}},D_1,D_2)\\
&=F(P_{\hat{X}\hat{Y}|XY}^*,P_{\hat{X}\hat{Y}}^*,D_1,D_2)\\
&=\mathbb{E}_{P_{XY}}\left[\Lambda(X,Y|P_{\hat{X}\hat{Y}}^*,\nu_1^*,\nu_2^*)\right]\\
&=R_{XY}(P_{XY},D_1,D_2).
\end{align}
Hence, we obtain
\begin{align}
P_{\hat{X}\hat{Y}|XY}^*(\hat{x},\hat{y}|x,y)
&=P_{\hat{X}\hat{Y}}^*(\hat{x},\hat{y})\exp\left(\Lambda(x,y|P_{\hat{X}\hat{Y}}^*,\nu_1^*,\nu_2^*)+\nu_1^*(D_1-d_{X}(x,\hat{X}))-\nu_2^* (D_2-d_{Y}(y,\hat{Y}))\right),
\end{align}
i.e.,
\begin{align}
\Lambda(x,y|P_{\hat{X}\hat{Y}}^*,\nu_1^*,\nu_2^*)=\log \frac{P_{\hat{X}\hat{Y}|XY}^*(\hat{x},\hat{y}|x,y)}{P_{\hat{X}\hat{Y}}^*(\hat{x},\hat{y})}+\nu_1^*(d_{X}(x,\hat{X})-D_1)-\nu_2^* (d_{Y}(y,\hat{Y})-D_2).
\end{align}
Note that $\imath_{XY}(x,y|D_1,D_2,P_{XY})=\Lambda(x,y|P_{\hat{X}\hat{Y}}^*,\nu_1^*,\nu_2^*)$. The proof is now complete.

\subsection{Proof of Lemma \ref{panglosstilted}}
\label{proofpanglosstilted}
Considering a rate triplet on the Pangloss plane, i.e., $(R_0^*,R_1^*,R_2^*)\in\calR_{\mathrm{pg}}(D_1,D_2|P_{XY})$ and $R_0^*>0$, we obtain
\begin{align}
R_0^*=\rvR_0(R_1^*,R_2^*,D_1,D_2|P_{XY})=R_{XY}(P_{XY},D_1,D_2)-R_1^*-R_2^*\label{sumjointrd}.
\end{align}
Hence, $\lambda_1^*=\lambda_2^*=1$ and 
\begin{align}
\gamma_1^*=-\frac{\partial \rvR_0(R_1^*,R_2^*,D,D_2|P_{XY})}{\partial D}\bigg|_{D=D_1}
&=-\frac{\partial R_{XY}(P_{XY},D,D_2)}{\partial D}\bigg|_{D=D_1}=\nu_1^*,\\
\gamma_2^*=-\frac{\partial \rvR_0(R_1^*,R_2^*,D_1,D|P_{XY})}{\partial D}\bigg|_{D=D_2}
&=-\frac{\partial R_{XY}(P_{XY},D_1,D)}{\partial D}\bigg|_{D=D_1}=\nu_2^*.
\end{align} 
Let $\calP(P_{XY},D_1,D_2)$ be the set of all joint distributions $P_{XYW\hat{X}\hat{Y}}$ satisfying
\begin{itemize}
\item The $\calX\times\calY$-marginal is $P_{XY}$;
\item The conditional distribution $P_{\hat{X}\hat{Y}|XY}$ achieving $R_{XY}(P_{XY},D_1,D_2)$
\item The following Markov chains hold: $\hat{X}\to W\to \hat{Y}$ and $(X,Y)\to (\hat{X},\hat{Y})\to W$.
\end{itemize}
The lossy Wyner's common information \cite{viswanatha2014}, is
\begin{align}
C_{\mathrm{W}}(D_1,D_2|P_{XY})
&:=\min\left\{R_0:(R_0,R_1,R_2)\in\calR_{\mathrm{pg}}(D_1,D_2|P_{XY})\right\}\label{lossywci}\\
&=\min_{P_{XYW\hat{X}\hat{Y}}\in\calP(P_{XY},D_1,D_2)}\left\{I(X,Y;W)\right\}.
\end{align}

Denote the joint distribution achieving \eqref{lossywci} as $P_{XYW\hat{X}\hat{Y}}^*$. According to Corollary 1 in \cite{viswanatha2014}, the random variables $(XYW\hat{X}\hat{Y})$ following $P_{XYW\hat{X}\hat{Y}}^*$ satisfy the following Markov chains: $\hat{X}\to(X,Y,W)\to \hat{Y}$, $\hat{X}\to(X,W)\to Y$ and $\hat{Y}\to (Y,W) \to X$. All the distributions used in this proof are marginals of $P_{XYW\hat{X}\hat{Y}}^*$. Invoking Lemma \ref{propertytilted}, we obtain that for every $(w,\hatx,\haty)$ such that $P_W^*(w)P_{\hatX|W}^*(\hatx|w)P_{\hatY|W}^*(\haty|w)>0$,
\begin{align}
&\nn\jmath_{XY}(x,y|R_1^*,R_2^*,D_1,D_2,P_{XY})\\
&=\log \frac{P_{W|XY}^*(w|xy)}{P_{W}^*(w)}+\log\frac{P_{\hatX|XW}^*(\hatx|x,w)}{P_{\hatX|W}^*(\hatx|w)}+\log\frac{P_{\hatY|YW}^*(\haty|y,w)}{P_{\hatY|W}^*(\haty|w)}-R_1^*-R_2^*+\gamma_1^*(d_X(x,\hatx)-D_1)+\gamma_2^*(d_Y(y,\haty)-D_2)\\
&=\log \frac{P_{W|XY}^*(w|xy)P_{\hat{X}|XW}^*(\hat{x}|xw)P_{\hat{Y}|YW}^*(\hat{y}|yw)}{P_{W}^*(w)P_{\hat{X}|W}^*(\hat{x}|w)P_{\hat{Y}|W}^*(\hat{y}|w)}+\nu_1^*(d_{X}(x,\hat{x})-D_1)+\nu_2^*(d_{Y}(y,\hat{y})-D_2)-R_1^*-R_2^*\\
&=\log \frac{P_{\hat{X}\hat{Y}|XY}^*(\hat{x},\hat{y}|x,y)}{P_{\hat{X}\hat{Y}}^*(\hat{x},\hat{y})}+\nu_1^*(d_{X}(x,\hat{x})-D_1)+\nu_2^*(d_{Y}(y,\hat{y})-D_2)-R_1^*-R_2^*\label{markovchains}\\
&=\imath_{XY}(x,y|D_1,D_2,P_{XY})-R_1^*-R_2^*,
\end{align}
where \eqref{markovchains} follow from the Markov chains implied by $P_{XYW\hat{X}\hat{Y}}^*$ in \cite{viswanatha2014}.

\subsection{Justification of Remark \ref{remarkprop} for $D_1=0$ and $D_2>0$}
\label{justifyremark}
For $D_1=0$ and $D_2>0$, the joint rate-distortion function is
\begin{align}
R_{XY}(P_{XY},0,D_2)
&=\min_{P_{\hat{X}\hat{Y}|XY}:\mathbb{E}[d_{X}(X,\hat{X})]\leq 0,~\mathbb{E}[d_{Y}(Y,\hat{Y})]\leq D_2} I(XY;\hat{X}\hat{Y})\\
&=\min_{P_{\hat{Y}|XY}:\mathbb{E}[d_{Y}(Y,\hat{Y})]\leq D_2} I(XY;X\hat{Y})\\
&=H(X)+R_{Y|X}(P_{XY},D_2).
\end{align}
The properties of $R_{XY}(P_{XY},0,D_2)$ is still valid with $(\hat{x},\hat{X})$ replaced by $(x,X)$ by invoking Lemma \ref{propertytilted}. Then we need to verify that Lemma \ref{panglosstilted} still holds. For $D_1=0$, recalling Lemma \ref{propertytilted} and Remark \ref{remarkmain}, we obtain
\begin{align}
\nn&\jmath_{XY}(x,y|R_1^*,R_2^*,D_1,D_2,P_{XY})\\*
&=\log\frac{P_{W|XY}^*(w|xy)}{P_{W}^*(w)}+\lambda_1^*\left(\log\frac{1}{P_{X|W}^*(x|w)}-R_1^*\right)+\lambda_2^*\left(\jmath_{Y|W}(y,D_2|w,P_{YW}^*)-R_2^*\right).
\end{align}
The rest of the proof is similar to Appendix \ref{proofpanglosstilted} by replacing $(\hat{x},\hat{X})$ replaced by $(x,X)$.

\subsection{Proof of Lemma \ref{achievable}}
\label{proofachievable}
\subsubsection{Properties on $Q_{W|XY}$ and $\calC_n$}
We first present Lemma 4 in \cite{watanabe2015second} which is the initial step for the proof of Lemma \ref{achievable}.
\begin{lemma}
\label{wtypesize}
Suppose $n$ is sufficiently large such that $(n+1)^4>n\log |\calX|\cdot|\calY|$. Given type $Q_{XY}\in\calP_{n}(\calX\times\calY)$ and any test channel $P_{W|XY}$, there exists a conditional type $Q_{W|XY}\in\calV_{n}(\calW,Q_{XY})$ such that for every triplet $(x,y,w)$ with $Q_{XY}(x,y)P_{W|XY}(w|xy)>0$,
\begin{align}
\label{approxtype}
\left|Q_{W|XY}(w|xy)-P_{W|XY}(w|xy)\right|\leq \frac{1}{nQ_{XY}(x,y)}.
\end{align}
Let $Q_{W}$ be the marginal type of $W$ induced by $Q_{XY}$ and $Q_{W|XY}$. In addition, there exists a set $\calC_n\subset\calT_{Q_{W}}$ such that
\begin{align}
|\calC_n|\leq \exp\left(nI(Q_{XY},Q_{W|XY})+|\calX|\cdot|\calY|\cdot|\calW|\log (n+1)\right),
\end{align}
and for each $(x^n,y^n)\in\calT_{Q_{XY}}$, there exists a $w^n\in\calC_n$ so that $(x^n,y^n,w^n)\in\calT_{Q_{XYW}}$, where $Q_{XYW}$ is the joint type induced by $Q_{XY}$ and $Q_{W|XY}$.
\end{lemma}

Given type $Q_{XY}$, let $P_{W|XY}^{*}$ be an optimal test channel which achieves $\rvR_0(R_1,R_2,D_1,D_2|Q_{XY})$, i.e., $I(Q_{XY},P_{W|XY}^*)=\rvR_0(R_1,R_2,D_1,D_2|Q_{XY})$. Let $P_{XW}^{*}$ be the joint distribution induced by $Q_{XY}$ and $P_{W|XY}^{*}$. Note that $R_1\geq R_{X|W}(P_{XW}^{*},D_1)$ and $R_2\geq R_{Y|W}(P_{YW}^{*},D_2)$. Lemma \ref{wtypesize} shows that there exists a conditional type $Q_{W|XY}$ such that for every $(x,y,w)$ with $Q_{XY}(x,y)P_{W|XY}^{*}(w|xy)>0$,
\begin{align}
\left|Q_{W|XY}(w|xy)-P_{W|XY}^{*}(w|xy)\right|\leq \frac{1}{nQ_{XY}(x,y)}\label{typeoptimal},
\end{align}
and there exists a set $\calC_n\in\calT_{Q_{W}}$ with size
\begin{align}
\frac{1}{n}\log |\calC_n|\leq I(Q_{XY},Q_{W|XY})+\left(|\calX|\cdot|\calY|\cdot|\calW|+4\right)\frac{\log (n+1)}{n},
\end{align}
and with the property that for each $(x^n,y^n)$, there exists $w^n\in\calC_n$ satisfying $(x^n,y^n,w^n)\in\calT_{Q_{XYW}}$. 
Hence, we obtain from \eqref{typeoptimal} that
\begin{align}
\left\|Q_{XYW}-Q_{XY}\times P_{W|XY}^{*}\right\|_{1}\leq \frac{|\calX|\cdot|\calY|\cdot|\calW|}{n}.
\end{align}
According to Corollary 1 in \cite{watanabe2015second},
\begin{align}
\left|I\left(Q_{XY},Q_{W|XY}\right)-I\left(Q_{XY},P_{W|XY}^{*}\right)\right|
\leq \frac{2|\calX|\cdot|\calY|\cdot|\calW|\log n}{n}.
\end{align}
Hence, we conclude 
\begin{align}
\frac{1}{n}\log |\calC_n|\leq \rvR_0(R_1,R_2,D_1,D_2|Q_{XY})+\left(3|\calX|\cdot|\calY|\cdot|\calW|+4\right)\frac{\log (n+1)}{n}.
\end{align}

\subsubsection{Properties of $\calB_{\hat{X}}(w^n)$ and $\calB_{\hat{Y}}(w^n)$}
We modify the proof of \cite[Lemma 8]{no2016} and \cite[Corollary 1]{watanabe2015second} to prove the existence of $\calB_{\hat{X}}(w^n)$ and $\calB_{\hat{Y}}(w^n)$ in Lemma \ref{achievable}. We prove here only for the properties $\calB_{\hat{X}}(w^n)$ since the properties of $\calB_{\hat{Y}}(w^n)$ can be proved in a similar manner. Define $D_1 ^*=D_1-\frac{|\calX|\cdot|\calW|\cdot|\hat{\calX}|}{n}\overline{d}_{X}$. Let $Q_{XW}$ and $Q_{X|W}$ be induced by $Q_{XY}$ and $Q_{W|XY}$. Let $Q_{\hat{X}|XW}^*$ be the optimal test channel achieving $R_{X|W}(Q_{XW},D_1^*)$, i.e.,
\begin{align}
R_{X|W}(Q_{XW},D_1^*)=I(Q_{X|W},Q_{\hat{X}|XW}^*|Q_{W}),
\end{align}
and
\begin{align}
\mathbb{E}[d(X,\hat{X})]
&=\sum_{x,w,\hat{x}}Q_{XW}(x,w)Q_{\hat{X}|XW}^*(\hat{x}|xw)d_X(x,\hat{x})\leq D_1^*.
\end{align}
Following the same procedure to prove \eqref{approxtype} in Lemma \ref{wtypesize}, we can prove that there exists a conditional type $Q_{\hat{X}|XW}$ such that such that for all $(x,w,\hat{x})$,
\begin{align}
\left|Q_{\hat{X}|XW}(\hat{x}|xw)-Q_{\hat{X}|XW}^*(\hat{x}|xw)\right|\leq 
\frac{1}{nQ_{XW}(x,w)}\label{xratedis}.
\end{align}
Let $\calT_{Q_{\hat{X}|XW}}(x^n,w^n)$ be the conditional type class given $(x^n,w^n)$, i.e., $\{\hat{x}^n:(x^n,w^n,\hat{x}^n)\in\calT_{Q_{XYW}}\}$. Then the following lemma shows that $(x^n,\hat{x}^n)$ satisfies the distortion level at $D_1$.
\begin{lemma}
\label{uppdistortion}
For any $x^n\in\calT_{Q_{X|W}}(w^n)$, there exists $\hat{x}^n\in\calT_{Q_{\hat{X}|XW}}(x^n,w^n)$ such that 
\begin{align}
d_{X}(x^n,\hat{x}^n)\leq D_1.
\end{align}
\end{lemma}
The proof of Lemma \ref{uppdistortion} is similar to \cite[Lemma 17]{no2016} and is deferred to the end of this subsection. Let $Q_{\hat{X}|W}$ be induced by $Q_{\hat{X}|XW}$ and $Q_{XW}$, i.e.,
\begin{align}
Q_{\hat{X}|W}(\hat{x}|w)=\sum_{x}Q_{X|W}(x|w)Q_{\hat{X}|XW}(x|xw).
\end{align}
Define the sets $\calA(w^n)=\calT_{Q_{\hat{X}|W}}(w^n)$ and $\calA(x^n,w^n)=\calT_{Q_{\hat{X}|XW}}(x^n,w^n)$. Given $w^n$, we randomly and uniformly generate $M_1$ codewords $(Z_1,Z_2,\ldots,Z_{M_1})$ from $\calA(w^n)$ to form the codebook $\calZ^{M_1}(w^n)$. Define the set of source sequences in $\calX^n$ that are not $D_1$-covered by the codebook $\calZ^{M_1}(w^n)$ as
\begin{align}
\calU_{X}(\calZ^{M_1},w^n):=\{x^n:(x^n,w^n)\in\calT_{Q_{XW}},d_{X}(x^n,Z_i)>D_1,~\forall~i\in[1:M_1]\}.
\end{align}
Following standard arguments (e.g.,  \cite{csiszar2011information}), we now upper bound the average size of $\calU_{X}(\calZ^{M_1},w^n)$  as follows
\begin{align}  
\nn &\mathbb{E}\left[\calU_{X}(\calZ^{M_1},w^n)\right]\\
&=\sum_{x^n\in\calT_{Q_{X|W}}(w^n)}\left(1-\Pr\left(d_{X}(x^n,Z_1)\right)\right)^{M_1}\\
&=\sum_{x^n\in\calT_{Q_{X|W}}(w^n)}\left(1-\frac{|\calA(x^n,w^n)|}{|\calA(w^n)|}\right)^{M_1}\\
&\leq \sum_{x^n\in\calT_{Q_{X|W}}(w^n)}\exp\left(-\frac{|\calA(x^n,w^n)|}{|\calA(w^n)|}M_1\right)\\
&\leq \sum_{x^n\in\calT_{Q_{X|W}}(w^n)}\exp\left(-M_1(n+1)^{-|\calX|\cdot|\calW|\cdot|\hat{\calX}|}\exp\left(nH(Q_{\hat{X}|XW}|Q_{XW})-nH(Q_{\hat{X}|W}|Q_{W})\right)\right)\\
&= |\calT_{Q_{X|W}}(w^n)|\exp\left(-M_1(n+1)^{-|\calX|\cdot|\calW|\cdot|\hat{\calX}|}\exp\left(-nI(Q_{X|W},Q_{\hat{X}|XW}|Q_{W})\right)\right)\\
&\leq \exp\left(-M_1(n+1)^{-|\calX|\cdot|\calW|\cdot|\hat{\calX}|}\exp\left(-nI(Q_{X|W},Q_{\hat{X}|XW}|Q_{W})\right)+nH(Q_{X|W}|Q_{W})\right).
\end{align}
Now choose $M_1$ such that
\begin{align}
M_1\leq (n+1)^{|\calX|\cdot|\calW|\cdot|\hat{\calX}|+4}\exp\left(nI(Q_{X|W},Q_{\hat{X}|XW}|Q_{W})\right).
\end{align}
Hence, for sufficiently large $n$ such that $nH(Q_{X|W}|Q_{W})\leq n\log|\calX|<(n+1)^4$, we have
\begin{align}
\mathbb{E}\left[\calU_{X}(\calZ^{M_1},w^n)\right]<1.
\end{align}
Hence, there exists set $\calB_{\hat{X}}(w^n)\in\hat{\calX}^n$ such that
\begin{align}
\frac{1}{n}\log |\calB_{\hat{X}}(w^n)|\leq I(Q_{X|W},Q_{\hat{X}|XW}|Q_{W})+\frac{(|\calX|\cdot|\calW|\cdot|\hat{\calX}|+4)\log n}{n},
\end{align}
and for every $x^n\in\calT_{Q_{X|W}}(w^n)$, there exists $\hat{x}^n\in\calB_{\hat{X}}(w^n)$ satisfying $d_{X}(x^n,\hat{x^n})\leq D_1$. Then we bound the difference between $I(Q_{X|W},Q_{\hat{X}|XW}|Q_{W})$ and $I(Q_{X|W},Q_{\hat{X}|XW}^*|Q_{W})$.

\begin{lemma}
\label{diffmutual}
\begin{align}
\left|I(Q_{X|W},Q_{\hat{X}|XW}|Q_{W})-I(Q_{X|W},Q_{\hat{X}|XW}^*|Q_{W})\right|\leq \frac{2|\calX|\cdot|\calW|\cdot|\hat{\calX}|\log n}{n}.
\end{align}
\end{lemma}
The proof of Lemma \ref{diffmutual} is similar to \cite[Lemma 18]{no2016} and is given in Appendix \ref{proofdiffmutual}. Invoking Lemma \ref{diffmutual}, we have proved that 
\begin{align}
\frac{1}{n}\log |\calB_{\hat{X}}(w^n)|
&\leq I(Q_{X|W},Q_{\hat{X}|XW}^*|Q_{W})+\frac{(3|\calX|\cdot|\calW|\cdot|\hat{\calX}|+4)\log n}{n}\\
&= R_{X|W}(Q_{XW},D_1^*)+\frac{(3|\calX|\cdot|\calW|\cdot|\hat{\calX}|+4)\log n}{n}\label{provestep1}.
\end{align}

The next step to prove Lemma \ref{achievable} is to bound the difference between $R_{X|W}(Q_{XW},D_1^*)$ and $R_{X|W}(Q_{XW},D_1)$.
\begin{lemma}
\label{continuityd2}
For $n$ sufficiently large such that $\log n\geq \frac{|\calX|\cdot|\calW|\cdot|\hat{\calX}|\overline{d}_{X}\log|\calX|}{D_1}$,  we obtain
\begin{align}
R_{X|W}(Q_{XW},D_1^*)\leq R_{X|W}(Q_{XW},D_1)+\frac{\log n}{n}\label{provestep2}.
\end{align}
\end{lemma} 
The proof of Lemma \ref{continuityd2} is similar to \cite[Lemma 19]{no2016} and is given in the end of this subsection. 

The remaining step to prove Lemma \ref{achievable} is to bound the difference between $R_{X|W}(Q_{XW},D_1)$ and $R_{X|W}(P_{XW}^*,D_1)$. Invoking \eqref{typeoptimal}, the difference between $Q_{XW}(x,w)$ and $P_{XW}^{*}(x,w)$ is bounded as
\begin{align}
\left|Q_{XW}(x,w)-P_{XW}^{*}(x,w)\right|
&=\left|\sum_{y}Q_{XY}(x,y)\left(Q_{W|XY}(w|xy)-P_{W|XY}^{*}(w|xy)\right)\right|\\
&\leq \sum_{y}Q_{XY}(x,y)\left|\left(Q_{W|XY}(w|xy)-P_{W|XY}^{*}(w|xy)\right)\right|\\
&\leq \sum_{y}Q_{XY}(x,y)\frac{1}{nQ_{XY}(x,y)}\\*
&=\frac{|\calY|}{n}.
\end{align}
Hence,
\begin{align}
\left\|Q_{XW}-P_{XW}^{*}\right\|_{1}
=\sum_{x,w}\left|Q_{XW}(x,w)-P_{XW}^{*}(x,w)\right|\leq \frac{|\calX|\cdot|\calY|\cdot|\calW|}{n}.
\end{align}
We now present a uniform continuity lemma for the conditional rate-distortion function. This serves to bound the difference between $R_{X|W}(Q_{XW},D_1)$ and $R_{X|W}(P_{XW}^*,D_1)$.
\begin{lemma}
\label{continuityp}
Given distortion measure $d:\calX\times\hat{\calX}\to[0,\infty]$, we define $\underline{d}:=\min_{x,\hat{x}:d(x,\hat{x})>0}d(x,\hat{x})$ and $\overline{d}:=\max_{x,\hat{x}}d(x,\hat{x})$. If the distortion measure $d$ satisfies that for each $x\in\calX$, there exists $\hat{x}\in\hat{\calX}$ such that $d(x,\hat{x})=0$, then for any two joint distributions $P_{XW}$ and $Q_{XW}$,
\begin{align}
\left|R_{X|W}(P_{XW},D)-R_{X|W}(Q_{XW},D)\right|\leq 10\frac{\overline{d}}{\underline{d}}\left\|P_{XW}-Q_{XW}\right\|_{1}\log \frac{|\calX|\cdot|\calW|\cdot|\hat{\calX|}}{\left\|P_{XW}-Q_{XW}\right\|_{1}}.
\end{align}
\end{lemma}
The proof of Lemma \ref{continuityp} is modified from \cite{palaiyanur2008uniform} and given in Appendix \ref{proofcontinuityp}. Invoking Lemma \ref{continuityp}, we obtain
\begin{align}
R_{X|W}(Q_{XW},D_1)
&\leq R_{X|W}(P_{XW}^{*},D_1)+\frac{10\overline{d}_X}{\underline{d}_X}\frac{|\calX|\cdot|\calY|\cdot|\calW|}{n}\log\frac{|\calX|\cdot|\calW|\cdot|\hat{\calX|}n}{|\calX|\cdot|\calY|\cdot|\calW|}\\
&\leq R_{X|W}(P_{XW}^{*},D_1)+\frac{11\overline{d}_X}{\underline{d}_X}\frac{|\calX|\cdot|\calY|\cdot|\calW|}{n}\log n,\label{upperb1}
\end{align}
where \eqref{upperb1} holds when $\log n\geq \log |\hat{\calX}|-\log|\calY|$.
Finally, we obtain
\begin{align}
\frac{1}{n}\log |\calB_{\hat{X}}(w^n)|
&\leq R_{X|W}(Q_{XW},D_1)+\frac{(3|\calX|\cdot|\calW|\cdot|\hat{\calX}|+5)\log n}{n}\label{provestep3}\\
&\leq R_{X|W}(P_{XW}^*,D_1)+\left(\frac{11\overline{d}_X}{\underline{d}_X}|\calX|\cdot|\calY|\cdot|\calW|+3|\calX|\cdot|\calW|\cdot|\hat{\calX}|+5\right)\frac{\log n}{n}\\
&\leq R_1+\left(\frac{11\overline{d}_X}{\underline{d}_X}|\calX|\cdot|\calY|\cdot|\calW|+3|\calX|\cdot|\calW|\cdot|\hat{\calX}|+5\right)\frac{\log n}{n},
\end{align}
where \eqref{provestep3} follows from \eqref{provestep1} and \eqref{provestep2}.

We now present the proofs of Lemmas \ref{uppdistortion} and \ref{continuityd2}.
\begin{proof}[Proof of Lemma \ref{uppdistortion}]
\begin{align}
d_{X}(x^n,\hat{x}^n)
&=\frac{1}{n}\sum_{i=1}^n d_{X}(x_i,\hat{x}_i)\\*
&=\sum_{x,w,\hat{x}}Q_{XW}(x,w)Q_{\hat{X}|XW}(\hat{x}|xw)d_{X}(x,\hat{x})\\
&\leq \sum_{x,w,\hat{x}}Q_{XW}(x,w)\left(Q_{\hat{X}|XW}^*(\hat{x}|xw)+\frac{1}{nQ_{XW}(x,w)}\right)d_{X}(x,\hat{x})\\
&\leq D_1^*+\frac{|\calX|\cdot|\calW|\cdot|\hat{\calX}|}{n}\overline{d}_{X}\\
&\leq D_1.
\end{align}
\end{proof}

\begin{proof}[Proof of Lemma \ref{continuityd2}]
The conditional rate-distortion function $R_{X|W}(Q_{XW},D)$ is a convex and non-increasing function of $D$. Hence,
\begin{align}
\frac{R_{X|W}(Q_{XW},D_1^*)-R_{X|W}(Q_{XW},D_1)}{D_1-D_1^*}\leq \frac{R_{X|W}(Q_{XW},0)-R_{X|W}(Q_{XW},D_1)}{D_1}\leq 
\frac{R_{X|W}(Q_{XW},0)}{D_1}\leq \frac{\log |\calX|}{D_1}. 
\end{align}
Recalling that $D_1^*=D_1-\frac{|\calX|\cdot|\calW|\cdot|\hat{\calX}|}{n}\overline{d}_{X}$, we conclude
\begin{align}
R_{X|W}(Q_{XW},D_1^*)
&\leq R_{X|W}(Q_{XW},D_1)+\frac{\log |\calX|}{D_1}\left(D_1-D_1^*\right)\\
&=R_{X|W}(Q_{XW},D_1)+\frac{|\calX|\cdot|\calW|\cdot|\hat{\calX}|\overline{d}_{X}\log |\calX|}{nD_1}\\
&\leq R_{X|W}(Q_{XW},D_1)+\frac{\log n}{n},
\end{align}
when $\log n\geq \frac{|\calX|\cdot|\calW|\cdot|\hat{\calX}|\overline{d}_{X}\log|\calX|}{D_1}$.
\end{proof}

\subsection{Proof of Lemma \ref{diffmutual}}
\label{proofdiffmutual}
Invoking \eqref{xratedis}, the difference between $Q_{\hat{X}|W}(\hat{x}|w)$ and $Q_{\hat{X}|W}^{*}(\hat{x}|w)$ can be bounded as follows:
\begin{align}
\left|Q_{\hat{X}|W}(\hat{x}|w)-Q_{\hat{X}|W}^{*}(\hat{x}|w)\right|
&=\left|\sum_{x}Q_{X|W}(x|w)\left(Q_{\hat{X}|XW}(\hat{x}|xw)-Q_{\hat{X}|XW}^*(\hat{x}|xw)\right)\right|\\
&\leq \sum_{x}Q_{X|W}(x|w)\left|\left(Q_{\hat{X}|XW}(\hat{x}|xw)-Q_{\hat{X}|XW}^*(\hat{x}|xw)\right)\right|\\
&\leq \sum_{x}Q_{X|W}(x|w)\frac{1}{nQ_{XW}(x,w)}\\
&=\sum_{x}\frac{1}{nQ_{W}(w)}\\
&=\frac{|\calX|}{nQ_{W}(w)}.
\end{align}
Hence,
\begin{align}
\left\|Q_{\hat{X}|W=w}-Q_{\hat{X}|W=w}^{*}\right\|_{1}=\sum_{\hat{x}}\left|Q_{\hat{X}|W}(\hat{x}|w)-Q_{\hat{X}|W}^{*}(\hat{x}|w)\right|\leq \frac{|\calX|\cdot|\hat{\calX}|}{nQ_{W}(w)}.
\end{align}
Invoking Lemma 1.2.7 in \cite{csiszar2011information}, we obtain that for large enough $n$ where $\frac{|\calX|\cdot|\hat{\calX}|}{nQ_{W}(w)}<\frac{1}{2}$,
\begin{align}
\left|H(Q_{\hat{X}|W=w})-H(Q_{\hat{X}|W=w}^*)\right|\leq -\frac{|\calX|\cdot|\hat{\calX}|}{nQ_{W}(w)}\log\frac{|\calX|}{nQ_{W}(w)}.
\end{align}
Hence, the difference between $H(Q_{\hat{X}|W}|Q_W)$ and $H(Q_{\hat{X}|W}^{*}|Q_W)$ is bounded as follows:
\begin{align}
\left|H(Q_{\hat{X}|W}|Q_{W})-H(Q_{\hat{X}|W}^{*}|Q_W)\right|
&=\left|\sum_{w}Q_{W}(w)\left(H(Q_{\hat{X}|W=w})-H(Q_{\hat{X}|W=w}^*)\right)\right|\\
&\leq \sum_{w}Q_{W}(w)\left|\sum_{\hat{x}}\left(H(Q_{\hat{X}|W=w})-H(Q_{\hat{X}|W=w}^*)\right)\right|\\
&\leq \sum_{w}Q_{W}(w)\frac{|\calX|\cdot|\hat{\calX}|}{nQ_{W}(w)}\log\frac{nQ_{W}(w)}{|\calX|}\\
&\leq \frac{|\calX|\cdot|\hat{\calX}|}{n}\sum_{w}\log nQ_{W}(w)\\
&\leq \frac{|\calX|\cdot|\hat{\calX}|\cdot|\calW|\log n}{n}\label{difft1}.
\end{align}
Define $f(x)=-x\log x$. We now bound the difference between $H(Q_{\hat{X}|XW}|Q_{XW})$ and $H(Q_{\hat{X}|XW}^*|Q_{XW})$ as follows
\begin{align}
\left|H(Q_{\hat{X}|XW}|Q_{XW})-H(Q_{\hat{X}|XW}^*|Q_{XW})\right|
&=\left|\sum_{x,w}Q_{XW}(x,w)\sum_{\hat{x}}\left(f\left(Q_{\hat{X}|XW}(\hat{x}|xw)\right)-f\left(Q_{\hat{X}|XW}^*(\hat{x}|xw)\right)\right)\right|\\
&\leq\sum_{x,w}Q_{XW}(x,w)\left|\sum_{\hat{x}}\left(f\left(Q_{\hat{X}|XW}(\hat{x}|xw)\right)-f\left(Q_{\hat{X}|XW}^*(\hat{x}|xw)\right)\right)\right|\\
&\leq\sum_{x,w}Q_{XW}(x,w)\sum_{\hat{x}}f\left(\left|Q_{\hat{X}|XW}(\hat{x}|xw)-Q_{\hat{X}|XW}^*(\hat{x}|xw)\right|\right)\label{lemma2.7csiszar}\\
&\leq\sum_{x,w}Q_{XW}(x,w)\sum_{\hat{x}}f\left(\frac{1}{nQ_{XW}(x,w)}\right)\\
&=\sum_{x,w}Q_{XW}(x,w)\sum_{\hat{x}}\frac{1}{nQ_{XW}(x,w)}\log nQ_{XW}(x,w)\\
&\leq\frac{|\calX|\cdot|\hat{\calX}|\cdot|\calW|\log n}{n}\label{difft2},
\end{align}
where \eqref{lemma2.7csiszar} holds since $|f(x)-f(y)|\leq f(|x-y|)$ when $|x-y|\leq \frac{1}{2}$. See \cite[Lemma 1.2.7]{csiszar2011information}.
Finally, we can bound the difference between $I(Q_{X|W},Q_{\hat{X}|XW}|Q_{W})$ and $I(Q_{X|W},Q_{\hat{X}|XW}^*|Q_{W})$ using \eqref{difft1}, \eqref{difft2} and $I(X;\hat{X}|W)=H(\hat{X}|W)-H(\hat{X}|XW)$:
\begin{align}
\nonumber &\left|I(Q_{X|W},Q_{\hat{X}|XW}|Q_{W})-I(Q_{X|W},Q_{\hat{X}|XW}^*|Q_{W})\right|\\
&=\left|\left(H(Q_{\hat{X}|W}|Q_{W})-H(Q_{\hat{X}|W}^{*}|Q_W)\right)-\left(H(Q_{\hat{X}|XW}|Q_{XW})-H(Q_{\hat{X}|XW}^*|Q_{XW})\right)\right|\\
&\leq \left|H(Q_{\hat{X}|W}|Q_{W})-H(Q_{\hat{X}|W}^{*}|Q_W)\right|+\left|H(Q_{\hat{X}|W}|Q_{W})-H(Q_{\hat{X}|W}^{*}|Q_W)\right|\\
&\leq\frac{2|\calX|\cdot|\hat{\calX}|\cdot|\calW|\log n}{n}.
\end{align}

\subsection{Proof of Lemma \ref{continuityp}}
\label{proofcontinuityp}
The proof follows \cite{palaiyanur2008uniform} and relies on the continuity of entropy function \cite[Lemma 1.2,7]{csiszar2011information}. Suppose $P_{\hat{X}|XW}^*$ achieves $R_{X|W}(P_{XW},D)$ and $Q_{\hat{X}|XW}^*$ achieves $R_{X|W}(Q_{XW},D)$. Define the distortion function $d(P_{XY},P_{\hat{X}|XW})$ as 
\begin{align}
d(P_{XW},P_{\hat{X}|XW}):=\sum_{x,w,\hat{x}}P_{XW}(x,w)P_{\hat{X}|XW}(\hat{x}|x,w)d(x,\hat{x}).
\end{align}
Hence,  
\begin{align}
d(Q_{XW},Q_{\hat{X}|XW}^*)\leq D.
\end{align}
For a source with distribution $P_{XW}$, if we choose the test channel to be $Q_{\hat{X}|XW}^*$, then the distortion function can be bounded above as
\begin{align}
d(P_{XW},Q_{\hat{X}|XW}^*)
&\leq d(Q_{XW},Q_{\hat{X}|XW}^*)+\left|d(P_{XW},Q_{\hat{X}|XW}^*)-d(Q_{XW},Q_{\hat{X}|XW}^*)\right|\\
&\leq D+\left|\sum_{x,w,\hat{x}}\left(P_{XW}(x,w)-Q_{XW}(x,w)\right)Q_{\hat{X}|XW}^*(\hat{x}|xw)d(x,\hat{x})\right|\\
&\leq D+\|P_{XW}-Q_{XW}\|_{1}\overline{d}\label{upperboundd}.
\end{align}
Define the following (conditional) distributions:
\begin{align}
(QP)_{\hat{X}|W}^*(\hat{x}|w)&:=\sum_{x}P_{X|W}(x|w)Q_{\hat{X}|XW}^*(\hat{x}|xw),\\
(QP)_{X\hat{X}|W}^*(x,\hat{x}|w)&:=P_{X|W}(x|w)Q_{\hat{X}|XW}^*(\hat{x}|xw),\\
(QP)_{XW\hat{X}}^*(x,w,\hat{x})&=P_{XW}(x,w)Q_{\hat{X}|XW}^*(\hat{x}|xw),\\
(QP)_{\hat{X}W}^*(\hat{x},w)&=\sum_{x}P_{XW}(x,w)Q_{\hat{X}|XW}^*(\hat{x}|xw).
\end{align}
According to the definition of conditional rate-distortion function $R_{X|W}(P_{XW},D)$, we obtain 
\begin{align}
\nonumber &R_{X|W}\left(P_{XW},d(P_{XW},Q_{\hat{X}|XW}^*)\right)\\
&\leq I(P_{X|W},Q_{\hat{X}|XW}^*|P_W)\\
&\leq I(Q_{X|W},Q_{\hat{X}|XW}^*|Q_{W})+\left|I(Q_{X|W},Q_{\hat{X}|XW}^*|Q_{W})-I(P_{X|W},Q_{\hat{X}|XW}^*|P_W)\right|\\
&=R_{X|W}(Q_{XW},D)+\left|I(Q_{X|W},Q_{\hat{X}|XW}^*|Q_{W})-I(P_{X|W},Q_{\hat{X}|XW}^*|P_W)\right|\label{upperboundPQ}.
\end{align}
Noting that
\begin{align}
I(X;\hat{X}|W)
&=H(X|W)+H(\hat{X}|W)-H(X\hat{X}|W)\\
&=H(XW)+H(\hat{X}W)+H(X\hat{X}W)-3H(W),
\end{align}
we can upper bound the second term in \eqref{upperboundPQ} as follows:
\begin{align}
\eqref{upperboundPQ}
\nonumber &\leq \left|H(P_{X|W}|P_{W})-H(Q_{X|W}|Q_{W})\right|+\left|H((QP)_{\hat{X}|W}^*|P_{W})-H(Q_{\hat{X}|W}^*|Q_{W})\right|\\
&\qquad+\left|H((QP)_{X\hat{X}|W}|P_W)-H(Q_{X\hat{X}}|Q_W)\right|\\
&\nonumber \leq \left|H(P_{XW})-H(Q_{XW})\right|+\left|H(P_W)-H(Q_W)\right|+
\left|H((QP)_{\hat{X}W}^*)-H(Q_{\hat{X}W}^*)\right|+\left|H(P_W)-H(Q_W)\right|\\
&\qquad+\left|H(QP)_{XW\hat{X}}^*-H(Q_{XW\hat{X}}^*)\right|+\left|H(P_W)-H(Q_W)\right|\label{upperbound2}.
\end{align}
Considering the $L_1$ norms
\begin{align}
\left\|P_{W}-Q_{W}\right\|_{1}
&\leq\sum_{w}\left|P_{W}(w)-Q_{W}(w)\right|\\
&=\sum_{w}\left|\sum_{x}\left(P_{XW}(x,w)-Q_{XW}(x,w)\right)\right|\\
&\leq \sum_{x,w}\left|P_{XW}(x,w)-Q_{XW}(x,w)\right|\\
&=\left\|P_{XW}-Q_{XW}\right\|_{1},\\
\left\|(QP)_{\hat{X}W}^*-Q_{\hat{X}W}^*\right\|_{1}
&\leq\sum_{\hat{x}w}\left|(QP)^*_{\hat{X}W}(\hat{x},w)-Q^*_{\hat{X}W}(\hat{x},w)\right|\\
&\leq \sum_{x,w,\hat{x}}Q_{\hat{X}|XW}^*(\hat{x}|xw)\left|P_{XW}(xw)-Q_{XW}(xw)\right|\\
&=\left\|P_{XW}-Q_{XW}\right\|_{1},\\
\left\|(QP)_{XW\hat{X}}^*-Q_{XW\hat{X}}^*\right\|_{1}
&\leq\sum_{x,w,\hat{x}}\left|(QP)_{XW\hat{X}}^*(x,w,\hat{x})-Q_{XW\hat{X}}^*(Q_{XW\hat{X}}^*)\right|\\
&\leq\left\|P_{XW}-Q_{XW}\right\|_{1}.
\end{align}
Invoking Lemma 1.2.7 in \cite{csiszar2011information}, we upper bound \eqref{upperbound2} by
\begin{align}
\eqref{upperbound2}\leq 6\left\|P_{XW}-Q_{XW}\right\|_{1}\log \frac{|\calX|\cdot|\calW|\cdot|\hat{\calX|}}{\left\|P_{XW}-Q_{XW}\right\|_{1}}.
\end{align}
Hence, the difference between $R_{X|W}(P_{XW},D)$ and $R_{Y|W}(Q_{XW},D)$ is bounded as follows:
\begin{align}
\nn& R_{X|W}(P_{XW},D)-R_{X|W}(Q_{XW},D)\\
&\leq R_{X|W}(P_{XW},D)-R_{X|W}\left(P_{XW},d(P_{XW},Q_{\hat{X}|XW}^*)\right)+6\left\|P_{XW}-Q_{XW}\right\|_{1}\log \frac{|\calX|\cdot|\calW|\cdot|\hat{\calX|}}{\left\|P_{XW}-Q_{XW}\right\|_{1}}\\
&\leq R_{X|W}(P_{XW},D)-R_{X|W}\left(P_{XW},D+\|P_{XW}-Q_{XW}\|_{1}\overline{d}\right)+6\left\|P_{XW}-Q_{XW}\right\|_{1}\log \frac{|\calX|\cdot|\calW|\cdot|\hat{\calX|}}{\left\|P_{XW}-Q_{XW}\right\|_{1}},\label{upper1}
\end{align}
where \eqref{upper1} follows from \eqref{upperboundd}.

The next lemma presents the uniform continuity of the conditional rate-distortion function in distortion level $D$.
\begin{lemma}
\label{continuityD}
The conditional rate-distortion function satisfies for any $D<D'$,
\begin{align}
R_{X|W}(Q_{XW},D)\leq R_{X|W}(Q_{XW},D')+\frac{4(D'-D)}{\underline{d}}\log\frac{|\calX|\cdot|\calW|\cdot|\hat{\calX}|\underline{d}}{2(D'-D)}
\end{align}
\end{lemma}
The proof of Lemma \ref{continuityD} is modified from \cite{palaiyanur2008uniform} and is given in Appendix \ref{proofcontinuityD}. Invoking Lemma \ref{continuityD}, we obtain
\begin{align}
R_{X|W}(P_{XW},D)-R_{X|W}\left(P_{XW},D+\|P_{XW}-Q_{XW}\|_{1}\overline{d}\right)
&\leq \frac{4\overline{d}}{\underline{d}}\|P_{XW}-Q_{XW}\|_{1}\log \frac{|\calX|\cdot|\calW|\cdot|\hat{\calX}|\underline{d}}{\|P_{XW}-Q_{XW}\|_{1}\overline{d}}.
\end{align}
Therefore, we conclude
\begin{align}
\nn&R_{X|W}(P_{XW},D)-R_{X|W}(Q_{XW},D)\\
&\leq\frac{4\overline{d}}{\underline{d}}\|P_{XW}-Q_{XW}\|_{1}\log \frac{|\calX|\cdot|\calW|\cdot|\hat{\calX}|\underline{d}}{\|P_{XW}-Q_{XW}\|_{1}\overline{d}}+6\left\|P_{XW}-Q_{XW}\right\|_{1}\log \frac{|\calX|\cdot|\calW|\cdot|\hat{\calX|}}{\left\|P_{XW}-Q_{XW}\right\|_{1}}\\
&\leq 10\frac{\overline{d}}{\underline{d}}\left\|P_{XW}-Q_{XW}\right\|_{1}\log \frac{|\calX|\cdot|\calW|\cdot|\hat{\calX|}}{\left\|P_{XW}-Q_{XW}\right\|_{1}}.
\end{align}
Symmetrically, we can prove
\begin{align}
R_{X|W}(Q_{XW},D)-R_{X|W}(P_{XW},D)
\leq 10\frac{\overline{d}}{\underline{d}}\left\|P_{XW}-Q_{XW}\right\|_{1}\log \frac{|\calX|\cdot|\calW|\cdot|\hat{\calX|}}{\left\|P_{XW}-Q_{XW}\right\|_{1}}.
\end{align}
The proof of Lemma \ref{continuityp} is now complete.

\subsection{Proof of Lemma \ref{continuityD}}
\label{proofcontinuityD}
Suppose $Q_{\hat{X}|XW}^*$ achieve the conditional rate-distortion function $R_{X|W}(Q_{XW},D'-D)$. Then
\begin{align}
D'-D
&\geq d(Q_{XW},Q_{\hat{X}|XW}^*)\\
&=\sum_{x,w}Q_{XW}(x,w)\sum_{\hat{x}:d_{X}(x,\hat{x})>0}Q_{\hat{X}|XW}^*(\hat{x}|xw)d(x,\hat{x})\\
&\geq \underline{d}\sum_{x,w}Q_{XW}(x,w)\sum_{\hat{x}:d_{X}(x,\hat{x})>0}Q_{\hat{X}|XW}^*(\hat{x}|xw)\label{dlevel}.
\end{align}
We define  $V_{\hat{X}|XW}^{*}(\hat{x}|xw)$ such that
\begin{align}
V_{\hat{X}|XW}^{*}(\hat{x}|xw)=
\left\{
\begin{array}{cl}
Q_{\hat{X}|XW}^*(\hat{x}|xw)+\frac{\sum_{\hat{x}':d(x,\hat{x}')>0}Q_{\hat{X}|XW}^*(\hat{x}'|xw)}{\left|\{\hat{x}':d(x,\hat{x}')=0\}\right|}& d(x,\hat{x})=0,\\
0 &\mathrm{otherwise.}
\end{array}
\right.
\end{align}
For fixed $(x,w)$, we obtain
\begin{align}
\sum_{\hat{x}}V_{\hat{X}|XW}^{*}(\hat{x}|xw)=\sum_{\hat{x}:d(x,\hat{x})=0}Q_{\hat{X}|XW}^*(\hat{x}|xw)+\sum_{\hat{x}':d(x,\hat{x}')>0}Q_{\hat{X}|XW}^*(\hat{x}'|xw)=1,
\end{align}
and
\begin{align}
\sum_{\hat{x}}\left|V_{\hat{X}|XW}^{*}(\hat{x}|xw)-Q_{\hat{X}|XW}^*(\hat{x}|xw)\right|
&=\sum_{\hat{x}:d(x,\hat{x})>0}Q_{\hat{X}|XW}^*(\hat{x}|xw)+\sum_{\hat{x}:d(x,\hat{x})=0}\frac{\sum_{\hat{x}':d(x,\hat{x}')>0}Q_{\hat{X}|XW}^*(\hat{x}'|xw)}{\left|\{\hat{x}':d(x,\hat{x}')=0\}\right|}\\
&=2\sum_{\hat{x}:d(x,\hat{x})>0}Q_{\hat{X}|XW}^*(\hat{x}|xw).
\end{align}
Therefore, $V_{\hat{X}|XW}^{*}(\hat{x}|xw)$ is a valid conditional distribution. Invoking \eqref{dlevel}, we obtain
\begin{align}
\sum_{x,w}Q_{XW}(x,w)\sum_{\hat{x}}\left|V_{\hat{X}|XW}^{*}(\hat{x}|xw)-Q_{\hat{X}|XW}^*(\hat{x}|xw)\right|
&=2\sum_{x,w}Q_{XW}(x,w)\sum_{\hat{x}:d(x,\hat{x})>0}Q_{\hat{X}|XW}^*(\hat{x}|xw)\\
&\leq \frac{2(D'-D)}{\underline{d}}.
\end{align}
Since 
\begin{align}
d(Q_{XW},V_{\hat{X}|XW}^{*}(\hat{x}|xw))
&=\sum_{x,w}Q_{XW}(x,w)\sum_{\hat{x}}V_{\hat{X}|xW}^{*}(\hat{x}|xw)d(x,\hat{x})=0,
\end{align}
we obtain
\begin{align}
R_{X|W}(Q_{XW},0)\leq I(Q_{X|W},V_{\hat{X}|XW}^*|Q_{W}).
\end{align}

Denote
\begin{align}
(VQ)_{XW\hat{X}}^*(x,w,\hat{x})&=Q_{XW}(x,w)V_{\hat{X}|XW}^*(\hat{x}|xw),\\
(VQ)_{W\hat{X}}^*(w,\hat{x})&=\sum_{x}Q_{XW}(x,w)V_{\hat{X}|XW}^*(\hat{x}|xw),\\
Q_{W\hat{X}}^*(w,\hat{x})&=\sum_{x}Q_{XW}(x,w)Q_{XW\hat{X}}^*(\hat{x}|xw).
\end{align}

Because the conditional rate-distortion function $R_{X|W}(Q_{XW},D)$ is convex and non-increasing in $D$, hence, we obtain
\begin{align}
R_{X|W}(Q_{XW},D)-R_{X|W}(Q_{XW},D')
&\leq R_{X|W}(Q_{XW},0)-R_{X|W}(Q_{XW},D'-D)\\
&\leq I(Q_{X|W},V_{\hat{X}|XW}^*|Q_{W})-I(Q_{X|W},Q_{\hat{X}|XW}^*|Q_{W})\label{lemma2.7again}.
\end{align}
Since,
\begin{align}
\left\|(VQ)_{W\hat{X}}^*-Q_{W\hat{X}}\right\|_{1}
&\leq \sum_{w,\hat{x}}\left|\sum_{x}Q_{XW}(x,w)\left(V_{XW\hat{X}}^*(\hat{x}|xw)-Q_{XW\hat{X}}^*(\hat{x}|xw)\right)\right|\\
&\leq\sum_{x,w}Q_{XW}(x,w)\sum_{\hat{x}}\left|V_{\hat{X}|W}^{*}(\hat{x}|xw)-Q_{\hat{X}|XW}^*(\hat{x}|xw)\right|\\
&\leq \frac{2(D'-D)}{\underline{d}},\\
\left\|(VQ)_{XW\hat{X}}^*-Q_{XW\hat{X}}^*\right\|_{1}
&\leq\sum_{x,w,\hat{x}}\left|Q_{XW}(x,w)\left(V_{\hat{X}|W}^{*}(\hat{x}|xw)-Q_{\hat{X}|XW}^*(\hat{x}|xw)\right)\right|\\
&\leq \frac{2(D'-D)}{\underline{d}},
\end{align}
by noting that
\begin{align}
I(X;\hat{X}|W)=H(W\hat{X})-H(X\hat{X}W)-H(W)+H(XW),
\end{align}
we can upper bound \eqref{lemma2.7again} using Lemma 2.7 in \cite{csiszar2011information} as
\begin{align}
\eqref{lemma2.7again}
&\leq 
\left|H((VQ)_{W\hat{X}}^*)-H(Q_{W\hat{X}}^*)\right|+
\left|H((VQ)_{XW\hat{X}}^*)-H(Q_{XW\hat{X}}^*)\right|\\
&\leq \frac{2(D'-D)}{\underline{d}}\log\frac{|\calW|\cdot|\hat{\calX}|}{\frac{2(D'-D)}{\underline{d}}}+\frac{2(D'-D)}{\underline{d}}\log\frac{|\calX|\cdot|\calW|\cdot|\hat{\calX}|}{\frac{2(D'-D)}{\underline{d}}}\\
&\leq \frac{4(D'-D)}{\underline{d}}\log\frac{|\calX|\cdot|\calW|\cdot|\hat{\calX}|\underline{d}}{2(D'-D)}.
\end{align}
The proof of Lemma \ref{continuityD} is now complete.

\subsection{Proof of Lemma \ref{upperboundexcessp}}
\label{proofupperboundexcessp}
Set $(R_0,R_1,R_2)=(R_{0,n},R_{1,n},R_{2,n})$. We achieve the excess-distortion probability by considering the following coding scheme. Given $(x^n,y^n)$, the encoder $0$ first calculates the joint type $Q_{XY}$ and transmits the type using at most $|\calX|\cdot|\calY|\log(n+1)$ bits. Then, encoder $0$ finds the optimal test channel $P_{W|XY}^*(Q_{XY})$ which achieves $\rvR_0(R_{0,n},R_{1,n},D_1,D_2|Q_{XY})$ and chooses a conditional distribution $Q_{W|XY}$ satisfying Lemma \ref{achievable} for test channel $P_{W|XY}^*(Q_{XY})$. If there is no such optimal test channel $P_{W|XY}^*(Q_{XY})$, the system declares an error. Given conditional type $Q_{W|XY}$, the encoder $0$ choose a set $\calC_n\subset\calT_{Q_{W}}$ satisfying properties in Lemma \ref{achievable} with $(R_1,R_2)$ replaced by $(R_{1,n},R_{2,n})$. Note that the choice of $\calC_n$ and mapping function $\phi_0$ are known by all the encoders and decoders. If $\log |\calC_n|>\log M_0-|\calX|\cdot|\calY|\log(n+1)$, the encoder $0$ declares an error directly, otherwise sends $w^n\in\calC_n$. From Lemma \ref{achievable}, we know that $(x^n,y^n,w^n)\in\calT_{Q_{XYW}}$ where $Q_{XYW}$ is induced by $Q_{XY}$ and $Q_{W|XY}$. Hence, we obtain $x^n\in\calT_{Q_{X|W}}(w^n)$ and $y^n\in\calT_{Q_{Y|W}}(w^n)$. According to Lemma \ref{achievable}, there exist sets $\calB_{\hat{X}}(w^n)$ and $\calB_{\hat{Y}}(w^n)$ satisfying that for each $(x^n,y^n)$, there exist $\hat{x}^n\in\calB_{\hat{X}}(w^n)$ and $\hat{y}^n\in\calB_{\hat{Y}}(w^n)$ with $d_{X}(x^n,\hat{x}^n)\leq D_1$ and $d_{Y}(y^n,\hat{y}^n)\leq D_2$. Therefore, encoder $1$ sends $\hat{x}^n\in\calB_{\hat{X}}(w^n)$ such that $\hat{x}^n$ minimizes $d_{X}(x^n,\hat{x}^n)$ and encoder $2$ sends $\hat{y}^n\in\calB_{\hat{Y}}(w^n)$ such that $\hat{y}^n$ minimizes $d_{Y}(y^n,\hat{y}^n)$. Invoking Lemma \ref{achievable}, the size of $\calB_{\hat{X}}(w^n)$ and $\calB_{\hat{Y}}(w^n)$ are upper bounded by
\begin{align}
\log|\calB_{\hat{X}}(w^n)|\leq nR_{1,n}+c_1\log n\leq \log M_1,\\
\log|\calB_{\hat{Y}}(w^n)|\leq nR_{2,n}+c_2\log n\leq \log M_2.
\end{align}

Finally, at the decoder side, if $w^n$ is decoded correctly, then both decoders can decode within the distortion threshold. The excess-distortion event occurs only if $\log |\calC_n|>\log M-|\calX|\cdot|\calY|\log(n+1)$ or $\rvR_0(R_{1,n},R_{2,n},D_1,D_2|Q_{XY})$ is not achieved by any test channel, which means $\rvR_0(R_{1,n},R_{2,n},D_1,D_2|Q_{XY})=\infty$. Hence, according to Lemma \ref{achievable}, the excess-distortion probability is upper bounded as in \eqref{eqn:ed_bd}.

\subsection{Proof of Lemma \ref{typestrongconverse}}
\label{prooftypeconverse}
The proof follows \cite{watanabe2015second} and uses the perturbation approach in \cite{wei2009strong}. Define the set
\begin{align}
\calD_{Q_{XY}}:=\{(x^n,y^n)\in\calT_{Q_{XY}}:d_{X}(x^n,\hat{x}^n)\leq D_1, d_{Y}(y^n,\hat{y}^n)\leq D_2\},
\end{align}
where $\hat{x}^n=\phi_1(f_0(x^n,y^n),f_1(x^n,y^n))$ and $\hat{y}^n=\phi_2(f_0(x^n,y^n),f_2(x^n,y^n))$. Recall that in Lemma \ref{typestrongconverse}, $U_{\calT_{Q_{XY}}}$ denotes the uniform distribution over type class $\calT_{Q_{XY}}$. Let $\beta=\frac{\log n}{n}$. Define the distribution $Q_{\calT_{Q_{XY}}}(x^n,y^n)$ such that
\begin{align}
Q_{\calT_{Q_{XY}}}(x^n,y^n):=\frac{\exp(n(\alpha+\beta))U_{\calT_{Q_{XY}}}(x^n,y^n)}{\exp(n(\alpha+\beta))U_{\calT_{Q_{XY}}}(D_{Q_{XY}})+(1-U_{\calT_{Q_{XY}}}(D_{Q_{XY}}))}
\end{align}
for $(x^n,y^n)\in\calD_{Q_{XY}}$ and
\begin{align}
Q_{\calT_{Q_{XY}}}(x^n,y^n):=\frac{U_{\calT_{Q_{XY}}}(x^n,y^n)}{\exp(n(\alpha+\beta))U_{\calT_{Q_{XY}}}(D_{Q_{XY}})+(1-U_{\calT_{Q_{XY}}}(D_{Q_{XY}}))}
\end{align}
for $(x^n,y^n)\notin\calD_{Q_{XY}}$.
Hence, invoking \eqref{assumption}, we obtain
\begin{align}
U_{\calT_{Q_{XY}}}(D_{Q_{XY}})\geq \exp(-n\alpha).
\end{align}
Thus, we obtain
\begin{align}
Q_{\calT_{Q_{XY}}}(D_{Q_{XY}})
&= \frac{\exp(n(\alpha+\beta))U_{\calT_{Q_{XY}}}(D_{Q_{XY}})}{\exp(n(\alpha+\beta))U_{\calT_{Q_{XY}}}(D_{Q_{XY}})+(1-U_{\calT_{Q_{XY}}}(D_{Q_{XY}}))}\\
&=\frac{\exp(n\beta)}{\exp(n\beta)+\exp(-n\alpha)\frac{1-U_{\calT_{Q_{XY}}}(D_{Q_{XY}})}{U_{\calT_{Q_{XY}}}(D_{Q_{XY}})}}\\
&\geq \frac{\exp(n\beta)}{\exp(n\beta)+1}\\
&\geq 1-\exp(-n\beta)=1-\frac{1}{n}\label{usebeta}
\end{align}
where \eqref{usebeta} results from that $\beta=\frac{\log n}{n}$. Hence, for source distribution $Q_{\calT_{Q_{XY}}}$, the excess distortion probability is upper bounded as
\begin{align}
Q_{\calT_{Q_{XY}}}(D^\mathrm{c}_{Q_{XY}})\leq \frac{1}{n}.
\end{align} 
Therefore, under source distribution $Q_{\calT_{Q_{XY}}}$, we have
\begin{align}
\mathbb{E}[d_{X}(X^n,\hat{X}^n)]
&=\sum_{(x^n,y^n)\in\calT_{Q_{XY}}} Q_{\calT_{Q_{XY}}}(x^n,y^n)d_X(x^n,\hat{x}^n)\\
&=\sum_{(x^n,y^n)\in\calD_{Q_{XY}}}Q_{\calT_{Q_{XY}}}(x^n,y^n)d_X(x^n,\hat{x}^n)+\sum_{(x^n,y^n)\notin\calD_{Q_{XY}}}Q_{\calT_{Q_{XY}}}(x^n,y^n)d_X(x^n,\hat{x}^n)\\
&\leq D_1+\frac{\overline{d}_X}{n}:=D_{1,n}\label{defd1n}.
\end{align}
and
\begin{align}
\mathbb{E}[d_{Y}(Y^n,\hat{Y}^n)]\leq D_2+\frac{\overline{d}_Y}{n}:=D_{2,n}\label{defd2n}.
\end{align}
Then we can follow the weak converse argument to lower bound $(M_0,M_1,M_2)$ as follows. Denote the encoded message from encoder $i$ as random variable $S_i$, i.e., $S_i=\phi_i(X^n,Y^n)$. In the following, all information densities are according to $Q_{\calT_{Q_{XY}}}$. Following similar steps as (134)--(142) in \cite{watanabe2015second}, we obtain
\begin{align}
\log M_0
&\geq n\left(I(X_J,Y_J;J,W_J)-\left(H(X_JY_J)-\frac{1}{n}H(X^n|Y^n)\right)\right),
\end{align}
where $W_i=(S_0,X^{i-1},Y^{i-1})$ and $J$ is uniformly distributed on $\{1,\ldots,n\}$ which is independent of all other random variables.
Modifying the ``weak converse'' proof in \cite{gray1974source} in a similar manner as \cite{watanabe2015second}, we obtain
\begin{align}
\log M_1
&\geq H(\hat{X}^n|S_0) \\
&\geq I(X^n;\hat{X}^n|S_0)\label{func01}\\
&=\sum_{i=1}^n I(X_i;\hat{X}^n|S_0,X^{i-1})\\
&\geq \sum_{i=1}^n I(X_i;\hat{X}_i|S_0,X^{i-1})\\
&=\sum_{i=1}^n I(X_i;\hat{X}_i|V_i)\\
&\geq \sum_{i=1}^n R_{X_i|V_i}(Q_{X_iV_i},D_{1,n}) \label{eqn:use_def_rd}\\
&= \sum_{i=1}^n R_{X_i|W_i}(Q_{X_iW_i},D_{1,n}) \\ &=nR_{X_J|W_J,J}(Q_{X_JW_J},D_{1,n})\label{finalm1},
\end{align}
where $V_i=(S_0,X^{i-1})$, \eqref{func01} holds since $\hat{X}^n$ is a function of $S_0$ and $S_1$ and given $S_0=s_0$, $H(\hat{X}^n|S_0=s_0)\leq \log M_1$ \cite{gray1974source} and \eqref{eqn:use_def_rd} follows from the definition of the conditional rate-distortion function (that it is a minimization of the conditional mutual information).  
Similarly, we obtain
\begin{align}
\log M_2\geq nR_{Y_J|W_J,J}(Q_{Y_JW_J},D_{2,n}).
\end{align}
According to \cite{wei2009strong}, there exists  $P_{W|X_JY_J}$ such that $|\calW|\leq |\calX|\cdot|\calY|+2$ and
\begin{align}
I(X_J,Y_J;J,W_J)&\geq I(X_J,Y_J;W),\\*
R_{X_J|W_J,J}(P_{X_JW_J},D_{1,n})&\geq R_{X_J|W}(P_{X_JW},D_{1,n})\label{finalm2},\\*
R_{Y_J|W_J,J}(P_{Y_JW_J},D_{2,n})&\geq R_{Y_J|W}(P_{Y_JW},D_{2,n}).
\end{align}
Following equations (155)-(167) in \cite{watanabe2015second}, we obtain that $P_{X_JY_J}(a,b)=Q_{XY}(a,b)$, and there exists $Q_{W|XY}$ such that $|\calW|\leq |\calX|\cdot|\calY|+2$ and
\begin{align}
I(X_J,Y_J;W)&=I(X,Y;W),\\*
R_{X_J|W}(P_{X_JW},D_{1,n})&=R_{X|W}(Q_{XW},D_{1,n})\label{finalm3},\\*
R_{Y_J|W}(P_{Y_JW},D_{2,n})&=R_{Y|W}(Q_{YW},D_{2,n}),
\end{align}
and
\begin{align}
\left|H(X_JY_J)-\frac{1}{n}H(X^n|Y^n)\right|
&\leq \frac{|\calX|\cdot|\calY|\log (n+1)}{n}+(\alpha+\beta),\\
&\leq \frac{\Big(|\calX|\cdot|\calY|+1\Big)\log (n+1)}{n}+\alpha\label{usebeta2},
\end{align}
where \eqref{usebeta2} follows from that $\beta=\frac{\log n}{n}<\frac{\log (n+1)}{n}$.

Invoking \eqref{defd1n}, \eqref{defd2n} and Lemma \ref{continuityd2}, we obtain
\begin{align}
R_{X|W}(Q_{XW},D_1)-R_{X|W}(Q_{XW},D_{1,n})\leq \frac{\log n}{n}\label{finalm4}
\end{align}
when $\log n\geq \overline{d}_X\log|\calX|$, and
\begin{align}
R_{Y|W}(Q_{YW},D_2)-R_{X|W}(Q_{YW},D_{2,n})\leq \frac{\log n}{n}
\end{align}
when $\log n\geq \overline{d}_Y\log|\calY|$.

Therefore, invoking \eqref{finalm1}, \eqref{finalm2}, \eqref{finalm3} and \eqref{finalm4}, we conclude that
\begin{align}
\log M_1\geq nR_{X|W}(Q_{XW},D_1)-\frac{\log n}{n}
\end{align}
when  $\log n\geq\overline{d}_X\log|\calX|$,
and similarly, we obtain
\begin{align}
\log M_2\geq nR_{X|W}(Q_{YW},D_2)-\frac{\log n}{n}
\end{align}
when  $\log n\geq\overline{d}_Y\log|\calY|$.

The proof of Lemma \ref{typestrongconverse} is now complete.

\subsection{Proof of Lemma \ref{lowerboundexcessp}}
\label{prooflbexcess} 

Invoking Lemma \ref{typestrongconverse} by setting $\alpha=\frac{\log n}{n}$, we obtain that if $(R_{0,n},R_{1,n},R_{2,n})\notin \calR(D_1,D_2|Q_{XY})$ ,
\begin{align}
\Pr\left(d_{X}(X^n,\hat{X}^n)\leq D_1, d_{Y}(Y^n,\hat{Y}^n)\leq D_2|(X^n,Y^n)\in\calT_{Q_{XY}}\right)<\exp(-\log n)=\frac{1}{n}\label{lemmaeq1},
\end{align}
or equivalently
\begin{align}
\Pr\left(d_{X}(X^n,\hat{X}^n)>D_1,~\mathrm{or}~d_{Y}(Y^n,\hat{Y}^n)>D_2|(X^n,Y^n)\in\calT_{Q_{XY}}\right)\geq 1-\frac{1}{n}\label{lemmaeq2}.
\end{align}
Note that when $(X^n,Y^n)\in\calT_{Q_{XY}}$, each sequence has the same probability. Hence, the probability in \eqref{lemmaeq1} and \eqref{lemmaeq2} is calculated with respect to uniform distribution over the type class $\calT_{Q_{XY}}$. Hence,
\begin{align}
\epsilon_n(D_1,D_2)
&=\sum_{(x^n,y^n)}P_{XY}^n(x^n,y^n)1\left\{d_{X}(x^n,\hat{x}^n)>D_1~\mathrm{or}~d_{Y}(y^n,\hat{y}^n)>D_2\right\}\\
&=\sum_{Q_{XY}\in\calP_n(\calX\times\calY)}\sum_{(x^n,y^n)\in\calT_{Q_{XY}}}P_{XY}^n(x^n,y^n)1\left\{d_{X}(x^n,\hat{x}^n)>D_1~\mathrm{or}~d_{Y}(y^n,\hat{y}^n)>D_2\right\}\\
&=\sum_{Q_{XY}\in\calP_n(\calX\times\calY)}P_{XY}^n(\calT_{Q_{XY}})\Pr\left(d_{X}(X^n,\hat{X}^n)>D_1,~\mathrm{or}~d_{Y}(Y^n,\hat{Y}^n)> D_2|(X^n,Y^n)\in\calT_{Q_{XY}}\right)\\
&\geq \sum_{\substack{Q_{XY}\in\calP_n(\calX\times\calY):\\(R_{0,n},R_{1,n},R_{2,n})\notin \calR(D_1,D_2|Q_{XY})}}\!\!\!\!\!\!\!\!\!\!\!\!\!\!\!\!\!\!\!\!\!\!\!\!P_{XY}^n(\calT_{Q_{XY}})\Pr\left(d_{X}(X^n,\hat{X}^n)>D_1,~\mathrm{or}~d_{Y}(Y^n,\hat{Y}^n)>D_2|(X^n,Y^n)\in\calT_{Q_{XY}}\right)\\
&\geq \sum_{\substack{Q_{XY}\in\calP_n(\calX\times\calY):\\(R_{0,n},R_{1,n},R_{2,n})\notin \calR(D_1,D_2|Q_{XY})}}P_{XY}^n(\calT_{Q_{XY}})\left(1-\frac{1}{n}\right)\\
&\geq \sum_{\substack{Q_{XY}\in\calP_n(\calX\times\calY):\\(R_{0,n},R_{1,n},R_{2,n})\notin \calR(D_1,D_2|Q_{XY})}}P_{XY}^n(\calT_{Q_{XY}})-\frac{1}{n}\\
&=\Pr\left((R_{0,n},R_{1,n},R_{2,n})\notin \calR(D_1,D_2|\hat{T}_{X^nY^n})\right)-\frac{1}{n}\\
&=\Pr\left(R_{0,n}< \rvR_0(R_{1,n},R_{2,n},D_1,D_2|\hat{T}_{X^nY^n})\right)-\frac{1}{n}.
\end{align}

\subsection*{Acknowledgements}
The authors would like to acknowledge very helpful discussions with Prof.\ Shun Watanabe. 

The authors are supported by a  Ministry of Education (MOE) Tier 2 grant (R-263-000-B61-112).

\bibliographystyle{IEEEtran}
\bibliography{IEEEfull_lin}

\begin{thebibliography}{10}
\providecommand{\url}[1]{#1}
\csname url@samestyle\endcsname
\providecommand{\newblock}{\relax}
\providecommand{\bibinfo}[2]{#2}
\providecommand{\BIBentrySTDinterwordspacing}{\spaceskip=0pt\relax}
\providecommand{\BIBentryALTinterwordstretchfactor}{4}
\providecommand{\BIBentryALTinterwordspacing}{\spaceskip=\fontdimen2\font plus
\BIBentryALTinterwordstretchfactor\fontdimen3\font minus
  \fontdimen4\font\relax}
\providecommand{\BIBforeignlanguage}[2]{{%
\expandafter\ifx\csname l@#1\endcsname\relax
\typeout{** WARNING: IEEEtran.bst: No hyphenation pattern has been}%
\typeout{** loaded for the language `#1'. Using the pattern for}%
\typeout{** the default language instead.}%
\else
\language=\csname l@#1\endcsname
\fi
#2}}
\providecommand{\BIBdecl}{\relax}
\BIBdecl

\bibitem{zhou2016}
L.~Zhou, V.~Y.~F. Tan, and M.~Motani, ``Second-order coding region for the
  discrete lossy gray-wyner source coding problem,'' in \emph{IEEE ISIT}, 2016,
  pp. 2409--2413.

\bibitem{gray1974source}
R.~Gray and A.~Wyner, ``Source coding for a simple network,'' \emph{Bell System
  Technical Journal}, vol.~53, no.~9, pp. 1681--1721, 1974.

\bibitem{watanabe2015second}
S.~Watanabe, ``Second-order region for {G}ray-{W}yner network,'' \emph{IEEE
  Trans. Inf. Theory}, 2017.

\bibitem{wei2009strong}
W.~Gu and M.~Effros, ``A strong converse for a collection of network source
  coding problems,'' in \emph{IEEE ISIT}, 2009, pp. 2316--2320.

\bibitem{viswanatha2014}
K.~B. Viswanatha, E.~Akyol, and K.~Rose, ``The lossy common information of
  correlated sources,'' \emph{IEEE Trans. Inf. Theory}, vol.~60, no.~6, pp.
  3238--3253, 2014.

\bibitem{xu2015}
G.~Xu, W.~Liu, and B.~Chen, ``A lossy source coding interpretation of wyner's
  common information,'' \emph{IEEE Trans. Inf. Theory}, vol.~62, no.~2, pp.
  754--768, 2016.

\bibitem{ingber2011dispersion}
A.~Ingber and Y.~Kochman, ``The dispersion of lossy source coding,'' in
  \emph{IEEE DCC}.\hskip 1em plus 0.5em minus 0.4em\relax IEEE, 2011, pp.
  53--62.

\bibitem{kostina2012fixed}
V.~Kostina and S.~Verd{\'u}, ``Fixed-length lossy compression in the finite
  blocklength regime,'' \emph{IEEE Trans. Inf. Theory}, vol.~58, no.~6, pp.
  3309--3338, 2012.

\bibitem{watanabe2015}
S.~Watanabe, S.~Kuzuoka, and V.~Y.~F. Tan, ``Nonasymptotic and second-order
  achievability bounds for coding with side-information,'' \emph{IEEE Trans.
  Inf. Theory}, vol.~61, no.~4, pp. 1574--1605, 2015.

\bibitem{yassaee2013technique}
M.~H. Yassaee, M.~R. Aref, and A.~Gohari, ``A technique for deriving one-shot
  achievability results in network information theory,'' in \emph{IEEE ISIT},
  2013, pp. 1287--1291.

\bibitem{no2016}
A.~No, A.~Ingber, and T.~Weissman, ``Strong successive refinability and
  rate-distortion-complexity tradeoff,'' \emph{IEEE Trans. Inf. Theory},
  vol.~62, no.~6, pp. 3618--3635, 2016.

\bibitem{Marton74}
K.~Marton, ``Error exponent for source coding with a fidelity criterion,''
  \emph{IEEE Trans. Inf. Theory}, vol.~20, no.~2, pp. 197--199, 1974.

\bibitem{ihara2000error}
S.~Ihara and M.~Kubo, ``Error exponent for coding of memoryless {Gaussian}
  sources with a fidelity criterion,'' \emph{IEICE Trans. Fundamentals},
  vol.~83, no.~10, pp. 1891--1897, 2000.

\bibitem{kanlis1996error}
A.~Kanlis and P.~Narayan, ``Error exponents for successive refinement by
  partitioning,'' \emph{IEEE Trans. Inf. Theory}, vol.~42, no.~1, pp. 275--282,
  1996.

\bibitem{tuncel2003}
E.~Tuncel and K.~Rose, ``Error exponents in scalable source coding,''
  \emph{IEEE Trans. Inf. Theory}, vol.~49, no.~1, pp. 289--296, 2003.

\bibitem{chen2007redundancy}
J.~Chen, D.-K. He, A.~Jagmohan, and L.~A. Lastras-Montano, ``On the
  redundancy-error tradeoff in {Slepian-Wolf} coding and channel coding,'' in
  \emph{IEEE ISIT}, 2007, pp. 1326--1330.

\bibitem{he2009redundancy}
D.-K. He, L.~A. Lastras-Monta{\v{n}}o, E.-H. Yang, A.~Jagmohan, and J.~Chen,
  ``On the redundancy of {Slepian--Wolf} coding,'' \emph{IEEE Trans. Inf.
  Theory}, vol.~55, no.~12, pp. 5607--5627, 2009.

\bibitem{altugwagner2014}
Y.~Altu\u{g} and A.~B. Wagner, ``Moderate deviations in channel coding,''
  \emph{IEEE Trans. Inf. Theory}, vol.~60, no.~8, pp. 4417--4426, 2014.

\bibitem{polyanskiy2010channel}
Y.~Polyanskiy and S.~Verd\'u, ``Channel dispersion and moderate deviations
  limits for memoryless channels,'' in \emph{Proc. 48th Annu. Allerton Conf.},
  2010, pp. 1334--1339.

\bibitem{altug2010moderate}
Y.~Altu\u{g} and A.~B. Wagner, ``Moderate deviation analysis of channel coding:
  Discrete memoryless case,'' in \emph{IEEE ISIT}, 2010, pp. 265--269.

\bibitem{altug2013lossless}
Y.~Altu\u{g}, A.~B. Wagner, and I.~Kontoyiannis, ``Lossless compression with
  moderate error probability,'' in \emph{IEEE ISIT}, 2013, pp. 1744--1748.

\bibitem{tan2012moderate}
V.~Y.~F. Tan, ``Moderate-deviations of lossy source coding for discrete and
  gaussian sources,'' in \emph{IEEE ISIT}, 2012, pp. 920--924.

\bibitem{borade2008}
S.~Borade and L.~Zheng, ``Euclidean information theory,'' in \emph{IEEE IZS},
  2008, pp. 14--17.

\bibitem{palaiyanur2008uniform}
H.~Palaiyanur and A.~Sahai, ``On the uniform continuity of the rate-distortion
  function,'' in \emph{IEEE ISIT}, 2008, pp. 857--861.

\bibitem{el2011network}
A.~El~Gamal and Y.-H. Kim, \emph{Network Information Theory}.\hskip 1em plus
  0.5em minus 0.4em\relax Cambridge University Press, 2011.

\bibitem{csiszar2011information}
I.~Csisz\'ar and J.~K{\"o}rner, \emph{Information Theory: Coding Theorems for
  Discrete Memoryless Systems}.\hskip 1em plus 0.5em minus 0.4em\relax
  Cambridge University Press, 2011.

\bibitem{haroutunian68}
E.~Haroutunian, ``Estimates of the error exponent for the semi-continuous
  memoryless channel,'' \emph{Problemy Pereda\v{c}i Informacii}, vol.~4, pp.
  37--48, 1968.

\bibitem{dembo2009large}
A.~Dembo and O.~Zeitouni, \emph{Large Deviations Techniques and
  Applications}.\hskip 1em plus 0.5em minus 0.4em\relax Springer Science \&
  Business Media, 2009, vol.~38.

\bibitem{TanBook}
V.~Y.~F. Tan, ``Asymptotic estimates in information theory with non-vanishing
  error probabilities,'' \emph{{Foundations and Trends$\,$\textregistered $ $
  in Communications and Information Theory}}, vol.~11, no. 1--2, pp. 1--184,
  2014.

\bibitem{kostina2012converse}
V.~Kostina and S.~Verd{\'u}, ``A new converse in rate-distortion theory,'' in
  \emph{CISS}, 2012, pp. 1--6.

\bibitem{kontoyiannis2000pointwise}
I.~Kontoyiannis, ``Pointwise redundancy in lossy data compression and universal
  lossy data compression,'' \emph{IEEE Trans. Inf. Theory}, vol.~46, no.~1, pp.
  136--152, 2000.

\bibitem{kostina2013lossy}
V.~Kostina, ``Lossy data compression: Non-asymptotic fundamental limits,''
  Ph.D. dissertation, Department of Electrical Engineering, Princeton
  University, 2013.

\bibitem{csiszar1974}
I.~Csisz\'ar, ``On an extremum problem of information theory,'' \emph{Studia
  Scientiarum Mathematicarum Hungarica}, vol.~9, no.~1, pp. 57--72, 1974.

\bibitem{berger1971rate}
T.~Berger, \emph{Rate-Distortion Theory}.\hskip 1em plus 0.5em minus
  0.4em\relax Wiley Online Library, 1971.

\bibitem{tan2014state}
M.~Tomamichel and V.~Y.~F. Tan, ``Second-order coding rates for channels with
  state,'' \emph{IEEE Trans. Inf. Theory}, vol.~60, no.~8, pp. 4427--4448,
  2014.

\bibitem{cover2012elements}
T.~M. Cover and J.~A. Thomas, \emph{Elements of information theory}.\hskip 1em
  plus 0.5em minus 0.4em\relax John Wiley \& Sons, 2012.

\bibitem{weissman2003inequalities}
T.~Weissman, E.~Ordentlich, G.~Seroussi, S.~Verdu, and M.~J. Weinberger,
  ``Inequalities for the {$L_1$} deviation of the empirical distribution,''
  Inf. Theory Res. Group, Hewlett-Packard Labs, Palo Alto, CA, USA, Tech. Rep.
  Tech. Rep. HPL-2003-97, 2003.

\bibitem{altug14a}
Y.~Altu\u{g} and A.~B. Wagner, ``Refinement of the sphere-packing bound,''
  \emph{IEEE Trans. Inf. Theory}, vol.~60, no.~3, pp. 1592--1615, 2014.

\bibitem{altug14b}
------, ``Refinement of the random coding bound,'' \emph{IEEE Trans. Inf.
  Theory}, vol.~60, no.~10, pp. 6005--6023, 2014.

\end{thebibliography}

\end{document}